\documentclass{lmcs} 


\usepackage{tikz}
\usetikzlibrary{positioning,arrows.meta}
\usepackage{hyperref}
\usepackage{mathbbol}
\usepackage{mathpartir}
\usepackage{thm-restate}
\usepackage{xspace}
\usepackage[all]{xy}

\newcommand{\sep}{\hspace{.3cm}}
\newcommand{\bsep}{\hspace{1cm}}
\newcommand{\HS}{\hspace{.5cm}}

\newcommand{\ie}{\emph{i.e.}\xspace}
\newcommand{\defn}[1]{\textbf{#1}}
\newcommand{\ih}{{\it i.h.}\xspace}

\newcommand{\CBN}{\textup{\footnotesize{CBN}}\xspace}
\newcommand{\CBNeed}{\textup{\footnotesize{CBNeed}}\xspace}
\newcommand{\CBV}{\textup{\footnotesize{CBV}}\xspace}
\newcommand{\LSC}{\textup{\footnotesize{LSC}}\xspace}
\newcommand{\VSC}{\textup{\footnotesize{VSC}}\xspace}
\newcommand{\VSCne}{\textup{\footnotesize{VSC}}$_{\textup{ne}}$\xspace}
\newcommand{\UOCBV}{\textup{{\footnotesize UCBV}}$^\Usym$\xspace}

\newcommand{\eqdef}{:=}
\newcommand{\eqgram}{\mathrel{::=}}

\newcommand{\set}[1]{\{#1\}}
\newcommand{\nat}{\mathbb{N}}
\newcommand{\length}[1]{|#1|}
\renewcommand{\implies}{\Rightarrow}
\newcommand{\disj}{\mathrel{\#}}
\newcommand{\iI}{{i \in I}}
\newcommand{\jJ}{{j \in J}}
\newcommand{\kK}{{k \in K}}

\newcommand{\relation}{\mathcal{R}}
\newcommand{\rewrite}[1]{\rightarrow_{#1}}
\newcommand{\rewriten}[1]{\rightarrow^*_{#1}}
\newcommand{\rewritensteps}[2]{\rightarrow^{#2}_{#1}}

\newcommand{\dbsym}{\mathsf{db}}
\newcommand{\lssym}{\mathsf{ls}}
\newcommand{\lsvsym}{\mathsf{lsv}}

\newcommand{\svsym}{\mathsf{sv}}
\newcommand{\svnesym}{\mathsf{svne}}

\newcommand{\tols}{\to_{\lssym}}
\newcommand{\todb}{\to_{\dbsym}}
\newcommand{\tosv}{\to_{\svsym}}
\newcommand{\tosvne}{\to_{\svnesym}}

\newcommand{\fv}[1]{\mathsf{fv}(#1)}
\newcommand{\rv}[1]{\mathsf{rv}(#1)}

\newcommand{\var}{x}
\newcommand{\vartwo}{y}
\newcommand{\varthree}{z}
\newcommand{\varfour}{w}

\newcommand{\tm}{t}
\newcommand{\tmtwo}{u}
\newcommand{\tmthree}{s}

\newcommand{\lam}[2]{\lambda#1.\,#2}
\newcommand{\sub}[2]{\{#1:=#2\}}
\newcommand{\esub}[2]{[#1\backslash#2]}

\newcommand{\id}{\mathtt{I}}

\newcommand{\val}{v}
\newcommand{\valtwo}{w}

\newcommand{\ctxhole}{\Diamond}
\newcommand{\ctxof}[2]{#1\langle#2\rangle}
\newcommand{\gctx}{\mathtt{C}}
\newcommand{\sctx}{\mathtt{L}}
\newcommand{\wctx}{\mathtt{W}}
\newcommand{\sctxtwo}{\mathtt{L}'}

\newcommand{\toVSCne}{\to_{\textup{\mbox{vsc}}_{\textup{\mbox{ne}}}}} 
\newcommand{\tovoriginal}[4]{\xrightarrow{\raisebox{-0.75ex}[0ex][0ex]{\footnotesize $\Usym$}}_{#1,#2,#3,#4}}
\newcommand{\tov}[4]{\to_{#1,#2,#3,#4}}             									
\newcommand{\tovinv}[4]{\leftarrow_{#1,#2,#3,#4}}
\newcommand{\tovTopb}[2]{\to_{\mathtt{prim},#1,#2}} 									

\newcommand{\absPred}[1]{\mathsf{abs(}#1\mathsf{)}}
\newcommand{\appPred}[1]{\mathsf{app(}#1\mathsf{)}}
\newcommand{\varPred}[1]{\mathsf{var(}#1\mathsf{)}}

\newcommand{\novscne}{\mathtt{no}}
\newcommand{\nevscne}{\mathtt{ne}}
\newcommand{\vvscne}{\mathtt{vr}}
\newcommand{\absvscne}{\mathtt{ab}}

\newcommand{\Usym}{\bullet}
\newcommand{\inv}[3]{\mathsf{inv}{(#1,#2,#3)}}

\newcommand{\HAbs}[1]{\mathsf{HA}_{#1}}
\newcommand{\Struct}[1]{\mathsf{St}_{#1}}

\newcommand{\Irred}[3]{\mathsf{Irred}^{\Usym}_{#1,#2,#3}}
\newcommand{\NF}[3]{\mathsf{NF}^{\Usym}_{#1,#2,#3}}

\newcommand{\rulename}{\rho}
\newcommand{\ruledb}{\dbsym}
\newcommand{\rulelsv}{\lsvsym}
\newcommand{\rulesub}[2]{\mathsf{sub}_{(#1,#2)}}
\newcommand{\aset}{\mathcal{A}}
\newcommand{\asettwo}{\mathcal{B}}
\newcommand{\sset}{\mathcal{S}}
\newcommand{\appflag}{\mu}
\newcommand{\app}{@}
\newcommand{\nonapp}{{\text{$\not$@}}}

\newcommand{\indrulename}[1]{\textup{\textsc{#1}}}

\newcommand{\ruleUDb}{\indrulename{db$^{\Usym}$}\xspace}
\newcommand{\ruleUSub}{\indrulename{sub$^{\Usym}$}\xspace}
\newcommand{\ruleULsv}{\indrulename{lsv$^{\Usym}$}\xspace}
\newcommand{\ruleUAppL}{\indrulename{appL$^{\Usym}$}\xspace}
\newcommand{\ruleUAppR}{\indrulename{appR$^{\Usym}$}\xspace}
\newcommand{\ruleUEsR}{\indrulename{esR$^{\Usym}$}\xspace}
\newcommand{\ruleUEsLAbs}{\indrulename{esLA$^{\Usym}$}\xspace}
\newcommand{\ruleUEsLStruct}{\indrulename{esLS$^{\Usym}$}\xspace}

\newcommand{\ruleUNFVar}{\indrulename{NF-var$^{\Usym}$}\xspace}
\newcommand{\ruleUNFLam}{\indrulename{NF-lam$^{\Usym}$}\xspace}
\newcommand{\ruleUNFApp}{\indrulename{NF-app$^{\Usym}$}\xspace}
\newcommand{\ruleUNFEsAbs}{\indrulename{NF-esA$^{\Usym}$}\xspace}
\newcommand{\ruleUNFEsStruct}{\indrulename{NF-esS$^{\Usym}$}\xspace}

\newcommand{\arUd}[4]{\ar[d]_-{#1,#2,#3,#4}}
\newcommand{\arUr}[4]{\ar[r]_-{#1,#2,#3,#4}}
\newcommand{\arsdUd}[4]{\ar@{.>}[d]_-{#1,#2,#3,#4}}
\newcommand{\arsdUr}[4]{\ar@{.>}[r]_-{#1,#2,#3,#4}}

\newcommand{\expansion}[2]{{#1}^{#2}}

\newcommand{\typesystem}{\mathcal{U}}
\newcommand{\systemV}{\mathcal{V}}

\newcommand{\isAppr}[2]{\mathtt{appropriate}_{#1}(#2)}

\newcommand{\ruleTypVarV}{\indrulename{var}_{\mathtt{v}}\xspace}
\newcommand{\ruleTypAbsV}{\indrulename{abs}_{\mathtt{v}}\xspace}
\newcommand{\ruleTypAppV}{\indrulename{app}_{\mathtt{v}}\xspace}
\newcommand{\ruleTypESV}{\indrulename{es}_{\mathtt{v}}\xspace}

\newcommand{\ruleTypVar}{\indrulename{var}\xspace}
\newcommand{\ruleTypAbs}{\indrulename{abs}\xspace}
\newcommand{\ruleTypAppP}{\indrulename{appP}\xspace}
\newcommand{\ruleTypAppC}{\indrulename{appC}\xspace}
\newcommand{\ruleTypES}{\indrulename{es}\xspace}

\newcommand{\derivable}[2]{\ensuremath{{#1} \rhd {#2}}}
\newcommand{\derivablep}[3]{\ensuremath{{#1} \rhd_{#3} {#2}}}

\newcommand{\deriv}{\Phi}
\newcommand{\derivtwo}{\Psi}

\newcommand{\sizeV}[1]{\ensuremath{\mathtt{sz}_{\systemV}\left({#1}\right)}}
\newcommand{\sizeU}[1]{\ensuremath{\mathtt{sz}_{\typesystem}\left({#1}\right)}}

\newcommand{\judg}[4][]{#2\vdash^{(#1)}#3:#4}
\newcommand{\judgs}[3]{#1\vdash #2:#3}

\newcommand{\dom}[1]{\mathsf{dom}(#1)}
\newcommand{\emptyctx}{\emptyset}
\newcommand{\tctx}{\Gamma}
\newcommand{\tctxtwo}{\Delta}
\newcommand{\TEnv}[3]{\mathsf{TEnv}(#1,#2,#3)}

\newcommand{\mset}[1]{[#1]}
\newcommand{\emset}{[\,]}

\newcommand{\typ}{\alpha}
\newcommand{\typV}{\sigma}
\newcommand{\typtwo}{\beta}
\newcommand{\typtwoV}{\tau}

\newcommand{\basetype}{\alpha}
\newcommand{\basetypetwo}{\beta}

\newcommand{\nityp}{\mathcal{M}}
\newcommand{\nityptwo}{\mathcal{N}}
\newcommand{\niunion}{\mathrel{\uplus}}

\newcommand{\mleq}{\ll}
\newcommand{\mtyp}{\mathcal{T}}
\newcommand{\mtyptwo}{\mathcal{S}}
\newcommand{\mtypthree}{\mathcal{R}}
\newcommand{\optmtyp}{\mtyp^?}
\newcommand{\optmtyptwo}{\mtyptwo^?}
\newcommand{\optmtypthree}{\mtypthree^?}

\newcommand{\none}{\bot}
\newcommand{\tight}{\mathbb{t}}
\newcommand{\tightN}{\mathbb{s}}

\newcommand{\cm}{m}
\newcommand{\ce}{e}
\newcommand{\cmtwo}{{m'}}
\newcommand{\cetwo}{{e'}}

\newcommand{\numarrows}[1]{\mathtt{ta}(#1)}

\DeclareMathAlphabet\mathbfcal{OMS}{cmsy}{b}{n}

\def\eg{{\em e.g.}}
\def\cf{{\em cf.}}

\begin{document}

\title{A Quantitative Interpretation of Useful Call-by-Value}

\titlecomment{{\lsuper*} This work is a revised and extended version of~\cite{BKM26}. 
It includes additional explanations, examples, and proofs.
\\
  This work is supported by the European Union through the 
  MSCA SE project QCOMICAL (Grant Agreement ID: 101182520) and
  through the EU Horizon 2020 research and innovation
  programme MSCA No 945332.
}

\author[P.~Barenbaum]{Pablo Barenbaum\lmcsorcid{0009-0003-2494-3345}}[a,b]
\author[D.~Kesner]{Delia Kesner\lmcsorcid{0000-0003-4254-3129}}[c]
\author[M.~Milicich]{Mariana Milicich\lmcsorcid{0009-0007-5722-2607}}[c,d]

\address{Universidad Nacional de Quilmes (CONICET), Argentina}

\address{Instituto de Ciencias de la Computaci\'on, UBA, Argentina}
\email{pbarenbaum@dc.uba.ar}

\address{Universit\'e Paris Cit\'e, CNRS, IRIF, France}
\email{kesner@irif.fr}

\address{University of Tartu, Estonia}
\email{mariana.milicich@ut.ee}



\begin{abstract}
  \emph{Useful evaluation} is an 
  optimised evaluation mechanism for functional programming languages.
  It relies on representing terms with \emph{sharing} and imposing a 
  restricted notion of \emph{useful substitutions}, that intuitively 
  disallows copying subterms that do not contribute to the progress of 
  the computation.

  In particular, \emph{useful call-by-value evaluation} optimises the
  standard \emph{call-by-value} strategy by preserving
  its original semantics. 
  This preservation result has been shown by means of syntactical 
  rewriting techniques, by relying on a detailed analysis of the reduction rules and their interactions.
  Such arguments are technically involved and closely tied to their particular operational semantics, making them difficult to adapt to alternative variants or extensions of the calculi at play.

  In this work, we present the first semantic interpretation of useful
  call-by-value evaluation through a
  \emph{non-idempotent intersection type system} called
  $\typesystem$. 
  One first contribution of this paper is a characterisation of
  termination of useful call-by-value evaluation via $\typesystem$-typability. 
  That is, a term is typable in system $\typesystem$
  if and only if it terminates in the useful call-by-value strategy.
  As a second contribution, we show that system $\typesystem$ provides
  a \emph{quantitative} interpretation for useful call-by-value
  evaluation, offering \emph{exact} step-count information for program
  evaluation. 
  Even
  though the specification of the operational semantics of useful
  evaluation is highly complex, system $\typesystem$ is notably
  simple.  As far as we know, system $\typesystem$ is one of the
  scarce quantitative type systems capturing exactly the substitution
  step-count for variables \emph{and} abstractions in an open
  call-by-value strategy.

  Our third contribution is a termination-equivalence result between standard call-by-value, which may perform erasing reductions, and useful call-by-value, which is non-erasing. 
  To establish this result, we introduce a non-erasing restriction of standard call-by-value, a refinement that serves as an intermediate calculus bridging standard and useful call-by-value. We show that all three calculi have the same termination behaviour, and hence induce the same \emph{observational equivalence}.
\end{abstract}

\maketitle

\section{Introduction}
\label{sec:intro}

Efficiently implementing an evaluation strategy of a programming language requires designing an operational semantics that faithfully captures the original model, while optimising its key operations.
The resulting model usually differs significantly from the
original one at the operational level, thus, it is not clear that the
former preserves the \emph{observational behaviour} of the latter:
given any program, is its observable behaviour in the optimised model
indistinguishable from its behaviour in the original model?

Call-by-need (\CBNeed) optimising call-by-name (\CBN) is a typical and
well-known example of this phenomenon, and both strategies 
have been shown to be observationally equivalent~\cite{AriolaFMOW95,MaraistOTW95,Kesner16}.

Another recent example of an optimised strategy is \emph{useful evaluation}, introduced by Accattoli
and Dal Lago~\cite{AccattoliL14}.
This notion is based on a graph 
\emph{representation of $\lambda$-terms} which implements
\emph{sharing} by means of \emph{explicit substitutions}\footnote{See~\cite{Kesner2009} for a survey on ESs.} (ESs). 
Useful evaluation aims to perform term duplication only
when it 
\emph{contributes to the progress of the computation}, that is, to 
the creation of $\beta$-redexes.
Defining useful evaluation is challenging because it must precisely capture what it means for a reduction step to ``contribute to the progress of the computation''. 
This contribution is not a property that can always be determined \emph{directly} from the term being evaluated, but rather \emph{indirectly} because it depends on how the current step affects the subsequent ones. 
Deciding whether a reduction step is useful therefore requires information about the evaluation context in which it occurs, making useful evaluation inherently \emph{context-sensitive} by nature.

In the call-by-value (\CBV) framework, useful evaluation can be seen 
as an optimised implementation of standard \CBV that effectively 
preserves its original semantics~\cite{AccattoliC15,AccattoliG17}.
This preservation result has been proved by means of syntactical tools
based on rewriting theory and abstract machines.

In general, proving equivalence between two calculi through operational reasoning 
is inherently complex, often relying on 
various syntactic techniques such as sharing, 
unfolding, residual theory, standardisation, among others. 
Moreover, syntactic proofs of observational equivalence are often 
\emph{ad hoc} and difficult to generalise to alternative variants of 
the calculi at play. 
In contrast, the semantic methods used in this work have provided abstract
and conceptual proofs when applied before to other cases; this is, for 
example, the emblematic observational equivalence result between \CBNeed and 
\CBN achieved through non-idempotent intersection typing~\cite{Kesner16},
which illustrates how type-theoretic methods can be used to reason about the equivalence of calculi.
Here, we extend these methods to the \CBV framework, focusing in particular on usefulness.

\paragraph{\textbf{Qualitative and Quantitative Interpretations.}} 
Building on the type-theoretic methodology mentioned above, this work relies on non-idempotent intersection type systems for capturing useful \CBV evaluation.
\defn{Intersection types} extend simple types 
with an intersection type constructor $\cap$, 
allowing a program $\tm$ to be typed with $\typ \cap \typtwo$ if $\tm$ types
with both $\typ$ and $\typtwo$ independently. 
Initially introduced as \emph{models} to capture
\emph{qualitative} computational
properties of functional programs~\cite{CoppoDezani78,CoppoD80},
intersection type systems are able, in particular, to
characterise termination of evaluation strategies: a program $\tm$
terminates in a given evaluation strategy if and only if $\tm$ is
typable in an associated type system.

In the original formulation of intersection types, the type
constructor $\cap$ is considered to be \defn{idempotent}, meaning
that $\typ \cap \typ = \typ$. 
So, an intersection $\typ_1 \cap \hdots \cap \typ_n$ can be 
represented as the \emph{set} $\set{\typ_1,\hdots,\typ_n}$. 
More recent works have adopted a \defn{non-idempotent} intersection 
type constructor~\cite{Gardner94,deCarvalho2007,deCarvalho18}. 
In this setting, an intersection $\typ_1 \cap \hdots \cap \typ_n$ can be 
represented as the \emph{multiset} $\mset{\typ_1,\hdots,\typ_n}$. 

Like idempotent intersection types,
non-idempotent intersection types still allow 
characterising operational properties of programs by means of 
typability but also yield a substantial improvement: they provide 
\emph{quantitative} measures about these 
properties~\cite{Gardner94,deCarvalho2007,deCarvalho18}.
For instance, from a typing derivation of a program $\tm$, it is possible to extract an
\emph{upper bound}, or even an \emph{exact measure}, for the
number of steps which is necessary to reach the normal form of $\tm$;
such information comes from calculating the \emph{size} of the typing derivation of $\tm$. Furthermore, exact measures can be obtained in the so-called
\emph{tight} quantitative type systems~\cite{AccattoliGK18}, by
decorating the typing judgements with \emph{counters}, and by imposing
a \emph{tightness} condition on typing derivations to ensure that
they are \emph{minimal}.
For instance, the counter $m$ may be used in the typing judgement 
$\tctx \vdash^m \tm : \typ$ to capture the fact that $\tm$ evaluates in
exactly $m$ reduction steps to a normal form.

It is worth noticing that semantic models of computation are often \defn{qualitative}, in that they
capture \emph{extensional} properties of programs, such as
observational equivalence. 
In contrast, type systems based on non-idempotent intersection types refine this perspective by
capturing \emph{intensional} properties of programs, such as the
length of reductions to normal form, and are therefore considered to
be \defn{quantitative}.

To (quantitatively) characterise termination for a given evaluation
strategy in non-idempotent intersection type systems, one relies on the results of
\emph{(quantitative) soundness} and \emph{(quantitative) completeness} of the type system with
respect to the corresponding evaluation strategy. 
Here, \defn{quantitative soundness} means that for any 
type derivation of a program $\tm$ with counter (or size) $\cm$, the program 
$\tm$ evaluates to normal form in \emph{at most} $\cm$ steps. 
Conversely, \defn{quantitative completeness} means that if a program $\tm$
evaluates to normal form in $\cm$ steps, then $\tm$
can be typed with a derivation whose size
provides an upper bound on $\cm$.
By restricting typing derivations to the so-called tight ones, it is possible to refine both quantitative soundness and completeness results.
\defn{Exact soundness} means that for any \emph{tight} type derivation of a program $\tm$ with counter $\cm$, the program $\tm$ evaluates to normal form in \emph{exactly} $\cm$ steps. 
\defn{Exact completeness} means that every reduction sequence to normal form of a given length $\cm$ corresponds to a \emph{tight} typing derivation with appropriate counter $\cm$.

Quantitative type systems based on non-idempotent intersection types have been applied to
various evaluation strategies in the $\lambda$-calculus, providing 
upper bounds and exact measures of evaluation lengths~\cite
{BernadetGL13,BucciarelliKV17,AccattoliGK18,deCarvalho18,KesnerViso22}.
This line of work was also developed in classical calculi~\cite
{KesnerVial17,KesnerVial20}, call-by-value~\cite
{Ehrhard12,AccattoliGuerrieri18,Guerrieri18,AccattoliGuerrieri22},
call-by-need~\cite
{BalabonskiBBK17,BarenbaumBM18,AccattoliGL19,Leberle21},
call-by-push-value~\cite{BucciarelliKRV20}, languages with
effects~\cite{AlvesKR23}, linear
logic~\cite{CarvalhoPF11,CarvalhoF16}, among others. 
The quantitative nature of these systems stems from non-idempotency: multiple occurrences of a type are kept distinct rather than collapsed, allowing typing derivations to retain information about the precise use of computational resources. As a consequence, typability can characterise qualitative properties such as termination, while suitable quantitative information extracted from typing derivations can provide bounds, or even exact measures, of evaluation lengths.

Beyond their operational use, these type systems also admit a semantic
reading.  Non-idempotent intersection types can be understood as
syntactic presentations of the relational model of linear
logic~\cite{Girard88,BucciarelliE01}.  This connection gives them a
dual role: they provide a denotational account of programs while
retaining sufficiently fine quantitative information to describe their operational behaviour.

\paragraph{\textbf{Motivation and Goal of this Paper.}}
In the framework of Plotkin's \CBV~\cite{Plotkin75}, non-idempotent intersection type systems
have been introduced by Ehrhard~\cite{Ehrhard12}, and then extended to
different calculi with sharing (see \eg~\cite{AccattoliGuerrieri18,AccattoliGL19,BucciarelliKRV20,KesnerViso22,AccattoliGuerrieri22}), none of which captures useful \CBV evaluation. 
All these type systems characterise termination of the 
corresponding \CBV calculus, but their underlying notion of 
\emph{tightness} does not appear to extend naturally to \emph{useful} \CBV. 
Said differently, useful \CBV has so far lacked a quantitative semantic interpretation, and existing quantitative interpretations of \CBV have not been shown to capture usefulness.
The problem of finding a semantic interpretation of useful \CBV capable of 
providing the \emph{exact} number of steps to reach a normal form 
thus remained open.
The motivation of this work is then to fill this gap, providing, to the best of our knowledge,
the first semantic account of useful \CBV.

Qualitative semantic models are unable to distinguish between
intensional/quantitative properties of alternative implementations,
such as \CBN vs \CBNeed, or \CBV vs useful \CBV.  In particular, we
aim to \textbf{deepen the understanding} of useful \CBV evaluation by
developing tools for reasoning \emph{quantitatively} about program
behaviour.  To this end, we develop a sound and complete
\textbf{quantitative interpretation} of useful \CBV taking the form of a \textbf{non-idempotent intersection type system} that admits a \emph{tight} extension. 

A key requirement of such an interpretation is the ability to distinguish between different kinds of reduction steps. 
Indeed, it is well-known that the number of function application steps and substitution steps may differ exponentially, 
a phenomenon known as the \emph{size explosion problem} (see \eg~\cite{AccattoliL14}).
Our quantitative interpretation is therefore designed to account separately for these two forms of reduction. 
This distinction is particularly relevant to the analysis of cost models and the design of optimisation techniques.

\paragraph{\textbf{Contributions and Structure of the Paper.}} 

Our contributions are twofold, encompassing both qualitative and quantitative results. 
On the qualitative side, we use semantic methods based on non-idempotent intersection types to establish the observational equivalence of \CBV and useful \CBV. 
This shows, in particular, that useful evaluation can be understood as an optimisation of \CBV that preserves its observational behaviour. 
On the quantitative side, we provide an exact account of useful \CBV evaluation through a new non-idempotent intersection type system $\typesystem$, whose typing derivations precisely capture the number of function application and substitution steps performed during evaluation. 
The remainder of this section presents these contributions
in more detail while outlining the structure of the paper.

Section~\ref{sec:preliminaries} recalls Accattoli and Paolini's
\emph{Value Substitution Calculus} (\VSC)~\cite{AccattoliP12}, the
formal framework serving as the starting point for studying the
operational semantics of \CBV in this work. 
Whereas most functional programming languages rely on reduction for \emph{closed} terms
(\ie~without free variables), reduction in \VSC applies also to
\emph{open} terms (\ie~with free variables).
Adopting an open \CBV setting broadens the applicability of our results, as open \CBV serves
as a foundational framework for more complex evaluation models, such as
\emph{strong} \CBV ~\cite{GregoireL02}, which is used to implement
proof assistants based on dependent type theory.

Section~\ref{sec:preliminaries} also recalls system $\systemV$~\cite{BucciarelliKRV20}, the non-idempotent intersection type system extending the relational semantics of \CBV (see \eg~\cite{Ehrhard12}).

\paragraph{\textbf{The qualitative results.}}

Since \VSC is a (potentially) \emph{erasing} calculus (\ie~values can be erased during computation), while useful \CBV is not,
relating their operational semantics is not straightforward.
To overcome this barrier, we make an intermediate step and define in Section~\ref{sec:non-erasing} a non-erasing variant of \VSC
we dub \VSCne, which enjoys in particular the following property:
\begin{center}
\textbf{(Prop$_{\mathbf{1}}$)}: 
  \VSC-termination is equivalent to \VSCne-termination
\end{center}
The proof of this property follows the methodology introduced in~\cite{Kesner16}
and relies on the non-idempotent intersection type system $\systemV$
(\cf~Section~\ref{sec:preliminaries}): a term $t$ is typable in $\systemV$ if and only if it is \VSC-terminating (resp. \VSCne-terminating). While the former characterisation is already known~\cite{BucciarelliKRV20}, the latter is relatively routine and is established in Section~\ref{sec:non-erasing}
(\autoref{cor:v-characterises-vscne}).

In Section~\ref{sec:ucbv}, we recall
an inductive formulation of the useful \CBV evaluation strategy,
called \UOCBV, recently defined~\cite{BKM-CADE25}.
In Section~\ref{sec:systemU}, we define a \emph{new} non-idempotent
intersection type system called $\typesystem$.
System $\typesystem$ is first shown to
capture \VSCne-termination
(\autoref{cor:u-characterises-vscne}), which implies in
particular that $\systemV$ and $\typesystem$ are equally expressive
(\autoref{cor:v-and-u}).

In Section~\ref{sec:ucbv-U}, system $\typesystem$ is also shown to capture \UOCBV-termination (\autoref{cor:qualitative_characterisation_of_termination}), 
thus solving the lack of a high-level, semantic characterisation of useful \CBV.

In Section~\ref{sec:obs_eq} we put together all the previous results:
following again the methodology from~\cite{Kesner16}, we have that
termination in \VSCne is thus equivalent to typability in system
$\typesystem$, which is in turn equivalent to \UOCBV-termination:
\begin{center}
\textbf{(Prop$_{\mathbf{2}}$)}: 
  \VSCne-termination is equivalent to \UOCBV-termination
\end{center}

From \textbf{(Prop$_{\mathbf{1}}$)} and \textbf{(Prop$_{\mathbf{2}}$)} we thus obtain
that \VSC, \VSCne, and \UOCBV are termination-equivalent
(\autoref{coro:main_result}), leading to conclude that \VSC, \VSCne, and \UOCBV are pairwise observationally equivalent (\autoref{cor:obseq}), so that, in particular,
\UOCBV evaluation can be seen as an optimisation of \CBV that preserves its original semantics.
A diagram summarising all our (qualitative) equivalence results can be found in~\autoref{fig:diagram_results}. 

\begin{figure}[t]
  \centering
  \begin{tikzpicture}[
    node distance=1cm and 2.5cm,
    mydoublearrow/.style={{Latex}-{Latex}},
    mydoublearrowdots/.style={{Latex}-{Latex},dotted},
  ]

  \node (vsc) {\VSC};
  \node[right=of vsc] (vscne) {\VSCne};
  \node[right=of vscne] (ucbv) {\UOCBV};
  \node[above=of vscne] (systemV) {System $\systemV$};
  \node[below=of vscne] (systemU) {System $\typesystem$};

  \draw[mydoublearrowdots] (vsc) -- (vscne);
  \draw[mydoublearrowdots] (vscne) -- (ucbv);
  \draw[mydoublearrow] (systemU) -- node[below right] {\tiny \autoref{cor:qualitative_characterisation_of_termination}} (ucbv);
  \draw[mydoublearrow, color = cyan!60!black] (systemV) -- node[above left] {\tiny \cite[Theorem~4.24]{BucciarelliKRV20}} (vsc);
  \draw[mydoublearrow] (systemV) -- node[right] {\tiny \autoref{cor:v-characterises-vscne}} (vscne);
  \draw[mydoublearrow] (systemU) -- node[left] {\tiny \autoref{cor:u-characterises-vscne}} (vscne);
  \end{tikzpicture}
\caption{Correspondence between \VSC, \VSCne, and \UOCBV through Systems $\systemV$ and $\typesystem$}
\label{fig:diagram_results}
\end{figure}

\paragraph{\textbf{The quantitative results.}} 

As a second main contribution, we refine system $\typesystem$ to provide
\emph{exact} quantitative information of \UOCBV evaluation. 
Indeed, we equip typing judgements of system $\typesystem$ with counters that are natural numbers, and we define a notion of \emph{tightness}, intuitively 
capturing minimality of derivations. 
The counters are meant to capture the exact step-count for program evaluation. 
More precisely, we show that a term $\tm$ is \emph{tightly} typable 
with counters $\cm$ and $\ce$ in system $\typesystem$ if and only if 
$\tm$ terminates in \UOCBV in exactly $\cm$ function 
application steps and $\ce$ substitution steps
(\autoref{cor:quantitative_characterisation_of_termination}).

In particular, system $\typesystem$ is one of the scarce quantitative 
type systems capturing exactly the substitution step-count for 
variables and abstractions in an open \CBV strategy.

\section{Preliminaries}
\label{sec:preliminaries}

In this section, we start by recalling some general notions of
rewriting theory~(\autoref{subsec:general_notions}). 
We then review the \emph{Value Substitution Calculus}~(\VSC), a \CBV 
calculus introduced by Accattoli and Paolini~\cite{AccattoliP12} (\autoref{subsec:vsc}), 
as well as the type system $\systemV$ introduced in~\cite{BucciarelliKRV20} (\autoref{subsec:systemV}).

\subsection{Preliminary Notions on Rewriting Theory}
\label{subsec:general_notions} 

Given a \defn{reduction system} $\relation$, we write:
\begin{itemize}
\item 
  $\rewrite\relation$ to denote the (one-step) reduction relation 
  associated with $\relation$
\item
  $\rewritensteps\relation{n}$ for the composition of $n$-steps of $\rewrite\relation$, also called an $\relation$-reduction sequence
\end{itemize}

A term $\tm$ is said to be:
\begin{itemize}
\item 
  \defn{$\relation$-irreducible}, or in 
  \defn{$\relation$-normal form}, written $\tm \not \rewrite\relation$,
  if there is no $\tmtwo$ such that $\tm \rewrite\relation \tmtwo$ 
\item
  \defn{$\relation$-weakly terminating} if there exists some $\relation$-normal form $\tmtwo$ such that 
  $\tm \rewriten\relation \tmtwo$
\item
  \defn{$\relation$-terminating} if there is no infinite 
  $\relation$-sequence starting at $\tm$
\item
  \defn{$\relation$-diamond} if $\tm \rewrite\relation \tm_0$ and 
  $\tm \rewrite\relation \tm_1$, with $\tm_0 \neq \tm_1$, imply there 
  is $\tm'$ such that $\tm_0 \rewrite\relation \tm'$ and 
  $\tm_1 \rewrite\relation \tm'$ 
\item
  \defn{$\relation$-confluent} if $\tm \rewriten\relation \tm_0$ and $\tm \rewriten\relation \tm_1$ 
  imply there is $\tm'$ such that $\tm_0 \rewriten\relation \tm'$ and 
  $\tm_1 \rewriten\relation \tm'$
\end{itemize} 
Moreover, we use the terminology \defn{$\relation$-evaluation} to talk about
a reduction sequence to $\relation$-normal form.

A reduction relation associated to $\relation$ is \defn{weakly-terminating} (resp. 
\defn{terminating}, \defn{diamond}, \defn{confluent}) if and only if 
every term is $\relation$-weakly terminating (resp. $\relation$-terminating, 
$\relation$-diamond, $\relation$-confluent). 

Any relation verifying the diamond property is in particular 
confluent.
Moreover, if $\tm$ is confluent, then its $\relation$-normal form, if 
it exists, is unique~\cite{Terese03}.

\subsection{The Value Substitution Calculus}
\label{subsec:vsc}

It is well-known that naive extensions of Plotkin's \CBV~\cite{Plotkin75}
to \emph{open \CBV} do not enjoy adequacy, in the sense that normal
forms are not necessarily \emph{solvable}, \ie ``meaningful''~\footnote{
See~\cite{PaoliniR99,Garcia-PerezN16,AccattoliGuerrieri22} for a general
discussion on \CBV solvability.}.
A typical example is the term $\tm_0 \eqdef (\lam \varthree \omega)\, (\vartwo \vartwo)\, \omega$, where $\omega \eqdef \lam \var \var\,\var$. 
Indeed, $\tm_0$ is a normal form, but is not solvable, as it is observationally equivalent to the unsolvable (and thus meaningless) term $\Omega \eqdef \omega\, \omega$.  
This weakness of naive extensions of Plotkin's \CBV to open terms
is widely accepted in the literature on
\CBV, which provides many alternative proposals to extend \CBV to the open setting
\cite{AccattoliG16,AccattoliP12,CarraroGuerrieri2014,GregoireL02,HerbelinZimmermann2009}. 
Among them, we chose the formalism in~\cite{AccattoliP12}, called here the \emph{Value Substitution Calculus}~(\VSC),
as it has good rewriting properties.
Indeed, \VSC is a non-trivial extension of Plotkin's \CBV to the open setting, using 
explicit substitutions and reduction \emph{at a distance} to recover adequacy.

Given a denumerable set of \defn{variables} 
($\var,\vartwo,\varthree,\hdots$), the syntax of \VSC is 
specified by:
\[ 
  \begin{array}{rlll}
    \mbox{\defn{(Terms)}} & \tm, \tmtwo  & \eqgram & \var \mid \lam\var\tm \mid \tm\,\tmtwo \mid \tm\esub\var{\tmtwo} \\
    \mbox{\defn{(Values)}} & \val, \valtwo & \eqgram & \var \mid \lam\var\tm
  \end{array}
\]
The set of terms includes \defn{variables}, \defn{abstractions},
\defn{applications}, and \defn{closures} $\tm\esub\var\tmtwo$, 
representing an \defn{explicit substitution} (ES) $\esub\var\tmtwo$ 
affecting a term $\tm$. 

Another recurring term throughout this article, besides $\omega$ and $\Omega$ mentioned before, is the identity function $\id \eqdef \lam\var\var$.

\defn{Free} and \defn{bound occurrences} of variables are defined as 
usual, where free occurrences of $\var$ in $\tm$ are bound in 
$\tm\esub\var\tmtwo$.
Terms are considered up to $\alpha$-renaming of bound variables.

We also use the following notions of contexts:
\[ 
  \begin{array}{rlll}
    \mbox{\defn{(Substitution contexts)}} & \sctx, \sctxtwo & \eqgram & \ctxhole \mid \sctx\esub\var\tm \\
    \mbox{\defn{(Weak contexts)}} & \wctx & \eqgram & \ctxhole \mid \wctx\, \tm \mid \tm \, \wctx \mid \wctx\esub\var{\tm} \mid \tm\esub\var{\wctx}
  \end{array}
\]

Substitution contexts are lists of ESs, and constitute a particular
class of weak contexts. 
We write $\ctxof \wctx \tm$ for 
the \defn{replacement} of the hole $\ctxhole$ in $\wctx$ by $\tm$, 
which may capture the free variables of $\tm$. This notation is
simplified  for substitution contexts, for which we may write $\tm\sctx$
instead of $\ctxof \sctx \tm$.
For example, if 
$\sctx = \esub\var{\vartwo\,\vartwo}\esub\vartwo\varthree$ and 
$\tm = \var$, then 
$\tm\sctx = \var\esub\var{\vartwo\,\vartwo}\esub\vartwo\varthree$.

We also consider a \defn{meta-level operation} $\tm\sub\var\tmtwo$
which substitutes all the free occurrences of 
$\var$ by $\tmtwo$ in $\tm$. 
Thanks to $\alpha$-renaming, we
can assume that this substitution operation
is \textit{capture-free}. 
The sets of \defn{free variables} of
a \defn{term} ($\fv\tm$) and of a
\defn{substitution context} ($\fv\sctx$) are defined as expected.

\defn{Reduction in \VSC} is defined 
by the following two rewriting rules, closed by arbitrary \emph{weak} 
contexts, \ie reduction is not allowed under 
abstractions:
\[
    (\lam\var\tm)\sctx\,\tmtwo
    \todb 
    \tm\esub\var\tmtwo\sctx
  \bsep
    \tm\esub\var{\val\sctx}
    \tosv
    \tm\sub\var\val\sctx
\]

The $\dbsym$-rule (``$\mathsf{d}$istant $\mathsf{b}$eta'') is a 
$\beta$-like rule that applies an abstraction $\lam\var\tm$ 
to an argument $\tmtwo$, creating an ES $\esub\var\tmtwo$ that binds 
$\var$ to the argument $\tmtwo$.
This rule acts at a \emph{distance}, meaning that an arbitrary list of ESs 
($\sctx$) may be in between the abstraction and its argument.
We assume
the rule to be capture-free due to $\alpha$-renaming (that is,
the set $\fv{\tmtwo}$ is not
affected by $\sctx$).
An example of $\dbsym$-reduction is: 
$(\lam\var{\lam\vartwo\var}) \, \tm \, \tmtwo
\todb (\lam\vartwo\var)\esub\var\tm \, \tmtwo
\todb \var\esub\vartwo\tmtwo\esub\var\tm$.

The $\svsym$-rule (``$\mathsf{s}$ubstitution of $\mathsf{v}$alues'')
fires an ES when $\var$ is bound to a value $\val$ surrounded by a context $\sctx$.
This performs the meta-level substitution of $\var$ by 
$\val$, potentially producing $n \geq 0$ copies of $\val$, while 
it extrudes the substitution context $\sctx$, which must remain 
shared.
As with the $\dbsym$-rule, we also assume the $\svsym$-rule to be capture-free due to $\alpha$-renaming. 
Two examples of $\svsym$-reduction follow:
on the first hand, 
$\var\esub\vartwo\varthree \tosv \var\sub\vartwo\varthree = \var$, which
results in the erasure of the value $\varthree$, while the step
$(\var\,\var)\esub\var{(\lam\vartwo\varthree)\esub\varthree\tm}
\tosv
((\lam\vartwo\varthree)\,\lam\vartwo\varthree)\esub\varthree\tm$
produces two copies of the ES $\lam\vartwo\varthree$ but keeps a 
single copy of $\esub\varthree\tm$.
We say that the first step is \emph{erasing}, while
the second one is \emph{non-erasing}.
Note that a term like $\var\esub\var{\vartwo \, \vartwo}$ is in 
\VSC-normal form because $\vartwo \, \vartwo$ is not a value.

The notion of \defn{reachable variable} of a term $\tm$ is 
relevant in the (weak) call-by-value setting because
reduction is closed under weak contexts.
The set of reachable variables $\rv\tm$ is defined as
the set of free variables that do not appear inside an abstraction.
Formally:
\[ 
  \begin{array}{lllllll}
    \rv\var                 & \eqdef & \set\var & \HS\HS
    \rv{\lam\var\tm}        & \eqdef & \emptyset
  \\
    \rv{\tm\,\tmtwo}        & \eqdef & \rv\tm \cup \rv\tmtwo & \HS\HS
    \rv{\tm\esub\var\tmtwo} & \eqdef & (\rv\tm \setminus \set\var) \cup \rv\tmtwo
  \end{array}
\]

For example, 
the set of free variables of the term $\var \esub\var\vartwo (\lam\var\varthree)$
is $\fv{\var \esub\var\vartwo (\lam\var\varthree)} = \set{\vartwo,\varthree}$,
whereas $\rv{\var \esub\var\vartwo (\lam\var\varthree)} = \set\vartwo$.

We end this section by introducing a working example that will be used throughout the paper.
\begin{exa}
\label{exa:working-example}
Let $\tm \eqdef (\var \, \varthree)\esub\var{\vartwo\esub\vartwo\id}$.
Then \[ \tm \tosv (\vartwo \, \varthree)\esub\var{\vartwo}\esub\vartwo\id
 \tosv (\id\, \varthree)\esub\var{\id}\esub\vartwo\id
 \todb  \varfour \esub \varfour \varthree\esub\var{\id}\esub\vartwo\id
\tosv   \varthree \esub \varfour \varthree\esub\var{\id}\esub\vartwo\id \] 
\end{exa}


\subsection{The Type System \texorpdfstring{$\systemV$}{V}}
\label{subsec:systemV}

We now recall the non-idempotent intersection type system $\systemV$ introduced in~\cite{BucciarelliKRV20},
which  extends the relational semantics of \CBV (see \eg~\cite{Ehrhard12}). 
More precisely, 
consider a set of \defn{base types} $(\basetype,\basetypetwo,\hdots)$.
The grammar of types of system $\systemV$ is given by:
\[
  \begin{array}{rlcl@{\sep\sep}rlcll}
    \text{\defn{(Types)}}       & \typV, \typtwoV & \eqgram & \basetype \mid \nityp \mid \nityp \to \typV
  \\
    \text{\defn{(Multi-Types)}} & \nityp, \nityptwo & \eqgram & \mset{\typV_i}_\iI & \text{where $I$ is a finite set}
  \\ 
  \end{array}
\] 

Like most non-idempotent intersection type systems, system $\systemV$ is inspired by Linear Logic~\cite{Girard87}. 
In particular, a multi-type $\mset{\typV_1,\hdots,\typV_n}$ can be understood as 
a multiplicative conjunction $\typV_1 \otimes \hdots \otimes \typV_n$.

We use $\length\nityp$ to denote the number of types in the multi-type $\nityp$.
\defn{Typing environments} $\tctx, \tctxtwo, \ldots$ are functions 
mapping variables to multi-types, assigning $\emset$ to all but 
finitely many variables.
The \defn{domain} of an environment $\tctx$ is
$\dom\tctx \eqdef \set{\var \mid \tctx(\var) \neq \emset}$, and
$\emptyctx$ denotes the empty typing environment, mapping each 
variable to $\emset$.

The \defn{union of multi-types}, written
$\nityp_1 \niunion \nityp_2$, is a multiset of types defined as
expected, where $\emset$ is the neutral element.
For typing environments $(\tctx_i)_\iI$ indexed over a finite set $I$, we write $+_\iI \tctx_i$
for the environment mapping each variable $\var$ to 
$\niunion_\iI \tctx_i(\var)$, where $\tctx + \tctxtwo$ and 
$\tctx +_\jJ \tctxtwo_j$ (meaning $\tctx + (+_\jJ \tctxtwo_j)$)
are taken to be particular instances of the general 
notation. 
When $\dom\tctx \cap \dom\tctxtwo = \emptyset$, we write 
$\tctx;\tctxtwo$ instead of $\tctx + \tctxtwo$ to emphasise that the 
domains of $\tctx$ and $\tctxtwo$ are disjoint. 
We write $\var_1:\nityp_1; \ldots; \var_n:\nityp_n$ to make explicit that each $\var_i$ is mapped to $\nityp_i$ ($i \in \set{1,\hdots,n}$). 
Note that the environment $\tctx;\var:\emset$ is identical to $\tctx$.

\defn{Typing judgements} are of the form
$\judgs\tctx\tm\typV$, where $\tctx$ is a type environment, $\tm$ is a term, and
$\typV$ is a type.
\defn{Typing rules} of system $\systemV$ are given in \autoref{fig:systemV}.

\begin{figure}[t]
  \begin{mathpar}
    \inferrule*[right = $\ruleTypVarV$]{
    }{
      \judgs{\var:\nityp}\var\nityp
    }

    \inferrule* [right = $\ruleTypAbsV$]{
      \left(
        \judgs{\tctx_i;\var:\nityp_i}\tm{\typV_i}
      \right)_\iI
    }{
      \judgs{+_\iI\tctx_i}{\lam\var\tm}{\mset{\nityp_i\to\typV_i}_\iI}
    }

    \inferrule* [right = $\ruleTypAppV$]{
      \judgs\tctx\tm{\mset{\nityp\to\typV}}
      \\
      \judgs\tctxtwo\tmtwo\nityp
    }{
      \judgs{\tctx + \tctxtwo}{\tm\,\tmtwo}\typV
    }
    
    \inferrule* [right = $\ruleTypESV$]{
      \judgs{\tctx;\var:\nityp}\tm\typV
      \\
      \judgs\tctxtwo\tmtwo\nityp
    }{
      \judgs{\tctx + \tctxtwo}{\tm\esub\var\tmtwo}\typV
    }
  \end{mathpar}
  \caption{Typing Rules of System $\systemV$}
  \label{fig:systemV}
\end{figure}

A \textbf{(type) derivation} is a tree built by applying the  (inductive) typing rules in~\autoref{fig:systemV}.
We write $\derivablep\deriv{\judgs\tctx\tm\typV}\systemV$ to denote 
a type derivation $\deriv$ in system $\systemV$ ending in the 
typing judgement $\judgs\tctx\tm\typV$. 
Sometimes, $\deriv$ or $\systemV$ may be omitted.
We say that a term $\tm$ is \defn{typable} in system $\systemV$ (or \defn{$\mathbf{\systemV}$-typable})
if there exists a typing derivation for $\tm$.

Multi-types in the conclusion of rule $\ruleTypAbsV$ reflect that \VSC-reduction does not proceed under abstractions.
For example, the derivation
$\derivable{}{\judgs\emptyctx{\lam\var\Omega}\emset}$ typing
the \VSC-normal form $\lam\var\Omega$ holds even if $\Omega$ is a
non-terminating subterm. 
Notice that rule $\ruleTypESV$ can be easily reconstructed from rules $\ruleTypAbsV$ and $\ruleTypAppV$.

Let us now consider a typing derivation for the term $\tm$ of our working~\autoref{exa:working-example}.
\begin{exa}
\label{ex:derivation-V}
Let $\nityp \eqdef \mset{\emset\to\emset}$, then we build the following derivation $\deriv$ for $\tm$:
\begin{mathpar}
  \inferrule*[]{
    \inferrule*[]{
      \inferrule*[]{
      }{
        \judgs
          {\var : \nityp}
          \var
          \nityp
      }{}
      \HS
      \inferrule*[]{
      }{
        \judgs
          {\varthree : \emset}
          {\varthree}
          \emset
      }{}
    }{
      \judgs
        {\varthree : \emset;\var:\nityp}
        {\var\,\varthree}
        \emset
    }{}
    \HS
    \inferrule*[]{
      \inferrule*[]{
      }{
        \judgs
          {\vartwo : \nityp}
          {\vartwo}
          \nityp
      }{}
      \HS
      \inferrule*[]{
        \inferrule*[]
          { }
          {\inferrule{}{\judgs{\varfour : \emset}\varfour\emset}}
      }{
        \judgs
          {\emptyctx}
          {\id}
          {\nityp}
      }{}
    }{
      \judgs
        {\emptyctx}
        {\vartwo\esub\vartwo\id}
        {\nityp}
    }{}
  }{
    \judgs
      {\varthree : \emset}
      {(\var \, \varthree)\esub\var{\vartwo\esub\vartwo\id}}
      \emset
  }{}
\end{mathpar}
\end{exa}

Next, let us focus on the quantitative information we can retrieve from $\systemV$. 
We measure the \defn{size of typing derivations} in system $\systemV$
using the function $\sizeV{\_}$, as defined in~\autoref{fig:sz_V}.

\begin{figure}[ht]
  \[
    \begin{aligned}
      \sizeV{
        \inferrule{
        }{
          \judgs{\var:\nityp}\var{\nityp} 
        }\ruleTypVarV
      } 
      & \eqdef \length\nityp
    \\
      \sizeV{
        \inferrule{
          (\deriv_i)_\iI
        }{
          \judgs\tctx{\lam\var\tm}{\mset{\typV_i}_\iI}
        }\ruleTypAbsV
      }
      & \eqdef \length{I} + \sum_\iI \sizeV{\deriv_i} 
    \\
      \sizeV{
        \inferrule{
          \deriv_1
          \HS
          \deriv_2
        }{
          \judgs\tctx{\tm\,\tmtwo}{\typV}
        }\ruleTypAppV
      }
     & \eqdef 1 + \sizeV{\deriv_1} + \sizeV{\deriv_2}
    \\
      \sizeV{
        \inferrule{
          \deriv_1
          \HS
          \deriv_2
        }{
          \judgs\tctx{\tm\esub\var\tmtwo}{\typV}
        }\ruleTypESV
      }
      & \eqdef
        \begin{cases*}
          1 + \sizeV{\deriv_1} + \sizeV{\deriv_2} & if $\var \in \fv{\tm}$
        \\
          \sizeV{\deriv_1} + \sizeV{\deriv_2}     & otherwise
        \end{cases*}
    \end{aligned}
  \]
  \caption{Size of Typing Derivations in System $\systemV$}
  \label{fig:sz_V}
\end{figure}

For example, the
size of the derivation $\deriv$ in~\autoref{ex:derivation-V} is $6$.
\medskip

We can now recall some properties of the type system $\systemV$ that will come in handy throughout the paper.
A first important property of $\systemV$ is that any typing derivation of a value can be split into smaller derivations, while several derivations of the same value can be merged into a single larger one:
\begin{lem}[Splitting/Merging for System $\systemV$]
\label{lem:splitting-merging-V}
Let $\nityp =\ \niunion_{\iI} \nityp_i$. The following statements are equivalent:
\begin{enumerate}
\item 
  $\derivablep\deriv{\judgs\tctx\val\nityp}\systemV$
\item 
  There exist typing derivations $(\deriv_i)_{\iI}$
  and typing environments $(\tctx_i)_{\iI}$ such that
  $\derivablep{\deriv_i}{\judgs{\tctx_i}\val{\nityp_i}}\systemV$
  and $\tctx = +_{\iI} \tctx_i$. 
\end{enumerate}
Moreover, $\sizeV{\deriv} = \sum_{\iI}\ \sizeV{\deriv_i} $.
\end{lem}
\begin{proof}
Both directions of the proof are straightforward.
\end{proof}

A second important property of $\systemV$, which will be used repeatedly throughout the paper, is that it is a \emph{relevant} type system:
\begin{lem}[Relevance]
\label{lem:relevance_systemV}
If $\derivablep\deriv{\judgs\tctx\tm\typV}\systemV$, then $\dom{\tctx} \subseteq \fv{\tm}$.
\end{lem}
\begin{proof}
The proof is by induction on $\deriv$ and is straightforward.
\end{proof}

We now recall the main property of system $\systemV$, namely 
that $\systemV$-typability characterises \VSC-termination:
\begin{thmC}[{\cite[Theorem~4.24]{BucciarelliKRV20}}]
\label{t:charact-V-VSC}
  Let $\tm$ be a term. Then $\tm$ is $\systemV$-typable if and only
  if $\tm$ is \VSC-terminating.
\end{thmC}

In Section~\ref{sec:non-erasing}, we also show that
$\systemV$-typability characterises termination in
the Non-Erasing Value Substitution Calculus (\VSCne), a
calculus introduced in that same section as an
intermediate calculus between \VSC and \UOCBV. 
This second characterisation result (\autoref{cor:v-characterises-vscne}), 
together with~\autoref{t:charact-V-VSC}, yield an observational equivalence property between \VSC
and \VSCne through the type system $\systemV$, following the
methodology in~\cite{Kesner16}, as already discussed
in the introduction.

\section{The \texorpdfstring{\VSCne}{VSCne} Calculus}
\label{sec:non-erasing}

In this section, we introduce the \defn{Non-Erasing Value Substitution
Calculus} (\VSCne). 
As its name suggests, \VSCne modifies the \VSC substitution mechanism
so as to prevent erasing steps from being fired. 
For instance, the \VSC-step $(\var\,\var)\esub\vartwo\val \tosv \var\,\var$ is an
erasing step for any value $\val$, so that it will be ruled out in \VSCne. 
The \VSCne serves as an intermediate calculus between \VSC
and \UOCBV, as it modifies the operational semantics of \VSC without adding the full
machinery that useful evaluation requires to ensure computation progress. 
Indeed, erasing steps are not useful, as they do not contribute to the creation of $\dbsym$-redexes.

The section is organised as follows. 
In \autoref{subsec:vscne-op}, we define the \VSCne operational semantics 
and provide a characterisation of \VSCne-normal forms (\autoref{prop:vscne_characterisation_nf}).
Then in \autoref{subsec:vscne_V}, we show that $\systemV$-typability characterises \VSCne-termination (\autoref{cor:v-characterises-vscne}).

\subsection{Operational Semantics of \texorpdfstring{\VSCne}{VSCne}}
\label{subsec:vscne-op}
The \VSCne syntax is identical to that of \VSC, but its reduction relation differs. 
Indeed, it is defined by the following reduction rules, closed by arbitrary weak contexts:
\[
  \begin{array}{rlll}
      (\lam\var\tm)\sctx \, \tmtwo
    & \todb 
    & \tm\esub\var\tmtwo\sctx
  \\
      \tm\esub\var{\val\sctx}
    & \tosvne
    & \tm\sub\var\val\esub\var\val\sctx
    & \sep \text{ if $\var \in \fv{\tm}$}
  \end{array}
\]

Note that the $\dbsym$-rule is the same as the one used in \VSC. 
The $\svnesym$-rule modifies the $\svsym$-rule of \VSC 
by restricting substitution only to the cases where the variable $\var$ occurs free in $\tm$.
We assume the $\svnesym$-rule to be capture-free due to
$\alpha$-renaming (that is, $\var \notin \fv{\val}$ and the set
$\fv{\tm}$ is not affected by $\sctx$). 
For example, $\var\esub\vartwo\varthree$ reduces to $\var$ in \VSC but is
already a \VSCne-normal form. 
Additionally, the $\svnesym$-rule keeps track of the substitutions being performed. 
For instance, the term $(\var\,\var)\esub\var{\vartwo\esub\vartwo\varthree}$
$\svnesym$-reduces to
$(\vartwo\,\vartwo)\esub\var\vartwo\esub\vartwo\varthree$, thus tracking $\esub \var \vartwo$. 
This explicit tracing of substitutions plays an active role in calculi with \emph{linear} substitution (\cf~Section~\ref{sec:ucbv}).

\paragraph{\textbf{Characterisation of \VSCne-Normal Forms.}}
To characterise the set of \VSCne-normal forms, we start by defining the following grammar of terms:
\[ 
  \begin{array}{lll}
      \vvscne 
    & \eqgram 
    & \var \mid \vvscne \esub\var\nevscne \mid \vvscne \esub\var\novscne\ (\var \notin \fv\vvscne)
  \\
      \nevscne 
    & \eqgram 
    & \vvscne \, \novscne \mid \nevscne \, \novscne \mid \nevscne\esub\var\nevscne \mid \nevscne\esub\var\novscne\ (\var \notin \fv\nevscne)
  \\
      \absvscne 
    & \eqgram 
    & \lam \var \tm \mid \absvscne \esub\var\nevscne \mid \absvscne \esub\var\novscne\ (\var \notin \fv\absvscne)
  \\
      \novscne 
    & \eqgram 
    & \vvscne \mid \nevscne \mid \absvscne
  \end{array} 
\]

We define three predicates $\absPred{\_}$, $\appPred{\_}$, and
$\varPred{\_}$ needed in the forthcoming proof. 
For any term $\tm$, the predicate $\absPred\tm$ holds
if and only if $\tm = (\lam\var\tmtwo)\sctx$ for some variable $\var$, 
some term $\tmtwo$ and some substitution context $\sctx$; the 
predicate $\appPred\tm$ holds if and only if
$\tm = (\tmtwo \, \tmthree)\sctx$ for some terms $\tmtwo, \tmthree$ 
and some substitution context $\sctx$; and the predicate $\varPred\tm$
holds if and only if $\tm = \var\sctx$ for some variable
$\var$ and some substitution context $\sctx$.

\begin{prop}[Characterisation of \VSCne-Normal Forms]
\label{prop:vscne_characterisation_nf}
Let $\tm$ be a term.
Then, $\tm \in \novscne$ if and only if $\tm$ is \VSCne-irreducible.
\end{prop}

\begin{proof}
  The proof is by simultaneous
  induction on the following refined statements:
\begin{enumerate}
\item \label{it:ch_vr}
  $\tm \in \vvscne$ if and only if $\tm$ is \VSCne-irreducible and
  $\varPred{\tm}$
\item \label{it:ch_ne}
  $\tm \in \nevscne$ if and only if $\tm$ is \VSCne-irreducible and
  $\appPred{\tm}$
\item \label{it:ch_ab}
  $\tm \in \absvscne$ if and only if $\tm$ is \VSCne-irreducible and $\absPred\tm$
\item \label{it:ch_no}
  $\tm \in \novscne$ if and only if $\tm$ is \VSCne-irreducible \qedhere
\end{enumerate}
\end{proof}

\subsection{Characterisation of \texorpdfstring{\VSCne}{VSCne}-Termination through \texorpdfstring{$\systemV$}{V}-Typability}
\label{subsec:vscne_V}

In this subsection, we show that $\systemV$-typability characterises \VSCne-termination.
To do so, we prove soundness and completeness of the type system $\systemV$ with respect to \VSCne evaluation.
Both results depend, respectively, on demonstrating that \VSCne-reduction and \VSCne-expansion preserve $\systemV$-typability.

\paragraph{\textbf{Soundness.}}
We start by addressing the \defn{soundness} of system $\systemV$ with 
respect to \VSCne, which states that 
$\systemV$-typability implies \VSCne-termination. 
The soundness result requires a subject reduction property (\autoref{prop:vscne_wsr_V}), which is based in turn
on a substitution lemma (\autoref{lem:vscne_substitution_V}).

The forthcoming substitution lemma indicates that the meta-substitution operation does not alter $\systemV$-typability.
More precisely, given a typing derivation for a term $\tm$, performing a meta-substitution on $\tm$ yields a term with the same type, possible under a modified typing environment.
In addition, the lemma gives a quantitative bound on the size of the resulting typing derivation in terms of the sizes of the typing derivations of the terms involved in the meta-substitution.

\begin{lem}[Weighted Substitution]
\label{lem:vscne_substitution_V}
Let $\derivablep{\deriv_\tm}{\judgs{\tctx;\var:\nityp}\tm\typV}{\systemV}$
and $\derivablep{\deriv_\val}{\judgs\tctxtwo\val\nityp}{\systemV}$. 
Then, there exists a typing derivation $\deriv_{\tm\sub\var\val}$ 
such that 
$\derivablep{\deriv_{\tm\sub\var\val}}{\judgs{\tctx + \tctxtwo}{\tm\sub\var\val}\typV}{\systemV}$
and moreover
$\sizeV{\deriv_{\tm\sub\var\val}} = \sizeV{\deriv_\tm} + \sizeV{\deriv_\val} - \length\nityp$.
\end{lem}
\begin{proof}
See~\cite[Lemma 4.18]{BucciarelliKRV23}.
\end{proof}

\begin{exa}
\label{ex:substitution-V}
Consider the subderivation 
$\derivable
  {\deriv_{\var\,\varthree}}
  {\judgs
    {\varthree : \emset; \var:\nityp}
    {\var \, \varthree}
    \emset}$ 
within derivation $\deriv$ in \autoref{ex:derivation-V}.
Consider also the subderivation 
$\derivable{\deriv_{\vartwo}}{\judgs{\vartwo : \nityp}\vartwo\nityp}$. 
Then, the previous lemma guarantees the existence of a derivation 
$\derivable
  {\deriv_{\vartwo \, \varthree}}
  {\judgs
    {\varthree : \emset; \vartwo : \nityp}
    {\vartwo \, \varthree}
    \emset}$. 
Moreover, we have 
$\sizeV {\deriv_{\vartwo \, \varthree}} 
= 2 
= 2+1 - 1 
= \sizeV{\deriv_{\var \, \varthree}} + \sizeV{\deriv_{\vartwo}} - \length \nityp$.
\end{exa}

We now move on to the subject reduction property, which establishes that \VSCne-reduction preserves $\systemV$-typability.
Moreover, this result also provides quantitative information,
and therefore it is called \emph{weighted subject reduction}, 
as it proves that the size of typing derivations strictly decreases with respect to reductions.

\begin{prop}[Weighted Subject Reduction]
\label{prop:vscne_wsr_V}
Let $\derivablep\deriv{\judgs\tctx\tm\typV}\systemV$ and 
$\tm \toVSCne \tm'$. 
Then, there exists a derivation $\deriv'$ such that 
$\derivablep{\deriv'}{\judgs\tctx{\tm'}\typV}\systemV$ and 
$\sizeV\deriv > \sizeV{\deriv'}$.
\end{prop}

\begin{proof}
We reason by induction on $\tm \toVSCne \tm'$. 
We only detail here the base case $\tm \tosvne \tm'$, as the inductive cases are straightforward.

Let $\tm = \tmtwo\esub\var{\val\sctx}$ and 
$\tm' = \tmtwo\sub\var\val\esub\var\val\sctx$, where $\var \in \fv\tmtwo$. 
We proceed by induction on the substitution context $\sctx$.
\begin{itemize}
\item $\sctx = \ctxhole$.
  Then, $\deriv$ has the form:
  \begin{mathpar}
    \inferrule* [right = $\ruleTypESV$]
      {\derivable{\deriv_\tmtwo}{\judgs{\tctx_1;\var : \nityp}\tmtwo\typV}
        \\ \derivable{\deriv_\val}{\judgs{\tctx_2}\val\nityp}
      }
      {\judgs{\tctx_1 + \tctx_2}{\tmtwo\esub\var\val}\typV}
  \end{mathpar}
  where $\tctx = \tctx_1 + \tctx_2$. 

  We can apply \autoref{lem:vscne_substitution_V} on $\deriv_\tmtwo$
  and $\deriv_\val$, yielding
  $\derivable
    {\deriv_{\tmtwo\sub\var\val}}
    {\judgs{\tctx_1 + \tctx_2}{\tmtwo\sub\var\val}\typV}$
  and $\sizeV{\deriv_{\tmtwo\sub\var\val}} 
  = \sizeV{\deriv_\tmtwo} + \sizeV{\deriv_\val} - \length\nityp$.

  Moreover, since $\var \notin \fv{\tmtwo\sub\var\val}$, then
  $\var \notin\dom{\tctx_1 + \tctx_2}$ by~\autoref{lem:relevance_systemV}. 
  We can then write 
  $\tctx_1 + \tctx_2$ as $(\tctx_1 + \tctx_2); \var : \emset$. 

  And it is straightforward to notice that for any value $\val$,
  we can build the typing derivation 
  $\derivable
    {\deriv'_\val}
    {\judgs\emptyctx\val\emset}$,
  with size $0$.
  Thus, we build the following derivation $\deriv'$:
  \begin{mathpar}
    \inferrule* [right = $\ruleTypESV$]
      {\derivable
        {\deriv_{\tmtwo\sub\var\val}}
        {\judgs{(\tctx_1 + \tctx_2);\var:\emset}{\tmtwo\sub\var\val}\typV}
        \\ \derivable{\deriv'_\val}{\judgs\emptyctx\val\emset}
      }
      {\judgs{\tctx_1 + \tctx_2}{\tmtwo\sub\var\val\esub\var\val}\typV}
  \end{mathpar}

  And we conclude that
  \[
    \begin{array}{rll}
        \sizeV\deriv 
      & = 1 + \sizeV{\deriv_\tmtwo} + \sizeV{\deriv_\val}
      & (\var \in \fv\tmtwo)
    \\
      & > \sizeV{\deriv_\tmtwo} + \sizeV{\deriv_\val} - \length\nityp
    \\
      & = \sizeV{\deriv_{\tmtwo\sub\var\val}}
    \\
      & = \sizeV{\deriv_{\tmtwo\sub\var\val}} + \sizeV{\deriv'_\val}
    \\
      & = \sizeV{\deriv'}
      & (\var \notin \fv{\tmtwo\sub\var\val})
    \end{array}
  \]
\item $\sctx = \sctxtwo\esub\vartwo\tmthree$.
  Then, $\deriv$ has the form:
  \begin{mathpar}
    \inferrule* [right = $\ruleTypESV$]
      {\derivable{\deriv_\tmtwo}{\judgs{\tctx';\var:\nityp}\tmtwo\typV}
        \\ \inferrule*[Right = $\ruleTypESV$]
              {\derivable{\deriv_{\val\sctxtwo}}{\judgs{\tctxtwo_1;\vartwo:\nityptwo}{\val\sctxtwo}\nityp} 
                \\ \derivable{\deriv_\tmthree}{\judgs{\tctxtwo_2}\tmthree\nityptwo}
              }
              {\judgs{\tctxtwo_1 + \tctxtwo_2}{\val\sctxtwo\esub\vartwo\tmthree}\nityp}
      }
      {\judgs
          {\tctx' + \tctxtwo_1 + \tctxtwo_2}
          {\tmtwo\esub\var{\val\sctxtwo\esub\vartwo\tmthree}}
          \typV 
      }
  \end{mathpar}
  where $\tctx = \tctx' + \tctxtwo_1 + \tctxtwo_2$.
  Moreover, we can assume $\vartwo \notin \fv\tmtwo$ by $\alpha$-conversion.
  Thus, we assume $\vartwo \notin \dom{\tctx'}$ by~\autoref{lem:relevance_systemV}.

  We now  build the following derivation $\derivtwo$:
  \begin{mathpar}
    \inferrule* [right = $\ruleTypESV$]
      {\derivable{\deriv_\tmtwo}{\judgs{\tctx';\var : \nityp}\tmtwo\typV}
        \\ \derivable{\deriv_{\val\sctxtwo}}
            {\judgs{\tctxtwo_1;\vartwo:\nityptwo}{\val\sctxtwo}\nityp}
      }
      {\judgs{\tctx' + (\tctxtwo_1;\vartwo:\nityptwo)}{\tmtwo\esub\var{\val\sctxtwo}}\typV}
  \end{mathpar}
  Moreover, $\tmtwo\esub\var{\val\sctxtwo} \tosvne \tmtwo\sub\var\val\esub\var\val\sctxtwo$,
  since $\var \in \fv\tmtwo$.
  Then, we can apply the \ih, yielding $\derivtwo'$ such that
  $\derivable{\derivtwo'}
    {\judgs{\tctx' + (\tctxtwo_1;\vartwo:\nityptwo)}{\tmtwo\sub\var\val\esub\var\val\sctxtwo}\typV}$
  and moreover, $\sizeV\derivtwo > \sizeV{\derivtwo'}$. 

  Since $\vartwo \notin \dom{\tctx'}$, we can write
  $\tctx' + (\tctxtwo_1;\vartwo:\nityptwo)$ as
  $(\tctx' + \tctxtwo_1);\vartwo:\nityptwo$.
  We build $\deriv'$ as follows: 
  \begin{mathpar}
    \inferrule* [right = $\ruleTypESV$]
      {\derivable{\derivtwo'}
        {\judgs{(\tctx' + \tctxtwo_1); \vartwo:\nityptwo}
          {\tmtwo\sub\var\val\esub\var\val\sctxtwo}\typV}
        \\
      \derivable{\deriv_\tmthree}{\judgs{\tctxtwo_2}\tmthree\nityptwo}
      }
      {\judgs{\tctx' + \tctxtwo_1 + \tctxtwo_2}
        {\tmtwo\sub\var\val\esub\var\val\sctxtwo\esub\vartwo\tmthree}\typV
      }
  \end{mathpar}

  To compare the sizes of both $\deriv$ and $\deriv'$, 
  we consider two cases, depending on whether $\vartwo \in \fv{\tmtwo\sub\var\val\esub\var\val\sctxtwo}$
  or not.
  Recall that $\var \in \fv\tmtwo$ by hypothesis.
  \begin{itemize}
  \item 
    If $\vartwo \notin \fv{\tmtwo\sub\var\val\esub\var\val\sctxtwo}$, 
    then in particular
    $\vartwo \notin \fv{\val\sctxtwo}$.
    Thus, we have
    \[
      \begin{array}{rll}
        \sizeV\deriv 
        & = 1 + \sizeV{\deriv_\tmtwo} + \sizeV{\deriv_{\val\sctxtwo}} + \sizeV{\deriv_\tmthree}
        & (\var \in \fv\tmtwo, \vartwo \notin \fv{\val\sctxtwo})
      \\
        & = \sizeV{\derivtwo} + \sizeV{\deriv_\tmthree}
      \\
        & > \sizeV{\derivtwo'} + \sizeV{\deriv_\tmthree}
        & (\ih)
      \\
        & = \sizeV{\deriv'}
        & (\vartwo \notin \fv{\tmtwo\sub\var\val\esub\var\val\sctxtwo})
      \end{array}
    \]
  \item
    Otherwise, $\vartwo \in \fv{\tmtwo\sub\var\val\esub\var\val\sctxtwo}$,
    and recall that $\vartwo \notin \fv\tmtwo$ by $\alpha$-conversion,
    so it must be the case that $\vartwo \in \fv{\val\sctxtwo}$.
    Therefore,
    \[
      \begin{array}{rllr}
          \sizeV\deriv 
        & = 1 + \sizeV{\deriv_\tmtwo} + 1 + \sizeV{\deriv_{\val\sctxtwo}} + \sizeV{\deriv_\tmthree}
        & (\var \in \fv\tmtwo, \vartwo \in \fv{\val\sctxtwo})
      \\ 
        & = 1 + \sizeV\derivtwo + \sizeV{\deriv_\tmthree}
      \\ 
        & > 1 + \sizeV{\derivtwo'} + \sizeV{\deriv_\tmthree}
        & (\ih)
      \\
        & = \sizeV{\deriv'}
        & (\vartwo \in \fv{\tmtwo\sub\var\val\esub\var\val\sctxtwo})
        & \mbox{\qedhere}
      \end{array}
    \]
  \end{itemize}
\end{itemize}
\end{proof}

\begin{exa}
Consider the subderivation 
$\derivable
  {\deriv_{\var \, \varthree}}
  {\judgs
    {\varthree : \emset; \var:\nityp}
    {\var \, \varthree}
    \emset}$ 
within the derivation $\deriv$ in \autoref{ex:derivation-V}.
Given that
$\tm$ reduces in one $\svnesym$-step to $\tm' \eqdef
(\var \, \varthree)\sub\var\vartwo\esub\var\vartwo\esub\vartwo\id = (\vartwo \, \varthree)\esub\var\vartwo\esub\vartwo\id$,
and given the derivation 
$\derivable
  {\deriv_{\vartwo \, \varthree}}
  {\judgs
    {\varthree : \emset; \vartwo : \nityp}
    {\vartwo \, \varthree}
    \emset}$ 
in ~\autoref{ex:substitution-V}, where in particular
$\varthree : \emset; \vartwo : \nityp= \varthree : \emset; \vartwo : \nityp; \var:\emset $, 
we build a typing derivation $\deriv'$ for $\tm'$ that preserves the type of $\tm$:
\begin{mathpar} \inferrule*[]{ \inferrule*[]{\derivable{\deriv_{\vartwo \, \varthree}}{\judgs
      {\varthree : \emset; \vartwo : \nityp; \var:\emset}
      {\vartwo \, \varthree}
      \emset} \\
  \inferrule*[]{
        }{
          \judgs{\vartwo : \emset}\vartwo\emset
        } }{ \judgs
  {\varthree : \emset; \vartwo:\nityp}
  {(\vartwo\,\varthree)\esub\var\vartwo} \emset }{} \\ \inferrule*[]{
          \inferrule*[]{
          }{
            \inferrule{}{\judgs{\varfour : \emset}\varfour\emset}
          }
        }{
          \judgs\emptyctx\id\nityp
        } 
  }{ \judgs {\varthree : \emset}
  {(\vartwo \, \varthree)\esub\var\vartwo\esub\vartwo\id} \emset
  }{}
\end{mathpar}

Note that the size of $\deriv'$ is $4$, which is strictly smaller than $6$. 
\end{exa}

Since we know that \VSCne-reduction preserves $\systemV$-typability thanks to the subject reduction property, we can now state
the following:
\begin{thm}[Soundness of $\systemV$ with Respect to \VSCne]
\label{thm:soundness_V_vscne}
Let t be a term.
If $\tm$ is $\systemV$-typable, then $\tm$ is \VSCne-terminating.
\end{thm}
\begin{proof}
Let $\derivable\deriv{\judgs\tctx\tm\typV}$,
and assume $\tm$ is not \VSCne-terminating. Then, 
there exists an infinite reduction sequence
\[
  \tm \toVSCne \tm_1 \toVSCne \tm_2 \toVSCne \hdots
\]

We can apply~\autoref{prop:vscne_wsr_V} on each reduction step
in the sequence, thus obtaining 
\[
\begin{array}{c}
  \derivable{\deriv_1}{\judgs\tctx{\tm_1}\typV} 
\\
  \derivable{\deriv_2}{\judgs\tctx{\tm_2}\typV} 
\\
  \hdots
\end{array}
\]
such that $\sizeV\deriv > \sizeV{\deriv_1} > \sizeV{\deriv_2} > \hdots$.

Then, there is necessarily a step $n$ such that $\sizeV{\deriv_n}$
is a negative number.
This is a contradiction, as the size function $\sizeV{\_}$ always 
returns numbers greater than or equal to $0$ by definition.

We then conclude that $\tm$ is \VSCne-terminating.
\end{proof}

\paragraph{\textbf{Completeness.}}
We now prove the \defn{completeness} of system $\systemV$ with 
respect to \VSCne, which states that \VSCne-termination 
implies $\systemV$-typability. 
The methodology is dual to that used to prove soundness: 
we require a subject expansion property (\autoref{prop:vscne_wse_V}) based 
on an anti-substitution lemma (\autoref{lem:vscne_anti-subst_V}).
Additionally, it is necessary to prove that all \VSCne-normal forms are $\systemV$-typable:

\begin{prop}[Normal Forms in \VSCne are $\systemV$-Typable]
\label{prop:vscne_normal_forms_are_typable_V}
Let $\tm$ be a \VSCne-normal form.
Then, $\tm$ is $\systemV$-typable.
\end{prop}

\begin{proof}
  The proof proceeds by simultaneous
  induction on the following refined statements:
\begin{enumerate}
\item \label{it:typing_vr_V}
  If $\tm \in \vvscne$, then for every multi-type $\nityp$ there exist $\deriv$ and $\tctx$
  such that $\derivablep\deriv{\judgs\tctx\tm\nityp}\systemV$. 
\item \label{it:typing_ne_V}
  If $\tm \in \nevscne$, then for every type $\typV$ there exist $\deriv$ and $\tctx$ 
  such that $\derivablep\deriv{\judgs\tctx\tm\typV}\systemV$. 
\item \label{it:typing_ab_V}
  If $\tm \in \absvscne$, then there exist $\deriv$ and $\tctx$ 
  such that $\derivablep\deriv{\judgs\tctx\tm\emset}\systemV$. 
\item \label{it:typing_no_V}
  If $\tm \in \novscne$, then there exist $\deriv$, $\tctx$, and $\nityp$ such 
  that $\derivablep\deriv{\judgs\tctx\tm\nityp}\systemV$. 
  \qedhere
\end{enumerate} 
\end{proof}

As its name suggests, the Anti-Substitution Lemma goes in the opposite direction to the Substitution Lemma. 
More precisely, from a typing derivation of a substituted term, it is possible to reconstruct typing derivations for the term prior to the substitution and for the value being substituted. 
Moreover, the sizes of the reconstructed derivations satisfy the same relation as in the Substitution Lemma. Formally:

\begin{lem}[Weighted Anti-Substitution]
\label{lem:vscne_anti-subst_V}
Let 
$\derivablep{\deriv_{\tm\sub\var\val}}{\judgs\tctx{\tm\sub\var\val}\typV}{\systemV}$.
Then there exist typing derivations $\deriv_\tm$ and $\deriv_\val$,
typing environments $\tctx_\tm$ and $\tctxtwo$, and a multi-type $\nityp$ satisfying:
\begin{enumerate}
\item 
  $\derivablep{\deriv_\tm}{\judgs{\tctx_\tm;\var : \nityp}\tm\typV}{\systemV}$
\item
  $\derivablep{\deriv_\val}{\judgs{\tctxtwo}\val\nityp}{\systemV}$
\item
  $\tctx = \tctx_\tm + \tctxtwo$
\end{enumerate}
Moreover, 
$\sizeV{\deriv_{\tm\sub\var\val}} 
= \sizeV{\deriv_\tm} + \sizeV{\deriv_\val} - \length\nityp$.
\end{lem}
\begin{proof}
See~\cite[Lemma 4.21]{BucciarelliKRV23}.
\end{proof}

We now move on to the subject expansion property, which establishes that \VSCne-expansion preserves $\systemV$-typability.
Since this result is also quantitative, it is called \emph{weighted subject expansion}, 
and proves, like weighted subject reduction, that the size of typing derivations strictly decreases with respect to reductions.

\begin{prop}[Weighted Subject Expansion]
\label{prop:vscne_wse_V}
Let $\tm \toVSCne \tm'$ and $\derivablep{\deriv'}{\judgs\tctx{\tm'}\mtyp}\systemV$. 
Then, there exists $\deriv$ such that 
$\derivablep\deriv{\judgs\tctx\tm\mtyp}\systemV$, with
$\sizeV\deriv > \sizeV{\deriv'}$.
\end{prop}

\begin{proof}
We reason by induction on $\tm \toVSCne \tm'$. 
We only detail here the base case $\tm \tosvne \tm'$, as the inductive cases are straightforward.

Let $\tm = \tmtwo\esub\var{\val\sctx} \tosvne 
\tmtwo\sub\var\val\esub\var\val\sctx = \tm'$, where 
$\var \in \fv\tmtwo$. 
We proceed by induction on the substitution context $\sctx$.
\begin{itemize}
\item $\sctx = \ctxhole$.
  Then, we are in the case where $\tm = \tmtwo\esub\var\val$ and 
  $\tm' = \tmtwo\sub\var\val\esub\var\val$. 
  Therefore, $\deriv'$ is of the form:
  \begin{mathpar}
    \inferrule* [right = $\ruleTypESV$]
      {\derivable{\deriv_{\tmtwo\sub\var\val}}{\judgs{\tctx_1;\var : \nityp}{\tmtwo\sub\var\val}\typV}
        \\ \derivable{\deriv^1_\val}{\judgs{\tctx_2}\val\nityp}}
      {\judgs{\tctx_1 + \tctx_2}{\tmtwo\sub\var\val\esub\var\val}\typV}
  \end{mathpar}
  where $\tctx = \tctx_1 + \tctx_2$.

  Moreover, $\var \notin \fv{\tmtwo\sub\var\val}$, therefore
  $\var \notin \dom{\tctx_1; \var : \nityp}$ by~\autoref{lem:relevance_systemV}, so $\nityp = \emset$.
  
  Applying~\autoref{lem:vscne_anti-subst_V}, we obtain typing 
  derivations $\deriv_\tmtwo$ and $\deriv^2_\val$, typing 
  environments $\tctx_\tmtwo$ and $\tctxtwo$, and a multi-type $\nityp'$ 
  such that:
  \begin{enumerate}
  \item 
    $\derivable{\deriv_\tmtwo}{\judgs{\tctx_\tmtwo; \var:\nityp'}\tmtwo\typV}$
  \item 
    $\derivable{\deriv^2_\val}{\judgs\tctxtwo\val\nityp'}$
  \item
    $\tctx_1; \var : \nityp = \tctx_\tmtwo + \tctxtwo$ 
  \end{enumerate}
  And moreover, $\sizeV{\deriv_{\tmtwo\sub\var\val}} = 
  \sizeV{\deriv_\tmtwo} + \sizeV{\deriv^2_\val} - \length{\nityp'}$.

  We apply~\autoref{lem:splitting-merging-V} to $\deriv^1_\val$ 
  and $\deriv^2_\val$, yielding 
  $\derivable
    {\deriv_\val}
    {\judgs{\tctx_2 + \tctxtwo}\val{\nityp + \nityp'}}$,
  where $\nityp + \nityp' = \nityp'$, as 
  $\nityp = \emset$.
  Moreover, 
  $\sizeV{\deriv_\val} = \sizeV{\deriv^1_\val} + \sizeV{\deriv^2_\val}$,
  where $\sizeV{\deriv^1_\val} = 0$ given $\nityp = \emset$.
  Hence, $\sizeV{\deriv_\val} = \sizeV{\deriv^2_\val}$.

  We then build $\deriv$ as follows:
  \begin{mathpar}
    \inferrule
      {\derivable{\deriv_\tmtwo}{\judgs{\tctx_\tmtwo;\var : \nityp'}\tmtwo\typV}
        \\ \derivable{\deriv_\val}{\judgs{\tctx_2 + \tctxtwo}\val{\nityp'}}}
      {\judgs{\tctx_\tmtwo + \tctx_2 + \tctxtwo}{\tmtwo\esub\var\val}\typV}
      \ruleTypESV
  \end{mathpar}

  Moreover,
  $
      \sizeV\deriv 
    =_{\var \in \fv\tmtwo} 1 
      + \sizeV{\deriv_\tmtwo} + \sizeV{\deriv_\val}
    = 
      1 + \sizeV{\deriv_\tmtwo} + \sizeV{\deriv^2_\val}
    > \sizeV{\deriv_\tmtwo} + \sizeV{\deriv^2_\val} - \length{\nityp'}
    = \sizeV{\deriv_{\tmtwo\sub\var\val}}
    = \sizeV{\deriv_{\tmtwo\sub\var\val}} + \sizeV{\deriv^1_\val}
    =_{\var \notin \fv{\tmtwo\sub\var\val}} \sizeV{\deriv'}
  $.
\item $\sctx = \sctxtwo\esub\vartwo\tmthree$.
  Then, we are in the case where 
  $\tm = \tmtwo\esub\var{\val\sctxtwo\esub\vartwo\tmthree}$ and 
  $\tm' = \tmtwo\sub\var\val\esub\var\val\sctxtwo\esub\vartwo\tmthree$. 
  Therefore, $\deriv'$ is of the form:
  \begin{mathpar}
    \inferrule
      {\derivable{\derivtwo'}{\judgs{\tctx_1; \vartwo:\nityp}{\tmtwo\sub\var\val\esub\var\val\sctxtwo}\typV}
        \\ \derivable{\deriv_\tmthree}{\judgs{\tctx_2}\tmthree\nityp}}
      {\judgs{\tctx_1 + \tctx_2}{\tmtwo\sub\var\val\esub\var\val\sctxtwo\esub\vartwo\tmthree}\typV}
      \ruleTypESV
  \end{mathpar}
  where $\tctx = \tctx_1 + \tctx_2$.

  Given that $\tmtwo\esub\var{\val\sctxtwo} \tosvne 
  \tmtwo\sub\var\val\esub\var\val\sctxtwo$ since $\var \in \fv\tmtwo$,
  we can then apply the \ih on $\sctxtwo$, yielding $\derivtwo$ such 
  that $\derivable
    \derivtwo
    {\judgs{\tctx_1;\vartwo : \nityp}{\tmtwo\esub\var{\val\sctxtwo}}\typV}$
  and moreover, $\sizeV\derivtwo > \sizeV{\derivtwo'}$.
  The conclusion in $\derivtwo$ can only be derived by rule $\ruleTypESV$:
  \begin{mathpar}
    \inferrule
      {\derivable{\deriv_\tmtwo}{\judgs{\tctx_{11}; \var : \nityptwo}\tmtwo\typV}
        \\ \derivable{\deriv_{\val\sctxtwo}}{\judgs{\tctx_{12};\vartwo:\nityp}{\val\sctxtwo}\nityptwo}}
      {\judgs{(\tctx_{11} + \tctx_{12});\vartwo:\nityp}{\tmtwo\esub\var{\val\sctxtwo}}\typV}
      \ruleTypESV
  \end{mathpar}
  where 
  $\tctx_1 = \tctx_{11} + \tctx_{12}$ and 
  $\vartwo \notin \fv\tmtwo$ by $\alpha$-conversion, so 
  $\vartwo \notin \dom{\tctx_{11}}$ by~\autoref{lem:relevance_systemV}. 

  Then, we build $\deriv$ as follows: 
  \begin{mathpar}
    \inferrule
      {\deriv_\tmtwo
        \\ 
          \inferrule*[Right = $\ruleTypESV$]
            {\deriv_{\val\sctxtwo}
              \\ \deriv_\tmthree}
            {\judgs{\tctx_{12} + \tctx_2}{\val\sctxtwo\esub\vartwo\tmthree}\nityptwo}}
      {\judgs{\tctx_{11} + \tctx_{12} + \tctx_2}{\tmtwo\esub\var{\val\sctxtwo\esub\vartwo\tmthree}}\typV} 
    \ruleTypESV
  \end{mathpar}
  where $\tctx_{11} + \tctx_{12} + \tctx_2 = 
  \tctx_1 + \tctx_2 = \tctx$.

  Moreover, to calculate the size of $\deriv$, we consider whether 
  $\vartwo \in \fv{\tmtwo\sub\var\val\esub\var\val\sctxtwo}$.
  We show only when 
  $\vartwo \in \fv{\tmtwo\sub\var\val\esub\var\val\sctxtwo}$, since 
  the other case is analogous.
  In particular, we have $\vartwo \notin \fv\tmtwo$ by 
  $\alpha$-conversion, thus (2) $\vartwo \in \fv{\val\sctxtwo}$, and 
  therefore
  \[
    \begin{array}{rllr}
        \sizeV\deriv 
      & = 1 + \sizeV{\deriv_\tmtwo} + 1 + \sizeV{\deriv_{\val\sctxtwo}} + \sizeV{\deriv_\tmthree}
      & (\text{by } \var \in \fv\tmtwo \text{ and (2)})
    \\
      & = 1 + \sizeV\derivtwo + \sizeV{\deriv_\tmthree}
      & (\var \in \fv\tmtwo)
    \\
      & > 1 + \sizeV{\derivtwo'} + \sizeV{\deriv_\tmthree}
      & (\ih)
    \\
      & = \sizeV{\deriv'}
      & (\vartwo \in \fv{\val\sctxtwo})
      & \mbox{\qedhere}
    \end{array}
  \]
\end{itemize}
\end{proof}

\medskip

We can now state the following result:
\begin{thm}[Completeness of $\systemV$ with Respect to \VSCne]
\label{thm:completeness_V_vscne}
Let t be a term.
If $\tm$ is \VSCne-terminating, then $\tm$ is $\systemV$-typable.
\end{thm}
\begin{proof}
Since $\tm$ is \VSCne-terminating, then there exists a finite sequence
from $\tm$ to normal form of length $n$:
\[
  \tm = \tm_0 \toVSCne \tm_1 \toVSCne \tm_2 \toVSCne \hdots \toVSCne \tm_n
\]
where $\tm_n$ is a \VSCne-normal form.
We proceed by induction on $n$.
\begin{itemize}
\item Base case: $n = 0$. 
  Then, $\tm$ is \VSCne-irreducible, which is equivalent to saying that
  $\tm$ is a \VSCne-normal form.
  By~\autoref{prop:vscne_normal_forms_are_typable_V},
  we conclude that $\tm$ is $\systemV$-typable. 
\item Inductive case: $n > 0$, assuming the statement holds for $n-1$. 
  Taking the reduction sequence of $(n-1)$ steps from $\tm_1$ to the 
  \VSCne-normal form $\tm_n$, we can apply the \ih, so that $\tm_1$ 
  is $\systemV$-typable.
  Since $\tm \toVSCne \tm_1$, then $\tm$ is $\systemV$-typable as well
  by~\autoref{prop:vscne_wse_V}, so we are done.
\qedhere
\end{itemize}
\end{proof}

From soundness (\autoref{thm:soundness_V_vscne}) and
completeness (\autoref{thm:completeness_V_vscne}), we obtain:
\begin{cor}[System $\systemV$ Characterises \VSCne-Termination]
\label{cor:v-characterises-vscne}
Let $\tm$ be a term.
Then, $\tm$ is $\systemV$-typable if and only if $\tm$ is \VSCne-terminating. 
\end{cor}

This last property, together with~\autoref{t:charact-V-VSC}, provide
the following equivalence:
\begin{cor}[\VSC-Termination is equivalent to \VSCne-Termination]
\label{cor:v-equiv-vscne}
Let $\tm$ be a term.
Then, $\tm$ is \VSC-terminating if and only if $\tm$ is \VSCne-terminating.
\end{cor}

\section{The \texorpdfstring{\UOCBV}{UCBV} Calculus}
\label{sec:ucbv}

This section recalls \UOCBV, the inductive formulation of the useful
open \CBV strategy which provides a well-suited basis for semantic
reasoning. 
We only introduce the \UOCBV
features relevant to this paper's results; a detailed presentation can be found
in~\cite{BKM-CADE25,BKM25}.

We begin by introducing the key notions of useful evaluation. 
Then we recall the operational semantics of \UOCBV and show some
rewriting properties of \UOCBV, such as the inductive
characterisation of \UOCBV-normal forms
(\autoref{thm:uocbv_characterisation_of_normal_forms}) and the diamond
property (\autoref{thm:diamond-property}). 
The former ensures in particular that the length of any \UOCBV reduction
sequence to normal form is independent of the specific order in which reduction rules are applied.

\paragraph{\textbf{Useful Evaluation.}} 
The calculus with ESs used to implement sharing in Accattoli and Dal 
Lago's work~\cite{AccattoliL14} is based on a previous ES calculus 
with reduction at a distance called the \emph{Linear Substitution 
Calculus} (\LSC)~\cite{AccattoliK10}.
In \LSC, the reduction rule that implements function application is 
the same $\dbsym$-rule as in \VSC (\cf~\autoref{subsec:vsc}).
Additionally, \LSC uses a \defn{linear substitution} rule, which only
allows the substitution of one \emph{single} occurrence 
of a variable $\var$ by a term $\tm$ whenever $\var$ is bound to 
$\tm$ by an explicit substitution (ES). 
More precisely, the linear substitution rule ($\lssym$) of \LSC, 
allowing substitutions one occurrence at a time, is 
defined by the rewriting rule $\ctxof\gctx\var\esub\var\tm \tols 
\ctxof\gctx\tm\esub\var\tm$, where $\gctx$ denotes an arbitrary 
evaluation context. 
Substitution is called \emph{linear} precisely because it replaces variables 
one occurrence at a time, for example 
$(\var\,\var)\esub\var\vartwo \tols
(\var\,\vartwo)\esub\var\vartwo \tols
(\vartwo\,\vartwo)\esub\var\vartwo$.

Useful evaluation is based on two principles: \emph
{sharing structures} and (linearly) \emph{substituting abstractions for progress}.
We briefly recall them below.

\emph{Sharing Structures.} 
  In useful evaluation, a \defn{structure} is a term whose 
  \emph{unfolding} (the result of performing all substitutions) 
  is an application headed by a variable. 
  For example, 
  $(\var\,\var)\esub\var{\vartwo\,\id}\esub\vartwo{\varthree\,\varthree}$ 
  is a structure, as it unfolds to 
  $\varthree\,\varthree\,\id\,(\varthree\,\varthree\,\id)$. 
  Structures must remain \emph{shared} in useful evaluation, as 
  substituting them does not create a $\dbsym$-redex. 
  Thus, a term like 
  $(\var\,\var)\esub\var{\vartwo\,\id}\esub\vartwo{\varthree\,\varthree}$ 
  is a normal form in useful evaluation.

\emph{Substituting Abstractions for Progress}.
  In useful evaluation, abstractions are substituted only if they 
  contribute to creating a $\dbsym$-redex, ensuring progress in the 
  computation. 
  For example, the reduction step 
  $\var\esub\var\id \, \vartwo \tols \id\esub\var\id \, \vartwo$ is 
  useful because substituting $\var$ by $\id$ creates the 
  $\dbsym$-redex $\id\esub\var\id \, \vartwo$.
  A subtler example of a useful step is
  $(\var\,\var)\esub\var\vartwo\esub\vartwo\id \tols
  (\vartwo\,\var)\esub\var\vartwo\esub\vartwo\id$, where the 
  substitution \emph{indirectly} contributes to creating a 
  $\dbsym$-redex in the next step:
  $(\vartwo\,\var)\esub\var\vartwo\esub\vartwo\id \tols
  (\id\,\var)\esub\var\vartwo\esub\vartwo\id$.
  Such indirect steps are precisely those that are challenging to 
  capture inductively.
  On the other hand, steps like $\var\esub\var\id \tols \id\esub\var\id$, 
  and 
  $(\var\,\var)\esub\var\id \tols (\var\, \id)\esub\var\id$ are not 
  useful, as these substitutions do not contribute to create a 
  $\dbsym$-redex. 

The principles just mentioned not only constitute the basis
of useful \CBN~\cite{AccattoliL14} but also that of useful 
\CBV~\cite{AccattoliC15,AccattoliG17}, including \UOCBV~\cite{BKM-CADE25}.

\paragraph{\textbf{Towards Useful Call-by-Value Evaluation.}} 

A key point for understanding the definition of \UOCBV evaluation is that the parameters decorating the reduction relations carry the essential information about the surrounding evaluation contexts in which the reductions take place. More precisely, they retain exactly the contextual information needed to determine which substitution steps are useful and which are not. 

Two of those parameters are sets of variables: one is called an \defn{(abstraction) frame} denoted $\aset$ and the 
second is called a \defn{(structure) frame} denoted $\sset$.
The purpose of these two variable sets 
is to classify the nature of variables bound by ESs.
For example, the variables bound by ESs in the term $(\var \, \varthree)\esub\var\vartwo\esub\varthree{\varfour\esub\varfour\id}$ are $\var$, $\varthree$, and $\varfour$.
The variable $\var$ is bound by the free variable $\vartwo$, which is a structure by definition, and thus $\var \in \sset$.
The variable $\varthree$ is bound by $\varfour$, which in turn is bound by an abstraction and thus $\varthree,\varfour \in \aset$.

We generalise this classification at the level of terms. 
Identifying terms that are already abstractions as well as those that
are currently variables --but will ultimately be substituted with
abstractions-- is necessary to correctly determine whether a substitution step
contributes to the progress of the evaluation.
This property is captured by the set of
\defn{hereditary abstractions under the frame $\aset$}, written $\HAbs\aset$, and defined by means of predicates $\tm \in \HAbs\aset$
through the inference rules in the first line of~\autoref{fig:HAbs_Struct}.
For example,
$\id\esub{\var}{\vartwo} \in \HAbs\emptyset$
holds by the third rule,
while $\var\esub{\var}{\id} \in \HAbs\emptyset$ holds by the fourth rule.
Note that the fourth rule allows building ``chains'' such as
$\var\esub{\var}{\vartwo}\esub{\vartwo}{\varthree\esub{\varthree}{\id}} \in \HAbs\emptyset$.

It is also necessary to define a predicate to identify those terms that represent an 
application headed by a free variable.
Such terms form the set of \defn{structures under the frame $\sset$}, 
written $\Struct\sset$. The predicates $\tm \in \Struct\sset$ are defined by the inference rules given in the second line of~\autoref{fig:HAbs_Struct}.
For example,
$\varfour\esub{\var}{\vartwo} \in \Struct{\set{\varfour}}$
holds by the second rule,
$\varfour\,(\id\,\var) \in \Struct{\set{\varfour}}$ by the third rule,
and $\var\esub{\var}{\varfour} \in \Struct{\set{\varfour}}$ by the fourth rule.

The interested reader can check further discussions and properties of both
sets $\HAbs\aset$ and $\Struct\sset$ in~\cite{BKM-CADE25}.

\begin{figure}[h]
\begin{mathpar}
  \inferrule
    {\var \in \aset}
    {\var \in \HAbs\aset}

  \inferrule
    { }
    {\lam\var\tm \in \HAbs\aset}

  \inferrule
    {\tm \in \HAbs\aset \\ \var \notin \aset}
    {\tm\esub\var\tmtwo \in \HAbs\aset}
  
  \inferrule{
    \tm \in \HAbs{\aset\cup\set\var}
    \\ \var\notin\aset
    \\ \tmtwo \in \HAbs\aset
  }{
    \tm\esub\var\tmtwo \in \HAbs\aset
  }
\end{mathpar}

\begin{mathpar}
  \inferrule
    {\var \in \sset}
    {\var \in \Struct\sset}

  \inferrule
    {\tm \in \Struct\sset \\ \var\notin\sset}
    {\tm\esub\var\tmtwo \in \Struct\sset}

  \inferrule
    {\tm \in \Struct\sset}
    {\tm\,\tmtwo \in \Struct\sset}

  \inferrule{
    \tm \in \Struct{\sset\cup\set\var}
    \\ \var\notin\sset
    \\ \tmtwo \in \Struct\sset
  }{
    \tm\esub\var\tmtwo \in \Struct\sset
  }
\end{mathpar}
  \caption{Inductive rules for Hereditary Abstractions (Top) and Structures (Bottom)}
  \label{fig:HAbs_Struct}
\end{figure}

\paragraph{\bf Useful Call-by-Value Reduction Rules.}
We can now recall the reduction rules of \UOCBV. 
We first define a \defn{step kind} as an element of the set
$\set{\dbsym, \lsvsym, \rulesub\var\val}$, with $\var$ a variable 
and $\val$ a value such that $\var \notin \fv\val$.
This set is used to decorate reduction steps so as to distinguish between their different nature.
Let $\tov\rulename\aset\sset\appflag$ be a family of binary relations, 
where $\rulename$ is a step kind, $\aset$ is an abstraction frame, 
$\sset$ is a structure frame, and $\appflag \in \set{\app, \nonapp}$ is 
a constant called a \defn{positional flag}.
The positional flag $\app$ indicates that the current focus of evaluation
is in an applied position, while $\nonapp$ corresponds to a non-applied position.
For example, the underlined $\var$ is applied
in $\underline{\var}\esub{\var}{\id}\,\tm$
but not in $(\varthree\,\underline{\var})\esub{\var}{\id}$.
The relation
$\tov\rulename\aset\sset\appflag$ of \UOCBV is 
inductively defined\footnote{In~\cite{BKM-CADE25}, this same 
relation is written $\tovoriginal\rulename\aset\sset\appflag$.} in~\autoref{fig:uocbv_reduction}.

\begin{figure}[t]
\begin{mathpar}
  \inferrule* [right = \ruleUDb]
    { }
    {(\lam\var\tm)\sctx\,\tmtwo \tov\ruledb\aset\sset\appflag \tm\esub\var\tmtwo\sctx}

  \inferrule* [right = \ruleUSub]
    { }
    {\var \tov{\rulesub\var\val}{\aset\cup\set\var}\sset\app \val}

  \inferrule* [right = \ruleULsv]
    {\tm \tov{\rulesub\var\val}{\aset\cup\set\var}\sset\appflag \tm'
    \\ \var \notin \aset\cup\sset
    \\ \val\sctx \in \HAbs\aset}
    {\tm\esub\var{\val\sctx} \tov\rulelsv\aset\sset\appflag \tm'\esub\var\val\sctx}

  \inferrule* [right = \ruleUEsR]
    {\tmtwo \tov\rulename\aset\sset\nonapp \tmtwo'}
    {\tm\esub\var\tmtwo \tov\rulename\aset\sset\appflag \tm\esub\var{\tmtwo'}}

  \inferrule* [right = \ruleUAppL]
    {\tm \tov\rulename\aset\sset\app \tm'}
    {\tm\,\tmtwo \tov\rulename\aset\sset\appflag \tm'\,\tmtwo}

  \inferrule* [right = \ruleUAppR]
    {\tm \in \Struct\sset \\ \tmtwo \tov\rulename\aset\sset\nonapp \tmtwo'}
    {\tm\,\tmtwo \tov\rulename\aset\sset\appflag \tm\,\tmtwo'}

  \inferrule* [right = \ruleUEsLAbs]{
    \tm \tov\rulename{\aset\cup\set\var}\sset\appflag \tm'
    \\ \tmtwo \in \HAbs\aset
    \\ \var \notin \aset\cup\sset
    \\ \var \notin \fv\rulename
  }{
    \tm\esub\var\tmtwo
    \tov\rulename\aset\sset\appflag
    \tm'\esub\var\tmtwo
  }

  \inferrule* [right = \ruleUEsLStruct] {
    \tm \tov\rulename\aset{\sset\cup\set\var}\appflag \tm'
    \\ \tmtwo \in \Struct\sset
    \\ \var \notin \aset\cup\sset
    \\ \var \notin \fv\rulename
  }{
    \tm\esub\var\tmtwo 
    \tov\rulename\aset\sset\appflag
    \tm'\esub\var\tmtwo
  }
\end{mathpar}
\caption{Reduction Rules of \UOCBV}
\label{fig:uocbv_reduction}
\end{figure}

Note that rule \ruleUDb is the same as the one
used in \VSC and \VSCne. Rule \ruleUSub is to be understood as an \emph{auxiliary} rule that performs a
\emph{linear} substitution on a particular free occurrence of a variable.
Rule \ruleULsv relies on an auxiliary \ruleUSub-step to perform a linear substitution within the left-hand side of a closure. All the remaining  rules
close the reduction relation under \emph{weak} contexts. Since \ruleUSub plays only an auxiliary role in the reduction, we are principally interested in the relation generated by $\dbsym$- and $\lsvsym$-steps. 
This motivates the forthcoming notion of \emph{primary reduction} introduced below.

The reduction in \UOCBV is non-erasing as well as in \VSCne, but the former can be seen as an optimisation of the latter (and thus of \VSC).
As explained before, reduction in \UOCBV is \emph{weak}, \ie it is not allowed under
abstractions, and it is \emph{open}, since it allows reducing
terms with free occurrences of variables.  The specification
of \UOCBV leaves the order of evaluation slightly underspecified. For
example, in a term like $\var\,\tm\,\tmtwo$, the subterms $\tm$ and
$\tmtwo$ can be evaluated in any order, without necessarily adhering
to a strict (\eg~left-to-right) order. Technically speaking, the rules
to evaluate applications are not mutually exclusive, and similarly for
ESs. 
This non-determinism is harmless, as in fact \UOCBV enjoys the
diamond property (\cf~\autoref{thm:diamond-property}). 
For a concrete example in which rules \ruleUAppL and \ruleUAppR overlap,
consider:
\[
    \var\,(\vartwo\esub\vartwo\id)\,(\id\,\id) \
  \tovinv\ruledb\emptyset{\set\var}\nonapp \
    \var\,(\id\,\id)\,(\id\,\id) \
  \tov\ruledb\emptyset{\set\var}\nonapp \
    \var\,(\id\,\id)\,(\vartwo\esub\vartwo\id)
\]

For a concrete example in which \ruleUEsLStruct and \ruleUEsR overlap,
consider:
\[
    (\vartwo\esub\vartwo\id)\esub\var{\varthree\,(\id\,\id)} \
  \tovinv\ruledb\emptyset{\set\varthree}\nonapp \
    (\id\,\id)\esub\var{\varthree\,(\id\,\id)} \
  \tov\ruledb\emptyset{\set\varthree}\nonapp \
    (\id\,\id)\esub\var{\varthree\,(\vartwo\esub\vartwo\id)}
\]

As discussed in the introduction, useful evaluation is \emph{context-sensitive}.
That is, the context surrounding a term determines whether a step is 
useful or not.
For instance, the term $\tm = \var\esub\var\id$ is already a 
\UOCBV-normal form with respect to
parameters $\aset=\emptyset$, $\sset=\emptyset$, and $\appflag=\nonapp$
because substituting $\var$ by $\id$ is not useful, 
as it does not create a $\dbsym$-redex.
However, if $\tm$ appears on the left of an application, as in 
$\var\esub\var\id \, (\vartwo \, \id)$, the substitution of $\var$ 
by $\id$ becomes useful, leading to the following reduction sequence:
\[
    \var\esub\var\id \, (\vartwo \, \id) \
  \tov\rulelsv\emptyset{\set{\vartwo}}\nonapp \
    \id\esub\var\id \, (\vartwo \, \id) \
  \tov\ruledb\emptyset{\set{\vartwo}}\nonapp \
    \varthree\esub\varthree{\vartwo \, \id}\esub\var\id
\]
Conversely, a useful step like 
$(\var\,\tmtwo)\esub\var\id 
\tov\rulelsv\emptyset\emptyset\nonapp (\id\,\tmtwo)\esub\var\id$ may 
not be allowed if the term is located below a context such as
$\lam\vartwo\ctxhole$; in fact, note that 
$\lam\vartwo{(\var\,\tmtwo)\esub\var\id}$ cannot be reduced in \UOCBV.

A term $\tm$ is \defn{$(\rulename,\aset,\sset,\appflag)$-irreducible}
if there is no term $\tm'$ such that 
$\tm \tov\rulename\aset\sset\appflag \tm'$.
A term $\tm$ belongs to the set $\Irred\aset\sset\appflag$ if $\tm$ 
is $(\rulename,\aset,\sset,\appflag)$-irreducible.

An abstraction frame $\aset$ and a structure frame $\sset$ are said 
to verify the \defn{correctness invariant} for $\tm$, written
$\inv\aset\sset\tm$, if $\aset \cap \sset = \emptyset$ and 
$\fv\tm \subseteq \aset \cup \sset$.
This invariant is often assumed in theorems and lemmas.

As explained before, rule \ruleUSub is a secondary rule
used to ensure that substitution in \UOCBV is \emph{linear}. In
general, we care about \defn{primary \UOCBV reduction}, 
defined as $\tovTopb{\aset}{\sset} {\eqdef} \tov\dbsym{\aset}\sset\nonapp \cup \tov\lsvsym{\aset}\sset\nonapp$,
where the free variables of the term being evaluated are included in $\aset \cup \sset$. 
From now on, \UOCBV reduction  refers to \emph{primary} \UOCBV reduction, unless stated otherwise. 
The same for the corresponding notion of \defn{\emph{primary} \UOCBV evaluation}.

\begin{exa}
\label{ex:reduction-for-tight}
The following is a \UOCBV-reduction sequence to normal form. 
Notice the non-erasing aspect of the reduction sequence:
\[
\begin{array}{llll}
& (\var \, \varthree)\esub\var{\vartwo\esub\vartwo\id}
& \tov\rulelsv\emptyset{\set\varthree}\nonapp
& (\vartwo \, \varthree)\esub\var\vartwo\esub\vartwo\id 
\\ [0.5em] 
  \tov\rulelsv\emptyset{\set\varthree}\nonapp
& (\id \, \varthree)\esub\var\vartwo\esub\vartwo\id
& \tov\ruledb\emptyset{\set\varthree}\nonapp 
& \var_1\esub{\var_1}\varthree\esub\var\vartwo\esub\vartwo\id
\end{array}
\]
\end{exa}

To conclude this section,
we recall two key operational properties of \UOCBV,
starting with a characterisation of the set of \UOCBV-normal forms:
\begin{thmC}[{\cite[Theorem~1]{BKM-CADE25}}]
\label{thm:uocbv_characterisation_of_normal_forms}
The set of irreducible terms $\Irred\aset\sset\appflag$ is exactly 
the set $\NF\aset\sset\appflag$ inductively defined as:

\begin{mathpar}
  \inferrule* [right = \ruleUNFVar]{
    \var \in \aset \text{ implies } \appflag = \nonapp
  }{
    \var \in \NF\aset\sset\appflag
  }

  \inferrule* [right = \ruleUNFLam]{
  }{ 
    \lam\var\tm \in \NF\aset\sset\nonapp
  }

  \inferrule* [right = \ruleUNFApp]{
    \tm \in \NF\aset\sset\app
    \\ \tmtwo \in \NF\aset\sset\nonapp
  }{
    \tm\,\tmtwo \in \NF\aset\sset\appflag
  }
  
  \inferrule* [right = \ruleUNFEsAbs]{
    \tm \in \NF{\aset\cup\set\var}\sset\appflag
    \\ \tmtwo \in \NF\aset\sset\nonapp
    \\ \tmtwo \in \HAbs\aset
  }{
    \tm\esub\var\tmtwo \in \NF\aset\sset\appflag
  }

  \inferrule* [right = \ruleUNFEsStruct]{
    \tm \in \NF\aset{\sset\cup\set\var}\appflag
    \\ \tmtwo \in \NF\aset\sset\nonapp
    \\ \tmtwo \in \Struct\sset
  }{
    \tm\esub\var\tmtwo \in \NF\aset\sset\appflag
  }
\end{mathpar}
\end{thmC}

Concerning \UOCBV reduction, a relevant property in this article is that it enjoys a strong form of confluence, namely the \emph{diamond property}.
To obtain such a result, we first provide a more general statement, for which it is straightforward to conclude the particular case of primary \UOCBV.
\begin{restatable}[Diamond Property of \UOCBV]{thm}{diamondproperty}
\label{thm:diamond-property}
Let $\tm$ be a term, $\aset$ an abstraction frame, and $\sset$ a 
structure frame such that $\inv\aset\sset\tm$ holds.
Consider two step kinds $\rulename_1$ and $\rulename_2$ and
suppose that
$\tm \tov{\rulename_1}\aset\sset\appflag \tm_1$ and
$\tm \tov{\rulename_2}\aset\sset\appflag \tm_2$, 
with $\tm_1 \neq \tm_2$.
Assume that if $\rulename_1 = \rulesub\var{\val_1}$ and 
$\rulename_2 = \rulesub\var{\val_2}$, then $\val_1 = \val_2$.

Then, there exists a term $\tm'$ such that
$\tm_1 \tov{\rulename_2}\aset\sset\appflag \tm'$ and
$\tm_2 \tov{\rulename_1}\aset\sset\appflag \tm'$.
\end{restatable}

\begin{proof}
By structural induction on $\tm$. Full proof in~\autoref{appendix}.
\end{proof}

\begin{cor}
\label{cor:diamond-property-primary-uocbv}
\UOCBV reduction enjoys the diamond property. 
More precisely,
let $\rulename_1, \rulename_2 \in \set{\dbsym, \lsvsym}$.
If $\tm \tov{\rulename_1}\aset\sset{\nonapp} \tm_1$ and
$\tm \tov{\rulename_2}\aset\sset{\nonapp} \tm_2$, 
then, there exists a term $\tm'$ such that
$\tm_1 \tov{\rulename_2}\aset\sset{\nonapp} \tm'$ and
$\tm_2 \tov{\rulename_1}\aset\sset{\nonapp} \tm'$.
\end{cor}

Some straightforward corollaries of the diamond property follow.
\begin{cor}
\label{cor:consequences-dp} \hfill
\begin{itemize}
\item \label{cor:consequences-dp-weakly-strong} 
  A term $\tm$ is weakly terminating in \UOCBV if and only if $\tm$ is terminating in primary \UOCBV.
\item \label{cor:consequences-dp-same-length} 
  Any two \UOCBV-reduction sequences to normal form have the same number of
  $\dbsym$- and $\lsvsym$-steps.
\item \label{cor:consequences-dp-conf} 
  \UOCBV reduction is confluent. 
\end{itemize}
\end{cor}
These corollaries are important for reaching our results in Section~\ref{sec:ucbv-U}.

\section{The Type System \texorpdfstring{$\typesystem$}{U}}
\label{sec:systemU}

In this section, we introduce the non-idempotent intersection type
system $\typesystem$ and state some properties that will be useful for the following sections. 
We begin with some preliminary considerations. 
First, normal forms in \VSCne and \UOCBV fall into three categories (recall \VSCne-normal
forms from~\autoref{subsec:vscne-op} and \UOCBV-normal forms
from~\autoref{thm:uocbv_characterisation_of_normal_forms}). 
A normal form \emph{without ESs} may be an \emph{abstraction} $\lam\var\tm$,
a \emph{variable} $\var$, or a \emph{structure}
$\var\,\tm_1\hdots\tm_n$, where $n \geq 1$ and each
$(\tm_i)_{i \in \set{1,\hdots,n}}$ is recursively in normal form.
Normal forms \emph{with ESs} also fit these categories up
to unfolding (\ie performing all ESs). 
Thus, for instance,
$(\varthree\,\vartwo)\esub\varthree{\var\,\vartwo}$ is a structure
because it can be understood as representing the term
$\var\,\vartwo\,\vartwo$.

Second, one key feature of non-idempotent intersection types is that a single
occurrence of an expression in the (static) source code of a
program may be (dynamically) used multiple times during execution, 
each time playing a different role, expressed by a different type.
Specifically, each abstraction is assigned a multiset of arrow types 
of the form $\mset{\typ_1,\hdots,\typ_n}$, where the cardinality $n$ 
corresponds to the number of times this abstraction\footnote{Or, 
more precisely, one of the \emph{descendants} of this abstraction.} 
takes part as the left-hand side of a $\dbsym$-redex during evaluation. 
For example, in $(\lam{f}{\var\,(f\,\vartwo)\,(f\,\varthree)}) \, \id$,
the identity function takes part in \emph{two} $\dbsym$-redexes 
during evaluation, so any type assignment for this term must provide 
\emph{two} different arrow types for $\id$. 
In contrast, abstractions that are never applied, such as $\id$ in 
$\var\esub\var\id$, are typed with the \emph{empty} multiset $\emset$.

In the type system $\typesystem$, terms whose normal form represent
abstractions (including variables bound to abstractions) are assigned
finite multisets of arrow types, called \defn{arrow multi-types}.  On
the other hand, terms whose normal forms represent structures
(including applied variables and variables bound to structures) are
assigned a distinguished type we denote with $\tightN$.  Formally, the
types of system $\typesystem$ are given by the following grammar:
\[
  \begin{array}{rlcl@{\sep\sep}rlcll}
    \text{\defn{(Arrow Types)}}       & \typ, \typtwo & \eqgram & \optmtyp \to \mtyp 
  \\
    \text{\defn{(Arrow Multi-Types)}} & \nityp, \nityptwo, \nityp_1, \nityp_2 & \eqgram & \mset{\typ_i}_\iI & \text{where $I$ is a finite set}
  \\ 
    \text{\defn{(Types)}}             & \mtyp, \mtyptwo, \mtyp_1, \mtyp_2 & \eqgram & \tightN \mid \nityp 
  \\
    \text{\defn{(Optional Types)}}    & \optmtyp, \optmtyptwo, \optmtyp_1, \optmtyp_2 & \eqgram & \none\ \mid \mtyp
  \end{array}
\] 

A type $\mtyp$ is said to be a \defn{constant type} if
$\mtyp \in \set{\tightN, \emset}$, in which case we denote
$\mtyp$ with the special symbol $\tight$. The type $\tightN$
is assigned to terms whose normal forms are structures, and $\emset$
is assigned to terms whose normal forms are abstractions not taking
part in $\dbsym$-redexes. Unlike in previous
non-idempotent intersection type systems for \CBV, we
distinguish the optional type $\none$ from $\emset$. 
The latter is used as explained above, while the former indicates that \emph{no typing
information is available}.
As in system $\systemV$, we want to count the number of elements in
a type, so that in system $\typesystem$, we extend the function $\length \_$ to optional types as follows:
$\length\none \eqdef 0$, $\length\tightN \eqdef 0$ and
$\length{\mset{\typ_k}_\kK} \eqdef \length{K}$. We remark that
$\length{\optmtyp_1+\optmtyp_2} = \length{\optmtyp_1}
+ \length{\optmtyp_2}$. 

\defn{Typing environments} $\tctx, \tctxtwo, \ldots$ are functions 
mapping variables to optional types, assigning $\none$ to all but 
finitely many variables.
The \defn{domain} of an environment $\tctx$ is
$\dom\tctx \eqdef \set{\var \mid \tctx(\var) \neq \none}$, and
$\emptyctx$ denotes the empty typing environment, mapping each 
variable to $\none$. 

The \defn{union of arrow multi-types}, written
$\nityp_1 \niunion \nityp_2$, is a multiset of types defined as
expected, where $\emset$ is the neutral element. The \defn{union of
types}, written $\mtyp_1 + \mtyp_2$, is the (associative) \emph
{partial} operation on types given by $\tightN
+ \tightN \eqdef \tightN$ and $\nityp_1
+ \nityp_2 \eqdef \nityp_1 \niunion \nityp_2$, with all other
cases  undefined.  The \defn{union of optional types} is
defined by $\none + \none \eqdef \none$, $\none + \mtyp \eqdef \mtyp$,
and $\mtyp + \none \eqdef \mtyp$, making $\none$ the neutral element.
For typing environments $(\tctx_i)_\iI$, we write $+_\iI \tctx_i$ for
the environment mapping each variable $\var$ to
$+_\iI \tctx_i(\var)$, where $\tctx + \tctxtwo$ and $\tctx
+_\jJ \tctxtwo_j$ (meaning $\tctx + (+_\jJ \tctxtwo_j)$) are taken to
be particular instances of the general notation. 
When $\dom\tctx \cap \dom\tctxtwo = \emptyset$, we write $\tctx;\tctxtwo$
instead of $\tctx + \tctxtwo$ to emphasise that the domains of $\tctx$
and $\tctxtwo$ are disjoint.
We write $\var_1 : \optmtyp_1;\hdots; \var_n : \optmtyp_n$ to make explicit that each $\var_i$ is mapped to $\optmtyp_i$ ($i \in \set{1,\hdots,n}$).
Note that the environment $\tctx; \var : \none$ is identical to $\tctx$, whereas $\tctx; \var : \emset$ is not.

We define a \defn{controlled weakening} relation $\mleq$ between 
optional types and types by declaring that $\none \mleq \tight$ and 
$\mtyp \mleq \mtyp$ hold.

\defn{Typing judgements} in system $\typesystem$ are of the form
$\judgs\tctx\tm\mtyp$, where $\tctx$ is a type environment, and
$\mtyp$ is a type.
\defn{Typing rules} of system $\typesystem$ are given in~\autoref{fig:systemU}.

\begin{figure}[t]
  \begin{mathpar}
    \inferrule* [right = \ruleTypVar]
      { }
      {\judgs{\var:\mtyp}\var\mtyp}
    
    \inferrule* [right = \ruleTypAbs]{
      \left(\judgs{\tctx_i;\var:\optmtyp_i}\tm{\mtyptwo_i}\right)_\iI
    }{
      \judgs{+_\iI\tctx_i}{\lam\var\tm}{\mset{\optmtyp_i\to\mtyptwo_i}_\iI}
    }
  \end{mathpar}

  \begin{mathpar}
    \inferrule* [right = \ruleTypAppP]{
      \judgs\tctx\tm\tightN \\ \judgs\tctxtwo\tmtwo\tight
    }{
      \judgs{\tctx+\tctxtwo}{\tm\,\tmtwo}\tightN
    }

    \inferrule* [right = \ruleTypAppC]{
      \judgs\tctx\tm{\mset{\optmtyp\to\mtyptwo}}
      \\ \optmtyp \mleq \mtyp 
      \\ \judgs\tctxtwo\tmtwo\mtyp
    }{
      \judgs{\tctx+\tctxtwo}{\tm\,\tmtwo}\mtyptwo
    }
    
    \inferrule* [right = \ruleTypES]{
      \judgs{\tctx;\var:\optmtyp}\tm\mtyptwo
      \\ \optmtyp \mleq \mtyp
      \\ \judgs\tctxtwo\tmtwo\mtyp
    }{
      \judgs{\tctx+\tctxtwo}{\tm\esub\var\tmtwo}\mtyptwo
    }
  \end{mathpar}
  \caption{Typing Rules of System $\typesystem$}
  \label{fig:systemU}
\end{figure}

As in system $\systemV$, we write $\derivablep\deriv{\judgs\tctx\tm\mtyp}\typesystem$ to denote 
a type derivation $\deriv$ in system $\typesystem$ ending in the 
typing judgement $\judgs\tctx\tm\mtyp$. 
Sometimes, $\deriv$ or  $\typesystem$ may be omitted.

There are no general \emph{weakening} and 
\emph{contraction} rules, though the facts that $\none \mleq \tight$ 
and $\tight + \tight = \tight$ hold whenever 
$\tight \in \set{\tightN, \emset}$ can be understood as 
controlled forms of weakening and contraction, respectively.

Rule \ruleTypAbs follows the same principle as rule $\ruleTypAbsV$ of
system $\systemV$ (\cf~\autoref{fig:systemV}), and rule \ruleTypES can be easily reconstructed from rules \ruleTypAbs and 
\ruleTypAppC. Rules for applications follow the 
\emph{consuming/persistent paradigm}~\cite{KesnerVial20,KesnerViso22} and
deserve some discussion. 
Indeed, system $\typesystem$ keeps track of the different natures of 
the application constructors involved in
the evaluation process of a program: a term constructor is 
\defn{consuming} if it is \emph{eliminated} by evaluation, whereas
a \defn{persistent} constructor \emph{remains} as a syntactical part
of the normal form. 
For example, given the following \VSCne-reduction to normal form 
$\var \, ((\lam\vartwo\vartwo) \, \varthree) \todb 
\var \, \vartwo\esub\vartwo\varthree$, we say that the leftmost 
application constructor in 
$\var \, ((\lam\vartwo\vartwo) \, \varthree)$ is persistent, while 
the rightmost is consuming.
Indeed, rule \ruleTypAppP (resp. \ruleTypAppC) is used to type this 
\emph{persistent} (resp. \emph{consuming}) application.
More specifically, rule \ruleTypAppP requires the left-hand side of
the application to have type $\tightN$, ensuring that no redex is
created during evaluation, so the application constructor is
persistent.  In particular, left-hand sides of applications
representing a potential structure (a term whose unfolding is an
application headed by a variable) are expected to be given the type
$\tightN$, which means that the underlined applications in
$\underline{\varthree\,\tm}$ and
$(\underline{\var\,\tm})\esub{\var}{\varthree}$ would type with
rule \ruleTypAppP.  Rule \ruleTypAppC requires the left-hand side of
the application to type with an arrow multi-type, ensuring a
$\dbsym$-redex is created at some point during evaluation. For
instance, the underlined applications in
$(\underline{\var\,\tm})\esub{\var}{\id}$ and
$(\underline{\var\,\tm})\esub{\var}{\vartwo}\esub{\vartwo}{\id}$ would
type with rule \ruleTypAppC. Notice, however, that the
consuming/persistent paradigm in system $\typesystem$ is used
exclusively for typing applications. This approach contrasts with
more radical solutions for \CBV (see~\cite{KesnerViso22}), which apply
the same paradigm to type \emph{all} constructors, thereby increasing
the number of typing rules and making the type system more difficult
to use.

Let us now type our working~\autoref{exa:working-example} in system $\typesystem$:
\begin{exa}
\label{ex:derivation}
Let $\mtyp \eqdef \mset{\tightN \to \tightN}$,
we build the following derivation
$\derivtwo$:
\begin{mathpar}
  \inferrule*[right = \ruleTypES]{
    \inferrule*[right = \ruleTypAppC]{
      \inferrule*[right = \ruleTypVar]{
      }{
        \judgs{\var : \mtyp}\var\mtyp
      }
      \\
      \inferrule*[Right = \ruleTypVar]{
      }{
        \judgs
          {\varthree : \tightN}
          {\varthree}
          \tightN
      }
    }{
      \judgs
        {\varthree : \tightN,\var:\mtyp}
        {\var\,\varthree}
        \tightN
    }
    \HS
    \inferrule*[Right = \ruleTypES]{
      \inferrule*[Right = \ruleTypVar]{
      }{
        \judgs{\vartwo : \mtyp}\vartwo\mtyp
      }
      \\
      \inferrule*[right = \ruleTypAbs]{
        \inferrule*[Right = \ruleTypVar]
          { }
          {\inferrule{}{\judgs{\varfour : \tightN}\varfour\tightN}}
      }{
        \judgs\emptyctx\id\mtyp
      }
    }{
      \judgs\emptyctx{\vartwo\esub\vartwo\id}\mtyp
    }
  }{
    \judgs
      {\varthree : \tightN}
      {(\var \, \varthree)\esub\var{\vartwo\esub\vartwo\id}}
      \tightN
  }
\end{mathpar}
where $\mtyp \mleq \mtyp$ and 
$\tightN \mleq \tightN$ trivially hold.
\end{exa}

The use of controlled weakening deserves some discussion as well.
Indeed, in \CBV, the arguments of applications and ESs are always
evaluated. 
Since system $\typesystem$ intends to characterise termination, all 
arguments must necessarily be typed. 
For example, in $(\lam\var\id)\,(\varthree\,\varthree)$, the 
subexpression $\varthree\,\varthree$ is a structure of type $\tightN$. 
One would perhaps expect to be able to type the abstraction 
$\lam\var\id$ with type $\tightN \to \mtyp$. 
But it is impossible to derive the judgement 
$\var:\tightN \vdash \id : \mtyp$, because $\var$ does not occur free 
in $\id$. 
This is where the controlled weakening relation ($\mleq$) in the 
premise of rules $\ruleTypAppC$ and $\ruleTypES$ comes into play. 
In our example, one can give a ``stronger'' type to $\lam\var\id$, 
namely $\none \to \mtyp$, using the fact that $\none \mleq \tightN$. 
A similar use of controlled weakening for the rule \ruleTypES is 
shown below, where it should be recalled that the empty typing 
environment $\emptyctx$ is identical to $\emptyctx;\var:\none$.
\begin{mathpar}
  \inferrule* [right = \ruleTypES]{
    \inferrule*[right = \ruleTypAbs]
      { }
      {\judgs{\emptyctx; \var : \none}\id\emset}
    \\ \none \mleq \tightN
    \\
    \inferrule*[right = \ruleTypAppP]{
      \inferrule*[right = \ruleTypVar]{
      }{
        \judgs{\varthree : \tightN}\varthree\tightN
      }
      \\
      \inferrule*[right = \ruleTypVar]{
      }{
        \judgs{\varthree : \tightN}\varthree\tightN
      }
    }{
      \judgs{\varthree : \tightN}{\varthree\,\varthree}\tightN
    }
  }{
    \judgs
      {\varthree : \tightN}
      {\id\esub\var{\varthree\,\varthree}}
      {\emset}
  }
\end{mathpar}

We now state properties of the type system $\typesystem$ which will be
used throughout this section and Section~\ref{sec:ucbv-U}. We start by focusing on the quantitative information we can retrieve from system $\typesystem$. 
The size of typing derivations in system $\typesystem$ can be measured using the function \sizeU{\_}, as defined in~\autoref{fig:sz_U}.
Notice that  function \sizeU{\_} counts similarly to function $\sizeV{\_}$ in~\autoref{subsec:systemV}. 

\begin{figure}[h]
\[
  \begin{aligned}
    \sizeU{
      \inferrule{ }{\judgs{\var : \mtyp}\var\mtyp}\ruleTypVar
    } 
  & \eqdef \length\mtyp
  \\
    \sizeU{
      \inferrule
        {(\deriv_i)_\iI}
        {\judgs\tctx{\lam\var\tmtwo}{\mset{\typ_i}_\iI}}\ruleTypAbs
    } 
  & \eqdef
    \length{I} + \sum_\iI \sizeU{\deriv_i}
  \\
    \sizeU{
      \inferrule
        {\deriv_1 \\ \deriv_2}
        {\judgs\tctx{\tmtwo \, \tmthree}\mtyp}{\ruleTypAppP/\ruleTypAppC}
    } 
  & \eqdef
    1 + \sizeU{\deriv_1} + \sizeU{\deriv_2}
  \\
    \sizeU{
      \inferrule
        {\deriv_1 \\ \deriv_2}
        {\judgs\tctx{\tmtwo\esub\var\tmthree}\mtyp}\ruleTypES
    } 
  & \eqdef
      \begin{cases*}
        1 + \sizeU{\deriv_1} + \sizeU{\deriv_2} & if $\var \in \fv\tmtwo$
      \\
        \sizeU{\deriv_1} + \sizeU{\deriv_2}     & otherwise
      \end{cases*}
  \end{aligned}
\]
  \caption{Size of Typing Derivations in System $\typesystem$}
  \label{fig:sz_U}
\end{figure}

Thus, for example, the size of the derivation $\derivtwo$ in~\autoref{ex:derivation} is $6$. 

System $\typesystem$, as well as system $\systemV$, is a relevant type system:
\begin{lem}[Relevance]
\label{lem:relevance_systemU}
If $\derivablep\deriv{\judgs\tctx\tm\mtyp}\typesystem$, then $\rv{\tm} \subseteq \dom{\tctx} \subseteq \fv{\tm}$.
\end{lem}
\begin{proof}
The proof is straightforward by induction on $\deriv$.
\end{proof}

In system $\typesystem$, as well as in other non-idempotent intersection types capturing \CBV strategies, it is possible to decompose the multiset types of values:
\begin{lem}[Splitting/Merging for System $\typesystem$]
\label{l:split-merge-value}
Let $\val$ be a value, $\tctx$ a typing environment, and $\mtyp_1, \mtyp_2$ be two types such that their union 
$\mtyp_1 + \mtyp_2$ is well-defined.
Then, the following are equivalent:
\begin{enumerate}
\item 
  There exists a derivation $\deriv$ such that
  $\derivable\deriv{\judgs\tctx\val{\mtyp_1 + \mtyp_2}}$.
\item 
  There exist derivations $\deriv_1,\deriv_2$ and typing 
  environments $\tctx_1,\tctx_2$ such that
  $\derivable{\deriv_1}{\judgs{\tctx_1}\val{\mtyp_1}}$, and
  $\derivable{\deriv_2}{\judgs{\tctx_2}\val{\mtyp_2}}$ , and
  $\tctx = \tctx_1 + \tctx_2$.
\end{enumerate}
Moreover,  $\sizeU\deriv = \sizeU{\deriv_1} + \sizeU{\deriv_2}$.
\end{lem}

\begin{proof}
The $\val$ is either a variable or an abstraction. 
We first show that $(1 \implies 2)$.
The judgement in the conclusion of $\deriv$ can be derived either by 
rule \ruleTypVar or rule \ruleTypAbs, so we analyse both cases:
\begin{itemize}
\item \ruleTypVar.
  Then, $\deriv$ is the axiom 
  $\judgs{\var : \mtyp_1 + \mtyp_2}\var{\mtyp_1 + \mtyp_2}$,
  where $\tctx = \var : \mtyp_1 + \mtyp_2$ and $\val = \var$.
  Taking $\tctx_1 \eqdef \var : \mtyp_1$ and 
  $\tctx_2 \eqdef \var : \mtyp_2$, we obtain derivations $\deriv_1$ 
  and $\deriv_2$ such that
  $\derivable{\deriv_1}{\judgs{\var : \mtyp_1}\var{\mtyp_1}}$ and
  $\derivable{\deriv_2}{\judgs{\var : \mtyp_2}\var{\mtyp_2}}$.
  Note that $\tctx 
  = \var : \mtyp_1 + \mtyp_2 
  = (\var : \mtyp_1) + (\var : \mtyp_2)
  = \tctx_1 + \tctx_2$.
  And moreover, 
  $\sizeU\deriv 
  = \length{\mtyp_1 + \mtyp_2} 
  = \length{\mtyp_1} + \length{\mtyp_2}
  = \sizeU{\deriv_1} + \sizeU{\deriv_2}$.
\item \ruleTypAbs.
  Then, $\deriv$ is of the form
  \begin{mathpar}
    \inferrule
      {\left(
        \derivable
          {\deriv'_i}
          {\judgs{\tctx_i; \var : \optmtyptwo_i}\tm{\mtypthree_i}}
      \right)_\iI}
      {\judgs{+_\iI \tctx_i}{\lam\var\tm}{\mset{\optmtyptwo_i \to \mtypthree_i}_\iI}}
  \end{mathpar}
  where $\tctx = +_\iI \tctx_i$, $\val = \lam\var\tm$, and
  $\mtyp_1 + \mtyp_2 = \mset{\optmtyptwo_i \to \mtypthree_i}_\iI$
  with $I$ finite.
  Then, we can write $I = I_1 \uplus I_2$ in such a way that
  $\mtyp_1 = \mset{\optmtyptwo_i \to \mtypthree_i}_{i \in I_1}$ and
  $\mtyp_2 = \mset{\optmtyptwo_i \to \mtypthree_i}_{i \in I_2}$, so
  that
  $\mset{\optmtyptwo_i \to \mtypthree_i}_\iI = 
  \mset{\optmtyptwo_i \to \mtypthree_i}_{i \in I_1} +
  \mset{\optmtyptwo_i \to \mtypthree_i}_{i \in I_2}$.

  Taking 
  $\tctx_1 \eqdef +_{i \in I_1} \tctx_i$ and
  $\tctx_2 \eqdef +_{i \in I_2} \tctx_i$, we obtain typing
  derivations $\deriv_1$ and $\deriv_2$ as follows:
  applying rule \ruleTypAbs to $(\deriv'_i)_{i \in I_1}$ yields
  $\derivable
    {\deriv_1}
    {\judgs
      {+_{i \in I_1} \tctx_i}
      {\lam\var\tm}
      {\mset{\optmtyptwo_i \to \mtypthree_i}_{i \in I_1}}}$, and 
  applying rule \ruleTypAbs to $(\deriv'_i)_{i \in I_2}$ yields
  $\derivable
    {\deriv_2}
    {\judgs
      {+_{i \in I_2} \tctx_i}
      {\lam\var\tm}
      {\mset{\optmtyptwo_i \to \mtypthree_i}_{i \in I_2}}}$.
  Note that $\tctx 
  = +_\iI \tctx_i 
  = +_{i \in I_1 \uplus I_2} \tctx_i 
  = +_{i \in I_1} \tctx_i +_{i \in I_2} \tctx_i
  = \tctx_1 + \tctx_2$.

  Moreover, 
  \begin{align*} 
  \sizeU\deriv 
  & = \length{I} + \sum_\iI \sizeU{\deriv'_i}
    = \length{I_1 \uplus I_2} + \left(\sum_{i \in I_1} \sizeU{\deriv'_i}\right) 
    + \left(\sum_{i \in I_2} \sizeU{\deriv'_i}\right)
  \\
  & = \length{I_1} + \length{I_2} + \left(\sum_{i \in I_1} \sizeU{\deriv'_i}\right) 
    + \left(\sum_{i \in I_2} \sizeU{\deriv'_i}\right)
  = \sizeU{\deriv_1} + \sizeU{\deriv_2}
  \end{align*}
\end{itemize}

We now show that $(2 \implies 1)$. 
The judgements in the conclusion of $\deriv_1$ and $\deriv_2$ can be 
derived either by rule \ruleTypVar or rule \ruleTypAbs, so we analyse 
both cases:
\begin{itemize}
\item \ruleTypVar.
  Then, $\derivable{\deriv_1}{\judgs{\var : \mtyp_1}\var{\mtyp_1}}$, 
  where $\tctx_1 = \var : \mtyp_1$ and $\val = \var$.
  On the other hand, 
  $\derivable{\deriv_2}{\judgs{\var : \mtyp_2}\var{\mtyp_2}}$, where 
  $\tctx_2 = \var : \mtyp_2$ and $\val = \var$.
  Since $\tctx = \tctx_1 + \tctx_2$ by hypothesis, we then obtain
  $\derivable
    \deriv
    {\judgs{\var : \mtyp_1 + \mtyp_2}\var{\mtyp_1 + \mtyp_2}}$ 
  by the \ruleTypVar axiom.
\item \ruleTypAbs.
  Then, $\deriv_1$ is of the form
  \begin{mathpar}
    \inferrule
      {\left(
          \derivable{\deriv'_j}{\judgs{\tctx_j; \var : \optmtyptwo_j}\tm{\mtypthree_j}}
      \right)_\jJ}
      {\judgs{+_\jJ \tctx_j}{\lam\var\tm}{\mset{\optmtyptwo_j \to \mtypthree_j}_\jJ}}
  \end{mathpar}
  where $\tctx_1 = +_\jJ \tctx_j$, $\val = \lam\var\tm$, and 
  $\mtyp_1 = \mset{\optmtyptwo_j \to \mtypthree_j}_\jJ$ with $J$ finite.
  On the other hand, $\deriv_2$ is of the form
  \begin{mathpar}
    \inferrule
      {\left(
          \derivable{\deriv'_k}{\judgs{\tctx_k; \var : \optmtyptwo_k}\tm{\mtypthree_k}}
      \right)_\kK}
      {\judgs{+_\kK \tctx_k}{\lam\var\tm}{\mset{\optmtyptwo_k \to \mtypthree_k}_\kK}}
  \end{mathpar}
  where $\tctx_2 = +_\kK \tctx_k$, $\val = \lam\var\tm$, and 
  $\mtyp_2 = \mset{\optmtyptwo_k \to \mtypthree_k}_\kK$ with $K$ finite.
  And moreover, $\tctx = \tctx_1 + \tctx_2 
  = +_\jJ \tctx_j +_\kK \tctx_k$.
  We take the set $I \eqdef J \uplus K$, and note that it is finite.
  We conclude with 
  $\derivable\deriv{\judgs\tctx{\lam\var\tm}{\mtyp_1 + \mtyp_2}}$ 
  by applying rule \ruleTypAbs with $(\deriv'_i)_\iI$ as premises.
\end{itemize}
And note that in both cases, 
$\sizeU\deriv = \sizeU{\deriv_1} + \sizeU{\deriv_2}$ holds by reasoning 
like in the $(1 \implies 2)$ direction.
\end{proof}

The following property states that the controlled weakening on a sum of optional types can be equivalently decomposed into componentwise controlled weakenings:
\begin{lem}[Type Splitting]
\label{lem:vscne_type_splitting}
Let $\optmtyptwo_1$, $\optmtyptwo_2$ be two optional types and
$\mtyptwo$ be a type. 
Then $\optmtyptwo_1 + \optmtyptwo_2$ is well-defined and $\optmtyptwo_1 + \optmtyptwo_2 \mleq \mtyptwo$ if and only if 
there exist types $\mtyptwo_1$ and $\mtyptwo_2$ such that 
$\mtyptwo_1 + \mtyptwo_2 = \mtyptwo$ and $\optmtyptwo_i \mleq \mtyptwo_i$
for all $i \in \set{1,2}$.
\end{lem}
\begin{proof} 
Both directions of the lemma follow by case analysis on $\optmtyptwo_1$
and $\optmtyptwo_2$.
\end{proof}

\section{Characterisation of \texorpdfstring{\VSCne}{VSCne}-Termination}
\label{sec:systemU-vscne}

In this section, we show that $\typesystem$-typability characterises \VSCne-termination (\autoref{cor:u-characterises-vscne}): a term $\tm$ is typable in system $\typesystem$ (or
$\typesystem$-typable) if and only if $\tm$ terminates in \VSCne.
Later, in Section~\ref{sec:ucbv-U},
this same type system $\typesystem$ is shown to
characterise \UOCBV-termination as well, thus establishing a
relation between \VSCne-termination, \UOCBV-termination, and
$\typesystem$-typability.

To do so, we proceed as in~\autoref{subsec:vscne_V}, and prove soundness and completeness of the type system $\typesystem$ with respect to evaluation in \VSCne.

\paragraph{\textbf{Soundness.}}
We first address the \defn{soundness} of system $\typesystem$
with respect to \VSCne, which states that 
$\typesystem$-typability implies \VSCne-termination.
This proof follows well-understood techniques, requiring a subject 
reduction property (\autoref{prop:vscne_wsr_U}) based on the following substitution lemma: 

\begin{lem}[Weighted Substitution]
\label{lem:vscne_substitution}
Let $\derivablep{\deriv_\tm}{\judgs{\tctx;\var:\optmtyptwo}\tm\mtyp}\typesystem$
and $\derivablep{\deriv_\val}{\judgs\tctxtwo\val\mtyptwo}\typesystem$,
with $\optmtyptwo \mleq \mtyptwo$ and $\var \in \fv\tm$. 
Then, there exists a typing derivation $\deriv_{\tm\sub\var\val}$ 
such that 
$\derivablep{\deriv_{\tm\sub\var\val}}{\judgs{\tctx + \tctxtwo}{\tm\sub\var\val}\mtyp}\typesystem$
and
$\sizeU{\deriv_{\tm\sub\var\val}} = \sizeU{\deriv_\tm} + \sizeU{\deriv_\val} - \length\optmtyptwo$.
\end{lem}

\begin{proof}
By induction on $\deriv_\tm$.
We show only the \ruleTypVar case,
as the 
inductive cases follow 
easily from the~\ih and Lemmas~\ref{l:split-merge-value} and~\ref{lem:vscne_type_splitting}.

If $\deriv_\tm$ ends with \ruleTypVar, then $\tm = \var$, 
since $\var \in \fv\tm$ by hypothesis. 
Then, $\tm\sub\var\val = \val$ and $\deriv_\tm$ is of the form 
$\judgs{\var:\mtyp}\var\mtyp$, thus $\tctx = \emptyctx$ and 
$\optmtyptwo = \mtyptwo = \mtyp$. 
We let $\deriv_{\tm\sub\var\val} \eqdef \deriv_\val$, 
and we are done. 
Moreover,
$\sizeU{\deriv_{\tm\sub\var\val}} 
= \sizeU{\deriv_\val} 
= \length\mtyp + \sizeU{\deriv_\val} - \length\mtyp
= \sizeU{\deriv_\tm} + \sizeU{\deriv_\val} - \length\optmtyptwo$.
\end{proof}

The condition $\var \in \fv\tm$ in the lemma above is necessary, 
since the statement would not hold otherwise.
As a counterexample, take 
$\val = \varthree$, $\tctxtwo = \varthree : \tight$, and 
$\tm = \vartwo$ with $\vartwo \neq \var$.

A straightforward remark on values typing with $\tight$ in system $\typesystem$ is that the size of their typing derivations is always $0$:
\begin{rem}
\label{rem:vscne_size_deriv_val_tight}
If $\derivablep\deriv{\judgs\tctx\val\tight}\typesystem$, 
then $\sizeU\deriv = 0$.
\end{rem}

We now move on to the subject reduction property.
Similar to the analogous result established for system $\systemV$,
this property also provides quantitative information for system $\typesystem$: the size of $\typesystem$-typing derivations strictly decreases with respect to \VSCne-reduction steps.

\begin{prop}[Weighted Subject Reduction]
\label{prop:vscne_wsr_U}
Let $\derivablep\deriv{\judgs\tctx\tm\mtyp}\typesystem$ and 
$\tm \toVSCne \tm'$. 
Then, there exists a derivation $\deriv'$ such that 
$\derivablep{\deriv'}{\judgs\tctx{\tm'}\mtyp}\typesystem$ and
$\sizeU\deriv > \sizeU{\deriv'}$.
\end{prop}

\begin{proof}
By induction on $\tm \toVSCne \tm'$. 
We only show the base case 
$\tm \rewrite\svnesym \tm'$.
Let $\tm = \tmtwo\esub\var{\val\sctx} \rewrite\svnesym 
\tmtwo\sub\var\val\esub\var\val\sctx = \tm'$, with $\var \in \fv\tmtwo$. 
We proceed by induction on $\sctx$.
\begin{itemize}
\item $\sctx = \ctxhole$.
  Then, $\deriv$ has the form:
  \begin{mathpar}
    \inferrule* [right = \ruleTypES]{
      \derivable{\deriv_\tmtwo}{\judgs{\tctx_1;\var : \optmtyptwo}\tmtwo\mtyp}
      \\ \optmtyptwo \mleq \mtyptwo
      \\ \derivable{\deriv_\val}{\judgs{\tctx_2}\val\mtyptwo}
    }{
      \judgs{\tctx_1 + \tctx_2}{\tmtwo\esub\var\val}\mtyp
    }
  \end{mathpar}
  where $\tctx = \tctx_1 + \tctx_2$.
  Notice that $\optmtyptwo = \optmtyptwo + \none$.
  Thus, let $\optmtyptwo_1 = \optmtyptwo$ and $\optmtyptwo_2 = \none$.
  We have to separate into cases depending on whether $\mtyptwo = \tightN$ or
  $\mtyptwo = \nityp$, showing only the former.
  
    Hence $\mtyptwo = \tightN$, and 
    by~\autoref{lem:vscne_type_splitting} 
    there exist $\mtyptwo_1 = \tightN$ and $\mtyptwo_2 = \tightN$ 
    such that $\mtyptwo = \tightN = \tightN + \tightN = \mtyptwo_1 + \mtyptwo_2$
    (by idempotency of the union of constant types),
    and $\optmtyptwo \mleq \tightN$ and $\none \mleq \tightN$.
    Recall that $\length\tightN = \length\none = 0$ by definition, 
    so in particular $\length\optmtyptwo = \length{\optmtyptwo_1} = 0$.

    Applying~\autoref{l:split-merge-value}, we obtain 
    $\derivable{\deriv^1_\val}{\judgs{\tctx_{21}}\val\tightN}$ and
    $\derivable{\deriv^2_\val}{\judgs{\tctx_{22}}\val\tightN}$ with
    $\tctx_2 = \tctx_{21} + \tctx_{22}$ and $\sizeU{\deriv^1_\val} + \sizeU{\deriv^2_\val} = \sizeU{\deriv_\val} = 0$
    by~\autoref{rem:vscne_size_deriv_val_tight}.
    We can apply~\autoref{lem:vscne_substitution} on $\deriv_\tmtwo$
    and $\deriv^1_\val$, thus obtaining
    $\derivable
      {\deriv_{\tmtwo\sub\var\val}}
      {\judgs{\tctx_1 + \tctx_{21}}{\tmtwo\sub\var\val}\mtyp}$ and 
    $\sizeU{\deriv_{\tmtwo\sub\var\val}} = \sizeU{\deriv_\tmtwo} +
      \sizeU{\deriv^1_\val} - |\optmtyptwo_1| =  \sizeU{\deriv_\tmtwo}$.
    Thus, we conclude building $\deriv'$:
    \[
      \inferrule{
        \derivable
          {\deriv_{\tmtwo\sub\var\val}}
          {\judgs{(\tctx_1 + \tctx_{21});\var:\none}{\tmtwo\sub\var\val}\mtyp}
        \HS
        \none \mleq \tightN
        \HS
        \derivable{\deriv^2_\val}{\judgs{\tctx_{22}}\val\tightN}
      }{
        \judgs{\tctx_1 + \tctx_2}{\tmtwo\sub\var\val\esub\var\val}\mtyp
      }\ruleTypES
    \]

    And
    $\sizeU\deriv 
      = 1 + \sizeU{\deriv_\tmtwo} + \sizeU{\deriv_\val}
      = 1 + \sizeU{\deriv_\tmtwo}
      > \sizeU{\deriv_\tmtwo}
      =
         \sizeU{\deriv_{\tmtwo\sub\var\val}}
      = \sizeU{\deriv_{\tmtwo\sub\var\val}} + \sizeU{\deriv^2_{\val}}
      = \sizeU{\deriv'}$.
\item $\sctx = \sctxtwo\esub\vartwo\tmthree$.
  Then, $\deriv$ has the form:
  \[
    \scalebox{0.9}{
      \inferrule{
        \derivable{\deriv_\tmtwo}{\judgs{\tctx';\var:\optmtyptwo}\tmtwo\mtyp}
        \HS 
        \optmtyptwo \mleq \mtyptwo
        \HS
        \inferrule*[Right = \ruleTypES]{
          \derivable{\deriv_{\val\sctxtwo}}{\judgs{\tctxtwo_1;\vartwo:\optmtypthree}{\val\sctxtwo}\mtyptwo} 
          \HS
          \optmtypthree \mleq \mtypthree
          \HS
          \derivable{\deriv_\tmthree}{\judgs{\tctxtwo_2}\tmthree\mtypthree}
        }{ 
          \judgs{\tctxtwo_1 + \tctxtwo_2}{\val\sctxtwo\esub\vartwo\tmthree}\mtyptwo
        }
      }{
        \judgs{\tctx' + \tctxtwo_1 + \tctxtwo_2}{\tmtwo\esub\var{\val\sctxtwo\esub\vartwo\tmthree}}\mtyp 
      }\ruleTypES
    }
  \]
  where $\tctx = \tctx' + \tctxtwo_1 + \tctxtwo_2$.
  Moreover, we can assume $\vartwo \notin \fv\tmtwo$ by $\alpha$-conversion.

  We now  build the following derivation $\derivtwo$:
  \[
    \inferrule{
      \derivable{\deriv_\tmtwo}{\judgs{\tctx';\var : \optmtyptwo}\tmtwo\mtyp}
      \\ \optmtyptwo \mleq \mtyptwo 
      \\
      \derivable{
        \deriv_{\val\sctxtwo}
      }{
        \judgs{\tctxtwo_1;\vartwo:\optmtypthree}{\val\sctxtwo}\mtyptwo
      }
    }{
      \judgs{\tctx' + (\tctxtwo_1;\vartwo:\optmtypthree)}{\tmtwo\esub\var{\val\sctxtwo}}\mtyp
    }\ruleTypES
  \]
  Moreover, $\tmtwo\esub\var{\val\sctxtwo} \rewrite\svnesym \tmtwo\sub\var\val\esub\var\val\sctxtwo$,
  since $\var \in \fv\tmtwo$.
  Then, the \ih yields a typing derivation $\derivtwo'$ such that
  $\derivable{\derivtwo'}{\judgs{\tctx' + (\tctxtwo_1;\vartwo:\optmtypthree)}{\tmtwo\sub\var\val\esub\var\val\sctxtwo}\mtyp}$
  and moreover, $\sizeU\derivtwo > \sizeU{\derivtwo'}$. 
  Then, we build $\deriv'$ as follows: 
  \[
    \inferrule{
      \derivable{\derivtwo'}{\judgs{\tctx' + (\tctxtwo_1; \vartwo:\optmtypthree)}{\tmtwo\sub\var\val\esub\var\val\sctxtwo}\mtyp}
      \HS
      \optmtypthree \mleq \mtypthree 
      \HS
      \derivable{\deriv_\tmthree}{\judgs{\tctxtwo_2}\tmthree\mtypthree}
    }{
      \judgs{\tctx' + \tctxtwo_1 + \tctxtwo_2}{\tmtwo\sub\var\val\esub\var\val\sctxtwo\esub\vartwo\tmthree}\mtyp
    }\ruleTypES
  \]
  Note that we can assume $\vartwo \notin \dom{\tctx'}$ by 
  $\alpha$-conversion and~\autoref{lem:relevance_systemU}. 

  To compare the sizes of both $\deriv$ and $\deriv'$, 
  we consider two cases, depending on whether $\vartwo \in \fv{\tmtwo\sub\var\val\esub\var\val\sctxtwo}$.
  Note that $\var \in \fv\tmtwo$ by hypothesis.
  \begin{itemize}
  \item 
    If $\vartwo \notin \fv{\tmtwo\sub\var\val\esub\var\val\sctxtwo}$, 
    then in particular
    $\vartwo \notin \fv{\val\sctxtwo}$.
    Thus, we have
    $\sizeU\deriv 
    = 1 + \sizeU{\deriv_\tmtwo} + \sizeU{\deriv_{\val\sctxtwo}} + \sizeU{\deriv_\tmthree}
    = \sizeU{\derivtwo} + \sizeU{\deriv_\tmthree}
    >_\ih \sizeU{\derivtwo'} + \sizeU{\deriv_\tmthree}
    = \sizeU{\deriv'}$.
  \item
    Otherwise, $\vartwo \in \fv{\tmtwo\sub\var\val\esub\var\val\sctxtwo}$,
    and recall that $\vartwo \notin \fv\tmtwo$ by $\alpha$-conversion,
    so that  $\vartwo \in \fv{\val\sctxtwo}$.
    Therefore,
    $\sizeU\deriv 
    = 1 + \sizeU{\deriv_\tmtwo} + 1 + \sizeU{\deriv_{\val\sctxtwo}} + \sizeU{\deriv_\tmthree}
    = 1 + \sizeU\derivtwo + \sizeU{\deriv_\tmthree}
    >_\ih 1 + \sizeU{\derivtwo'} + \sizeU{\deriv_\tmthree}
    = \sizeU{\deriv'}$. 
    \qedhere
  \end{itemize}
\end{itemize}
\end{proof}

We can then state the following theorem:
\begin{thm}[Soundness of $\typesystem$ with Respect to \VSCne]
\label{thm:soundness_U_vscne}
Let t be a term.
If $\tm$ is $\typesystem$-typable, then $\tm$ is \VSCne-terminating.
\end{thm}
\begin{proof}
The proof is similar to that of~\autoref{thm:soundness_V_vscne} but uses~\autoref{prop:vscne_wsr_U}
instead of~\autoref{prop:vscne_wsr_V}.
\end{proof}

\paragraph{\textbf{Completeness.}}
We now prove the \defn{completeness} of system $\typesystem$ with 
respect to \VSCne, which states that \VSCne-termination
implies $\typesystem$-typability.
This result follows well-understood techniques, 
requiring a subject expansion property (\autoref{prop:vscne_wse_U}) based on an anti-substitution lemma (\autoref{lem:vscne_anti-subst}).
Additionally, it is necessary to prove that all \VSCne-normal forms are $\typesystem$-typable:

\begin{prop}[Normal Forms in \VSCne are $\typesystem$-Typable]
\label{prop:vscne_normal_forms_are_typable}
Let $\tm$ be a \VSCne-normal form.
Then, $\tm$ is $\typesystem$-typable with a constant type $\tight$.
More specifically, there exist $\deriv$, $\tctx$, and $\tight$
such that $\derivablep\deriv{\judgs\tctx\tm\tight}\typesystem$, where
$\tctx(\var) = \tightN$ for all $\var \in \dom\tctx$.
\end{prop}

\begin{proof}
  The proof is by simultaneous
  induction on the following statements:
\begin{enumerate}
\item \label{it:typing_vr}
  If $\tm \in \vvscne$, then there are $\deriv$ and $\tctx$
  such that $\derivablep\deriv{\judgs\tctx\tm{\tightN}}\typesystem$,
  and $\tctx(\vartwo) = \tightN$ for all $\vartwo \in \dom\tctx$. 
\item \label{it:typing_ne}
  If $\tm \in \nevscne$, then  there are $\deriv$ and $\tctx$ 
  such that $\derivablep\deriv{\judgs\tctx\tm{\tightN}}\typesystem$,
  and $\tctx(\vartwo) = \tightN$ for all $\vartwo \in \dom\tctx$. 
\item \label{it:typing_ab}
  If $\tm \in \absvscne$, then there are $\deriv$ and $\tctx$ 
  such that $\derivablep\deriv{\judgs\tctx\tm\emset}\typesystem$,
  and $\tctx(\vartwo) = \tightN$ for all $\vartwo \in \dom\tctx$. 
\item \label{it:typing_no}
  If $\tm \in \novscne$, then there are $\deriv$, $\tctx$, and $\tight$ such 
  that $\derivablep\deriv{\judgs\tctx\tm{\tight}}\typesystem$, and 
  $\tctx(\vartwo) = \tightN$ for all $\vartwo \in \dom\tctx$. \qedhere
\end{enumerate}
\end{proof}

We now state the anti-substitution lemma:
\begin{lem}[Weighted Anti-Substitution]
\label{lem:vscne_anti-subst}
Let 
$\derivablep{\deriv_{\tm\sub\var\val}}{\judgs\tctx{\tm\sub\var\val}\mtyp}\typesystem$, 
with $\var \in \fv\tm$.
Then, there exist typing derivations $\deriv_\tm$ and $\deriv_\val$,
typing environments $\tctx_\tm$ and $\tctxtwo$, an optional type 
$\optmtyptwo$, and a type $\mtyptwo$ satisfying:
\begin{enumerate}
\item 
  $\derivablep{\deriv_\tm}{\judgs{\tctx_\tm;\var : \optmtyptwo}\tm\mtyp}\typesystem$
\item
  $\derivablep{\deriv_\val}{\judgs{\tctxtwo}\val\mtyptwo}\typesystem$
\item
  $\tctx = \tctx_\tm + \tctxtwo$ and $\optmtyptwo \mleq \mtyptwo$
\end{enumerate}
Moreover, 
$\sizeU{\deriv_{\tm\sub\var\val}} 
= \sizeU{\deriv_\tm} + \sizeU{\deriv_\val} - \length\optmtyptwo$.
\end{lem}

\begin{proof}
  By structural induction on $\tm$. 
  We only show the case when $\tm$ is a variable. 
  The inductive cases follow easily from the \ih and~\autoref{l:split-merge-value}.

  When $\tm$ is a variable, then necessarily $\tm = \var$, 
  as $\var \in \fv\tm$ by hypothesis. 
  Then, $\tm\sub\var\val = \val$. 
  We let $\deriv_\val \eqdef \deriv_{\tm\sub\var\val}$, and we obtain 
  $\derivable{\deriv_\tm}{\judgs{\var:\mtyp}\var\mtyp}$ by rule \ruleTypVar,
  so that $\tctx_\tm \eqdef \emptyctx$, $\tctxtwo \eqdef \tctx$, and 
  $\optmtyptwo = \mtyptwo = \mtyp$. 
  Moreover, 
  $\sizeU{\deriv_{\tm\sub\var\val}} 
  = \sizeU{\deriv_\val}
  = \length\mtyptwo + \sizeU{\deriv_\val} - \length\mtyptwo
  = \sizeU{\deriv_{\tm}} + \sizeU{\deriv_\val} - \length\optmtyptwo$. \qedhere
\end{proof}

The condition $\var \in \fv\tm$ from the previous lemma is necessary, 
since otherwise, the statement does not hold.
As a counterexample, it is sufficient to take $\tm = \vartwo$ with 
$\vartwo \neq \var$ and $\val = \varthree$ so that 
$\tm\sub\var\varthree = \vartwo$ and the derivation
$\derivable{\deriv_{\tm\sub\var\varthree}}{\judgs{\vartwo:\mtyp}\vartwo\mtyp}$.

We now move on to the subject expansion property, similar to the one established for system $\systemV$, since it
also provides quantitative information for system $\typesystem$: the size of $\typesystem$-typing derivations strictly increases with respect to \VSCne-expansions.

\begin{prop}[Weighted Subject Expansion]
\label{prop:vscne_wse_U}
Let $\tm \toVSCne \tm'$ and 
$\derivablep{\deriv'}{\judgs\tctx{\tm'}\mtyp}\typesystem$. 
Then, there exists $\deriv$ such that 
$\derivablep\deriv{\judgs\tctx\tm\mtyp}\typesystem$, with
$\sizeU\deriv > \sizeU{\deriv'}$.
\end{prop}
\begin{proof}
By induction on $\tm \toVSCne \tm'$ where, in particular,~\autoref{lem:vscne_anti-subst} is used in the base case
$\tm = \tmtwo\esub\var{\val\sctx} \rewrite\svnesym
\tmtwo\sub\var\val\esub\var\val\sctx = \tm'$.
\end{proof}

We can then state the following theorem:
\begin{thm}[Completeness of $\typesystem$ with Respect to \VSCne]
\label{thm:completeness_U_vscne}
Let t be a term.
If $\tm$ is \VSCne-terminating, then $\tm$ is $\typesystem$-typable.
\end{thm}
\begin{proof}
The proof is similar to that of~\autoref{thm:completeness_V_vscne} but uses~\autoref{prop:vscne_normal_forms_are_typable} and~\autoref{prop:vscne_wse_U}
instead of~\autoref{prop:vscne_normal_forms_are_typable_V}
and~\autoref{prop:vscne_wse_V}.
\end{proof}

From soundness (\autoref{thm:soundness_U_vscne}) and
completeness (\autoref{thm:completeness_U_vscne}), we can now state
the result of the characterisation of \VSCne-termination via
$\typesystem$-typability:
\begin{cor}[System $\typesystem$ Characterises \VSCne-Termination]
\label{cor:u-characterises-vscne}
Let $\tm$ be a term. Then, $\tm$ is $\typesystem$-typable if and only 
if $\tm$ is \VSCne-terminating. 
\end{cor}

From the corollary above and~\autoref{cor:v-characterises-vscne}, we obtain an interesting equivalence between $\systemV$ and
$\typesystem$:
\begin{cor}[$\systemV$-Typability is Equivalent to $\typesystem$-Typability]
\label{cor:v-and-u}
Let $\tm$ be a term. Then, $\tm$ is $\systemV$-typable if and only 
if $\tm$ is $\typesystem$-typable.
\end{cor}

\section{Characterisation of \texorpdfstring{\UOCBV}{UCBV}-Termination}
\label{sec:ucbv-U}

In this section, we refine system $\typesystem$ to capture
quantitative information of \UOCBV evaluation. 
Specifically, we equip typing judgements of system $\typesystem$ with counters, and we define
a notion of tight derivation, intuitively corresponding to that of
\emph{minimal} derivation. 
The counters are meant to capture the exact step count for program evaluation. 
More precisely, tight derivations in system $\typesystem$ not only guarantee termination in \UOCBV but
also provide \emph{exact} quantitative information about this
evaluation strategy: we show that a term $\tm$ is tightly typable with
counters $(\cm,\ce)$ if and only if $\tm$ evaluates in \UOCBV to a
normal form in exactly $\cm$ function application steps (\ie
$\ruledb$-steps) and $\ce$ substitution steps (\ie $\rulelsv$-steps).
In this sense, we say that system $\typesystem$ is a
\emph{quantitative interpretation} of useful \CBV.

The section is organised as follows.
In~\autoref{subsec:U_counters}, we equip typing judgements of system $\typesystem$ with counters and outline basic properties of the resulting system
that will come in handy to quantitatively characterise \UOCBV evaluation.
Then in~\autoref{subsec:uocbv_U}, we show that
$\typesystem$-typability qualitatively characterises \UOCBV-termination
(\autoref{cor:qualitative_characterisation_of_termination}),
while tight $\typesystem$-typability quantitatively characterises \UOCBV-termination
(\autoref{cor:quantitative_characterisation_of_termination}).

\subsection{Adding Counters to System \texorpdfstring{$\typesystem$}{U}} 
\label{subsec:U_counters}
As explained above, we start by decorating the typing judgements of 
system $\typesystem$ with counters. 
Formally, \defn{typing judgements} are now of the form 
$\judg[\cm,\ce]\tctx\tm\mtyp$, where $\cm$ and $\ce$ are natural 
numbers called \defn{counters}. 
Under appropriate conditions, these counters correspond 
\emph{exactly} to the number of $\dbsym$-steps ($\cm$) and 
$\lsvsym$-steps ($\ce$) in \UOCBV evaluation, as shown in 
Theorems~\ref{thm:quantitative_soundness_typing} and~\ref{thm:quantitative_completeness_typing}. 
Typing rules are refined accordingly in~\autoref{fig:systemU_counters}. 
It should be noted that counters are a convenient way to synthesise information \emph{already
existing} in a typing derivation: they do
not impose any side conditions on derivations.
Formally, $\derivable{}{\judgs\tctx\tm\mtyp}$ if and only if 
$\derivable{}{\judg[\cm,\ce]\tctx\tm\mtyp}$ for some $(\cm,\ce)$.
Moreover, the counters can be entirely extracted from the type information and typing rules of a derivation; 
no reference to the terms themselves is needed to compute them. 

We introduce an operation $\numarrows{\_}$ that counts the 
number of \texttt{t}op-level \texttt{a}rrows in a type (or optional type). 
It is defined by $\numarrows\none \eqdef 0$, 
$\numarrows\tightN \eqdef 0$, and 
$\numarrows{\mset{\typ_i}_\iI} \eqdef \length{I}$. 
Note that $\numarrows{\mtyp_1 + \mtyp_2 } 
= \numarrows{\mtyp_1} + \numarrows{\mtyp_2}$. 

\begin{figure}[t]
  \begin{mathpar}
    \inferrule* [right = \ruleTypVar]{
      n = \numarrows\mtyp
    }{
      \judg[0,n]{\var:\mtyp}\var\mtyp
    }

    \inferrule* [right = \ruleTypAbs]{
      \left(\judg[\cm_i,\ce_i]{\tctx_i;\var:\optmtyp_i}\tm{\mtyptwo_i}\right)_\iI
    }{
      \judg[+_\iI \cm_i, +_\iI \ce_i]
        {+_\iI\tctx_i}{\lam\var\tm}{\mset{\optmtyp_i\to\mtyptwo_i}_\iI}
    }

    \inferrule* [right = \ruleTypAppC]{
      \judg[\cm,\ce]\tctx\tm{\mset{\optmtyp\to\mtyptwo}}
      \\ \optmtyp \mleq \mtyp
      \\ \judg[\cm',\ce']\tctxtwo\tmtwo\mtyp
    }{
      \judg[1 + \cm + \cm', \ce + \ce']{\tctx + \tctxtwo}{\tm\,\tmtwo}\mtyptwo
    }

    \inferrule* [right = \ruleTypAppP]{
      \judg[\cm,\ce]\tctx\tm\tightN
      \\ \judg[\cm',\ce']\tctxtwo\tmtwo\tight
    }{
      \judg[\cm + \cm', \ce + \ce']{\tctx + \tctxtwo}{\tm\,\tmtwo}\tightN
    }

    \inferrule* [right = \ruleTypES]{
      \judg[\cm,\ce]{\tctx;\var:\optmtyp}\tm\mtyptwo
      \\ \optmtyp \mleq \mtyp
      \\ \judg[\cm',\ce']\tctxtwo\tmtwo\mtyp
    }{
      \judg[\cm + \cm', \ce + \ce']{\tctx+\tctxtwo}{\tm\esub\var\tmtwo}\mtyptwo
    }
  \end{mathpar}
  \caption{Typing Rules of System $\typesystem$ with Counters}
  \label{fig:systemU_counters}
\end{figure}

Only two typing rules increase the counters.
Rule \ruleTypVar initialises the second counter to the number of 
top-level arrows in the type of the variable.
This counter corresponds to the number of times the variable is 
(directly or indirectly) substituted by an abstraction taking part in 
some $\dbsym$-redex during evaluation.
That is why the second counter tracks substitution step-count. 
Rule \ruleTypAppC increases the first counter by one, to reflect that 
the application constructor introduced by this rule is \emph{consuming}, 
meaning it will be eliminated during evaluation.
That is why the first counter tracks function application step-count.
\medskip

We now outline some basic properties of system $\typesystem$.
\paragraph{\textbf{Appropriateness.}} 
To reason inductively about typing derivations, it is sometimes 
necessary to ensure invariants for the types of the free variables 
in a term.
For example, for a term like $\tm\esub\var\id$, we may need to track 
that $\var$ is bound to an abstraction when applying the inductive 
hypothesis for the subterm $\tm$. 
Specifically, we say that a typing environment $\tctx$ is
\defn{appropriate} with respect to an abstraction frame $\aset$,
written $\isAppr\aset\tctx$, if, for each $\var \in \aset$, we have 
$\tctx(\var) \neq \tightN$.
In other words, $\tctx(\var)$ must be $\none$ or an arrow multi-type 
$\nityp$. For instance, $\isAppr{\set\var}{\var: \emset, \vartwo:\tightN}$ holds, while
$\isAppr{\set\var}{ \var:\tightN,\vartwo:\tightN}$ does not.

\paragraph{\textbf{Tightness.}}
Typing in system $\typesystem$ not only characterises \UOCBV-termination, 
but it is also capable of providing \emph{exact measures} for 
reduction lengths to normal form in \UOCBV. 
Said in other words, tightness is essential for proving that system $\typesystem$ is \emph{exactly}
sound and complete with respect to \UOCBV.
To state this formally, we define a subset of typing derivations, 
called \defn{tight derivations}, as follows:
\begin{itemize}
\item 
  A \defn{type} is \defn{tight} if it is a constant type $\tight \in \set{\tightN,\emset}$.
\item 
  An \defn{optional type} is \defn{tight} if it is either $\none$ or a tight type.
\item
  A \defn{typing environment} $\tctx$ is \defn{tight} if $\tctx(\var)$ 
  is tight for every variable $\var \in \dom\tctx$.
\item
  A \defn{typing judgement} $\judg[\cm,\ce]\tctx\tm\mtyp$ is \defn{tight} 
  if both $\tctx$ and $\mtyp$ are tight.
\item
  A \defn{typing derivation} $\derivable\deriv{\judg[\cm,\ce]\tctx\tm\mtyp}$ 
  is \defn{tight} if its conclusion $\judg[\cm,\ce]\tctx\tm\mtyp$ is tight.
\end{itemize}

For example, the derivable judgement 
$\judg[0,0]{\var : \emset}\var\emset$ is tight, where the counters 
$(0,0)$ indicate that evaluating $\var$ requires no reduction steps,
meaning that $\var$ is already in \UOCBV-normal form. 
In contrast, the derivable judgement 
$\judg[0,1]{\var : \mset{\mtyp\to\mtyptwo}}\var{\mset{\mtyp\to\mtyptwo}}$ 
is \emph{not} tight, and the counters $(0,1)$ provide an upper bound 
on the number of steps required for evaluating $\var$.

Let us now add counters to the typing derivation $\derivtwo$ (\cf~\autoref{ex:derivation}) typing our working~\autoref{exa:working-example}.

\begin{exa}
\label{ex:derivation-for-tight}
Let $\mtyp \eqdef \mset{\tightN \to \tightN}$. We build the following 
derivation $\derivtwo'$, which turns out to be tight. We omit the rule names for improving readability.
\begin{mathpar}
  \inferrule* {
    \inferrule* {
      \inferrule* {
      }{
        \judg[0,1]{\var : \mtyp}\var\mtyp
      }
      \\
      \inferrule* {
      }{
        \judg[0,0]{\varthree : \tightN}\varthree\tightN
      }
    }{
      \judg[1,1]
        {\varthree : \tightN, \var : \mtyp}
        {\var \, \varthree}
        \tightN
    }
    \sep
    \inferrule* {
      \inferrule* {
      }{
        \judg[0,1]{\vartwo : \mtyp}\vartwo\mtyp
      }
      \\
      \inferrule* {
        \inferrule* {
        }{
          \judg[0,0]{\varfour:\tightN}\varfour\tightN
        }
      }{
        \judg[0,0]\emptyctx\id\mtyp
      }
    }{
      \judg[0,1]\emptyctx{\vartwo\esub\vartwo\id}\mtyp
    }
  }{
    \judg[1,2]
      {\varthree : \tightN}
      {(\var \, \varthree)\esub\var{\vartwo\esub\vartwo\id}}
      \tightN
  }
\end{mathpar}
where $\tightN \mleq \tightN$ and 
$\mtyp = \mset{\tightN \to \tightN} \mleq \mset{\tightN \to \tightN} = \mtyp$
both hold.

Note that the only subderivations of $\derivtwo'$ which are tight 
(besides $\derivtwo'$ itself) correspond to the one concluding in 
$\judg[0,0]{\varthree : \tightN}\varthree\tightN$ and the one
concluding in $\judg[0,0]{\varfour : \tightN}\varfour\tightN$.
The remaining subderivations are not tight.
\end{exa}

To conclude this subsection, we state a variant of~\autoref{l:split-merge-value}
that takes into account the quantitative information registered in the counters:

\begin{lem}[Splitting/Merging of System $\typesystem$ with Counters]
\label{lem:systemU_split-merge-value_counters}
Let $\val$ be a value, $\tctx$ a typing environment, $(\cm,\ce)$ a counter, and $\mtyp_1, \mtyp_2$ be two types such 
that their union $\mtyp_1 + \mtyp_2$ is well-defined.
Then, the following are equivalent:
\begin{enumerate}
\item 
  There exists a derivation $\deriv$ such that
  $\derivable\deriv{\judg[\cm,\ce]\tctx\val{\mtyp_1 + \mtyp_2}}$.
\item 
  There exist derivations $\deriv_1,\deriv_2$, typing 
  environments $\tctx_1,\tctx_2$, and counters $(\cm_1,\ce_1),(\cm_2,\ce_2)$ 
  such that
  $\derivable{\deriv_1}{\judg[\cm_1,\ce_1]{\tctx_1}\val{\mtyp_1}}$, and
  $\derivable{\deriv_2}{\judg[\cm_2,\ce_2]{\tctx_2}\val{\mtyp_2}}$ and
  $\tctx = \tctx_1 + \tctx_2$ and $\cm = \cm_1 + \cm_2$ and 
  $\ce = \ce_1 + \ce_2$.
  \end{enumerate}
  \end{lem}
\begin{proof}
The proof is analogous to that of~\autoref{l:split-merge-value}.
\end{proof}

\subsection{Characterisation of \texorpdfstring{\UOCBV}{UCBV}-Termination Through \texorpdfstring{$\mathbfcal{U}$}{U}-Typability}
\label{subsec:uocbv_U}

We first address the \defn{soundness} of system $\typesystem$ with 
respect to \UOCBV evaluation, stating that $\typesystem$-typability implies \UOCBV-termination.
This proof follows well-understood techniques, requiring a subject 
reduction property (\autoref{prop:subject_reduction}) based on a \emph{linear} substitution lemma (\autoref{lem:linear_substitution}).
Due to the \emph{contextual} nature of \UOCBV, these lemmas are not 
trivial, as they demand formulating complex invariants involving the 
contextual parameters $\aset$, $\sset$, and $\appflag$ that define 
the evaluation strategy.

The following property ensures that the type assignment of hereditary abstractions is well-behaved, 
in the sense that the type system cannot assign the constant $\tightN$,
aimed at typing potential structures, 
to objects that are potential abstractions.

\begin{lem}[Types of Hereditary Abstractions]
\label{lem:systemU_types_of_hereditary_abstractions}
Let $\aset$ be an abstraction frame and $\tm$ a term such that 
$\tm \in \HAbs\aset$.
Let $\derivable\deriv{\judg[\cm,\ce]\tctx\tm\mtyp}$
with $\isAppr\aset\tctx$.
Then, $\mtyp \neq \tightN$.
\end{lem}

\begin{proof}
By induction on $\deriv$. 
In cases \ruleTypAppP and \ruleTypAppC, the term $\tm$ is an 
application, but at the same time the hypothesis states that $\tm$ is 
a hereditary abstraction under $\aset$, which is impossible.
Therefore, these cases do not hold.
\begin{itemize}
\item \ruleTypVar.
  Then, $\derivable\deriv{\judg[0,\numarrows\mtyp]{\var : \mtyp}\var\mtyp}$. 
  Moreover, the hypothesis $\var \in \HAbs\aset$ can only be obtained 
  by the first rule in~\autoref{fig:HAbs_Struct} (above), thus $\var \in \aset$.
  Since $\isAppr\aset{\var : \mtyp}$ holds, then $\mtyp \neq \tightN$ 
  by definition.
\item \ruleTypAbs.
  This case is immediate since
  $\derivable
    \deriv{
      \judg[\sum_\iI \cm_i, \sum_\iI \ce_i]
      {+_\iI \tctx_i}
      {\lam\var\tmtwo}
      {\mset{\optmtyptwo_i \to \mtypthree_i}_\iI}}$, 
  where in particular 
  $\mtyp$ is the arrow multi-type $\mset{\optmtyptwo_i \to \mtypthree_i}_\iI$.
\item \ruleTypES.
  Then, $\deriv$ is of the form:
  \begin{mathpar}
    \inferrule
      {\derivable{\deriv_\tmtwo}{\judg[\cm_1,\ce_1]{\tctx_1; \var : \optmtyptwo}\tmtwo\mtyp}
        \\ \optmtyptwo \mleq \mtyptwo
        \\ \derivable{\deriv_\tmthree}{\judg[\cm_2,\ce_2]{\tctx_2}\tmthree\mtyptwo}}
      {\judg[\cm_1 + \cm_2,\ce_1 + \ce_2]{\tctx_1 + \tctx_2}{\tmtwo\esub\var\tmthree}\mtyp}
  \end{mathpar}
  where $\tctx = \tctx_1 + \tctx_2$, $\cm = \cm_1 + \cm_2$,
  $\ce = \ce_1 + \ce_2$, and $\tm = \tmtwo\esub\var\tmthree$.
  Moreover, the hypothesis $\tmtwo\esub\var\tmthree \in \HAbs\aset$ 
  can be derived by the last two rules in~\autoref{fig:HAbs_Struct} (above):
  \begin{itemize}
  \item 
    If we are in the case
    \[
    \inferrule
      {\tmtwo \in \HAbs\aset \\ \var \notin \aset}
      {\tmtwo\esub\var\tmthree \in \HAbs\aset}
    \]
    given that $\dom{\tctx_1} \subseteq \dom\tctx$, then 
    $\isAppr\aset\tctx$ implies $\isAppr\aset{\tctx_1}$.
    Thus, $\isAppr\aset{\tctx_1; \var : \optmtyptwo}$ holds.
    And the \ih on $\deriv_\tmtwo$ yields $\mtyp \neq \tightN$.
  \item 
    If we are in the case
    \[
    \inferrule{
      \tmtwo \in \HAbs{\aset\cup\set\var}
      \\ \var\notin\aset
      \\ \tmthree \in \HAbs\aset
    }{
       \tmtwo\esub\var\tmthree \in \HAbs\aset
    }
    \]
    given that $\dom{\tctx_2} \subseteq \dom\tctx$, then 
    $\isAppr\aset\tctx$ implies $\isAppr\aset{\tctx_2}$.
    The \ih on $\deriv_\tmthree$ yields $\mtyptwo \neq \tightN$.
    Since $\optmtyptwo \mleq \mtyptwo$, then $\optmtyptwo \neq \tightN$.
    Moreover, $\isAppr\aset\tctx$ implies $\isAppr\aset{\tctx_1}$,
    so $\isAppr{\aset \cup \set\var}{\tctx_1; \var : \optmtyptwo}$ also holds.
    Applying the \ih on $\deriv_\tmtwo$, we conclude with $\mtyp \neq \tightN$. 
    \qedhere
  \end{itemize}
\end{itemize}
\end{proof}

We can now state the linear substitution lemma:
\begin{lem}[Exact Linear Substitution]
\label{lem:linear_substitution}
Let $\tm \tov{\rulesub\var\val}{\aset \cup \set\var}\sset\appflag \tm'$
and let $\deriv_\tm$ be 
$\derivable{\deriv_\tm}{\judg[\cm,\ce]{\tctx;\var : \optmtyptwo_\tm}\tm\mtyp}$
such that $\isAppr{\aset \cup \set\var}{\tctx; \var : \optmtyptwo_\tm}$.
Suppose moreover that if $\appflag = \app$, then either $\mtyp = \tightN$ 
or $\mtyp = \mset{\typ}$.
Then, there exists a splitting $\optmtyptwo_\tm = \mset{\typtwo} + \optmtyptwo_{\tm'}$
such that for every derivation $\derivable{\deriv_\val}{\judg[\cmtwo,\cetwo]\tctxtwo\val{\mset\typtwo}}$,
we have $\ce > 0$ and $\derivable{\deriv_{\tm'}}{\judg[\cm + \cm', \ce + \ce' - 1]{\tctx + \tctxtwo; \var : \optmtyptwo_{\tm'}}{\tm'}\mtyp}$.
\end{lem}

\begin{proof}
By induction on $\tm \tov{\rulesub\var\val}{\aset \cup \set\var}\sset\appflag \tm'$.
We only show cases \ruleUSub, \ruleUAppL, and \ruleUEsLStruct.

\begin{itemize}
\item \ruleUSub. 
  Then, $\tm = \var \tov{\rulesub\var\val}{\aset \cup \set\var}\sset\app \val = \tm'$,
  where $\appflag = \app$.
  On the other hand, $\deriv_\tm$ can be only obtained from the 
  \ruleTypVar axiom, so it is of the form 
  $\judg[0,\numarrows\mtyp]{\var : \mtyp}\var\mtyp$, where 
  $\tctx = \emptyctx$, $\optmtyptwo_\tm = \mtyp$, $\cm = 0$, and 
  $\ce = \numarrows\mtyp$.

  Then, $\isAppr{\aset \cup \set\var}{\var : \mtyp}$ implies
  $\mtyp \neq \tightN$.
  This inequality and the fact that $\appflag = \app$ imply 
  $\mtyp = \mset\typ$, for some arrow type $\typ$.
  Therefore, $\ce = \numarrows\mtyp = 1 > 0$.

  Let us take $\typtwo \eqdef \typ$ and 
  $\optmtyptwo_{\tm'} \eqdef \none$ so that 
  $\optmtyptwo_\tm = \mtyp = \mset\typ = \mset\typ + \none$.
  We conclude taking $\deriv_{\tm'} \eqdef \deriv_\val$, whenever
  $\derivable{\deriv_\val}{\judg[\cmtwo,\cetwo]\tctxtwo\val{\mset\typ}}$.
\item \ruleUAppL.
  Then, the reduction step is of the form
  \begin{mathpar}
    \inferrule
      {\tmtwo \tov{\rulesub\var\val}{\aset \cup \set\var}\sset\app \tmtwo'}
      {\tm = \tmtwo \, \tmthree \tov{\rulesub\var\val}{\aset \cup \set\var}\sset\appflag \tmtwo' \, \tmthree = \tm'}
  \end{mathpar}

  The derivation $\deriv_\tm$ can be obtained either by rule 
  \ruleTypAppP or rule \ruleTypAppC.
  We consider the former, since the latter is analogous.
  Therefore, $\deriv_\tm$ is of the form:
  \begin{mathpar}
    \inferrule* [right = \ruleTypAppP]
      {\derivable
        {\deriv_\tmtwo}
        {\judg[\cm_\tmtwo,\ce_\tmtwo]{\tctx_\tmtwo; \var:\optmtyptwo_\tmtwo}\tmtwo\tightN}
        \\
        \derivable{\deriv_\tmthree}{\judg[\cm_\tmthree,\ce_\tmthree]{\tctx_\tmthree; \var:\optmtyptwo_\tmthree}\tmthree\tight}}
      {\judg[\cm_\tmtwo + \cm_\tmthree, \ce_\tmtwo + \ce_\tmthree]
        {(\tctx_\tmtwo + \tctx_\tmthree); \var : (\optmtyptwo_\tmtwo + \optmtyptwo_\tmthree)}
        {\tmtwo \, \tmthree}
        \tightN}
  \end{mathpar}
  where $\tctx = \tctx_\tmtwo + \tctx_\tmthree$,
  $\optmtyptwo_\tm = \optmtyptwo_\tmtwo + \optmtyptwo_\tmthree$,
  $\cm = \cm_\tmtwo + \cm_\tmthree$, $\ce = \ce_\tmtwo + \ce_\tmthree$,
  and $\mtyp = \tightN$.

  Since the typing environment $\tctx_\tmtwo; \var:\optmtyptwo_\tmtwo$ is 
  a subset of 
  $(\tctx_\tmtwo + \tctx_\tmthree); \var : (\optmtyptwo_\tmtwo + \optmtyptwo_\tmthree)$,
  then
  $\isAppr
    {\aset \cup \set\var}
    {(\tctx_\tmtwo + \tctx_\tmthree); \var : (\optmtyptwo_\tmtwo + \optmtyptwo_\tmthree)}$
  implies 
  $\isAppr{\aset \cup \set\var}{\tctx_\tmtwo; \var:\optmtyptwo_\tmtwo}$.
  Moreover, $\appflag = \app$ in the reduction step 
  $\tmtwo \tov{\rulesub\var\val}{\aset \cup \set\var}\sset\app \tmtwo'$.
  The conclusion of $\deriv_\tmtwo$ types with the constant $\tightN$.
  Therefore, the \ih on 
  $\tmtwo \tov{\rulesub\var\val}{\aset \cup \set\var}\sset\app \tmtwo'$ 
  yields a splitting 
  $\optmtyptwo_\tmtwo = \mset\typtwo + \optmtyptwo_{\tmtwo'}$
  such that for every 
  $\derivable{\deriv_\val}{\judg[\cmtwo,\cetwo]\tctxtwo\val{\mset\typtwo}}$,
  we have $\ce_\tmtwo > 0$ and 
  $\derivable
    {\deriv_{\tmtwo'}}
    {\judg
      [\cm_\tmtwo + \cmtwo, \ce_\tmtwo + \cetwo - 1]
      {(\tctx_\tmtwo + \tctxtwo); \var : \optmtyptwo_{\tmtwo'}}
      {\tmtwo'}
      \tightN}$.
  
  We take $\optmtyptwo_{\tm'} \eqdef 
  \optmtyptwo_{\tmtwo'} + \optmtyptwo_\tmthree$ and applying 
  rule \ruleTypAppP to $\deriv_{\tmtwo'}$ and $\deriv_\tmthree$ we obtain 
  $\derivable
    {\deriv_{\tm'}}
    {\judg
      [\cm_\tmtwo + \cmtwo + \cm_\tmthree, \ce_\tmtwo + \cetwo + \ce_\tmthree - 1]
      {(\tctx_\tmtwo + \tctxtwo + \tctx_\tmthree); \var : (\optmtyptwo_{\tmtwo'} + \optmtyptwo_\tmthree)}
      {\tmtwo' \, \tmthree}
      \tightN}$.

  Note that $\ce = \ce_\tmtwo + \ce_\tmthree > 0$ since the \ih yields $\ce_\tmtwo > 0$.
  Moreover, 
  $\optmtyptwo_\tm = \optmtyptwo_\tmtwo + \optmtyptwo_\tmthree 
  =_\ih \mset\typtwo + (\optmtyptwo_{\tmtwo'} + \optmtyptwo_\tmthree)
  = \mset\typtwo + \optmtyptwo_{\tm'}$.
\item \ruleUEsLStruct.
  Then, the reduction step is of the form:
  \begin{mathpar}
    \inferrule{
      \tmtwo \tov{\rulesub\var\val}{\aset \cup \set\var}{\sset \cup \set\vartwo}\appflag \tmtwo'
      \\ \tmthree \in \Struct\sset
      \\ \vartwo \notin (\aset \cup \set\var) \cup \sset
      \\ \vartwo \notin \fv{\rulesub\var\val}
    }{
      \tm = \tmtwo\esub\vartwo\tmthree \tov{\rulesub\var\val}{\aset \cup \set\var}\sset\appflag \tmtwo'\esub\vartwo\tmthree = \tm'
    }
  \end{mathpar}

  By $\alpha$-conversion, $\var \neq \vartwo$.
  The derivation $\deriv_\tm$ can only be obtained from rule
  \ruleTypES, so it is of the form:
  \begin{mathpar}
    \inferrule
      {\derivable{\deriv_\tmtwo}
        {\judg
          [\cm_\tmtwo,\ce_\tmtwo]
          {\tctx_\tmtwo; \var : \optmtyptwo_\tmtwo; \vartwo: \optmtypthree}
          \tmtwo
          \mtyp}
        \\ \optmtypthree \mleq \mtypthree
        \\ \derivable{\deriv_\tmthree}
            {\judg[\cm_\tmthree,\ce_\tmthree]{\tctx_\tmthree; \var : \optmtyptwo_\tmthree}\tmthree\mtypthree}}
      {\judg[\cm_\tmtwo + \cm_\tmthree, \ce_\tmtwo + \ce_\tmthree]
        {(\tctx_\tmtwo + \tctx_\tmthree); \var : (\optmtyptwo_\tmtwo + \optmtyptwo_\tmthree)}
        {\tmtwo\esub\vartwo\tmthree}
        \mtyp}
  \end{mathpar}
  where $\tctx = \tctx_\tmtwo + \tctx_\tmthree$,
  $\optmtyptwo_\tm = \optmtyptwo_\tmtwo + \optmtyptwo_\tmthree$,
  $\cm = \cm_\tmtwo + \cm_\tmthree$ and $\ce = \ce_\tmtwo + \ce_\tmthree$.

  Since the domain of the typing environment 
  $\tctx_\tmtwo; \var : \optmtyptwo_\tmtwo$ is a subset of the domain 
  of the typing environment $(\tctx_\tmtwo + \tctx_\tmthree); 
  \var : (\optmtyptwo_\tmtwo + \optmtyptwo_\tmthree)$, then 
  $\isAppr{\aset \cup \set\var}
    {(\tctx_\tmtwo + \tctx_\tmthree); \var : (\optmtyptwo_\tmtwo + \optmtyptwo_\tmthree)}$
  and $\vartwo \notin (\aset \cup \set\var) \cup \sset$
  imply $\isAppr{\aset \cup \set\var}{\tctx_\tmtwo; \var : \optmtyptwo_\tmtwo; \vartwo : \optmtypthree}$.
  Additionally, if $\appflag = \app$ in the reduction step
  $\tmtwo \tov{\rulesub\var\val}{\aset \cup \set\var}{\sset \cup \set\vartwo}\appflag \tmtwo'$, 
  then either $\mtyp = \tightN$ or $\mtyp = \mset\typ$ by hypothesis.
  Then, the \ih on $\tmtwo$ yields a splitting 
  $\optmtyptwo_\tmtwo = \mset\typtwo + \optmtyptwo_{\tmtwo'}$ 
  such that for every 
  $\derivable{\deriv_\val}{\judg[\cmtwo,\cetwo]\tctxtwo\val{\mset\typtwo}}$,
  we have $\ce_\tmtwo > 0$ and 
  $\derivable{\deriv_{\tmtwo'}}
    {\judg
      [\cm_\tmtwo + \cmtwo, \ce_\tmtwo + \cetwo - 1]
      {((\tctx_\tmtwo; \vartwo : \optmtypthree) + \tctxtwo); \var : \optmtyptwo_{\tmtwo'}}
      {\tmtwo'}
      \mtyp}$.
  Since $\vartwo \notin \fv\val$, then 
  $\vartwo \notin \dom\tctxtwo$ by~\autoref{lem:relevance_systemU},
  so we can write 
  $((\tctx_\tmtwo; \vartwo : \optmtypthree) + \tctxtwo); \var : \optmtyptwo_{\tmtwo'}$ 
  as
  $(\tctx_\tmtwo + \tctxtwo); \var : \optmtyptwo_{\tmtwo'}; \vartwo : \optmtypthree$.

  We take $\optmtyptwo_{\tm'} \eqdef 
  \optmtyptwo_{\tmtwo'} + \optmtyptwo_\tmthree$ and applying
  rule \ruleTypES to $\deriv_{\tmtwo'}$ and $\deriv_\tmthree$, we obtain
  \[
  \derivable{\deriv_{\tm'}}
    {\judg
      [\cm_\tmtwo + \cmtwo + \cm_\tmthree, \ce_\tmtwo + \cetwo + \ce_\tmthree - 1]
      {(\tctx_\tmtwo + \tctxtwo + \tctx_\tmthree); \var : (\optmtyptwo_{\tmtwo'} + \optmtyptwo_\tmthree)}
      {\tmtwo'\esub\vartwo\tmthree}
      \mtyp}
  \]

  Note that $\ce = \ce_\tmtwo + \ce_\tmthree > 0$ since the \ih yields $\ce_\tmtwo > 0$.
  Moreover, $\optmtyptwo_\tm 
  = \optmtyptwo_\tmtwo + \optmtyptwo_\tmthree
  =_\ih (\mset\typtwo + \optmtyptwo_{\tmtwo'}) + \optmtyptwo_\tmthree
  = \mset\typtwo + \optmtyptwo_{\tm'}$.
  \qedhere
\end{itemize}
\end{proof}

We illustrate by the following example the lemma above.
\begin{exa}
Recall the typing derivation
$\derivable{\derivtwo'}{\judg[1,2]
  {\varthree : \tightN}
  {(\var \, \varthree)\esub\var{\vartwo\esub\vartwo\id}}
  {\tightN}}$
from~\autoref{ex:derivation-for-tight}.
Let us consider its subderivation typing
the subterm $\var \, \varthree$ with counters $(1,1)$.
The term $\var \, \varthree$ 
$(\rulesub\var\vartwo,\set\var,\set\varthree,\nonapp)$-reduces in \UOCBV to 
$\vartwo \, \varthree$.

Consider the splitting $\optmtyptwo_\tm \eqdef \mtyp + \none$ (which is equal to $\mtyp$),
where we recall that $\mtyp$ is the one-element multiset $\mset{\tightN \to \tightN}$.
Given that $\judg[0,1]{\vartwo : \mtyp}\vartwo\mtyp$ holds, 
and that $\varthree : \tightN; \vartwo : \mtyp$ is identical to 
$\varthree : \tightN; \vartwo : \mtyp; \var : \none$, we can easily build the derivation
$\judg[1,1]{\varthree : \tightN; \vartwo : \mtyp; \var : \none}{\vartwo \, \varthree}\tightN$.

Moreover, $\tm \eqdef (\var \, \varthree)\esub\var{\vartwo\esub\vartwo\id}$ reduces in \UOCBV in one $\lsvsym$-step to 
$\tm' \eqdef (\vartwo \, \varthree)\esub\var\vartwo\esub\vartwo\id$
(\cf~\autoref{ex:reduction-for-tight}).
Let us build a typing derivation $\derivtwo''$ for $\tm'$ under the same 
typing environment and type as $\derivtwo'$. 
Moreover, the counter in $\derivtwo''$ should be $(1,1)$.

Recall that
$\tightN \mleq \tightN$, $\mtyp \mleq \mtyp$, and
$\emset \mleq \emset$ always hold. 
Omitting the rule names to improve readability, we get:
\[
  \inferrule* 
    {\inferrule*
      {\inferrule*
        {\inferrule* { }{\judg[0,1]{\vartwo : \mtyp}\vartwo\mtyp}
          \\
          \inferrule* { }{\judg[0,0]{\varthree : \tightN}\varthree\tightN} 
        }
        {\judg[1,1]{\varthree : \tightN; \vartwo : \mtyp; \var : \none}{\vartwo \, \varthree}\tightN}
        \\
        \inferrule* { }{\judg[0,0]{\vartwo : \emset}\vartwo\emset}}
      {\judg[1,1]{\varthree : \tightN; \vartwo : \mtyp}{(\vartwo \, \varthree)\esub\var\vartwo}\tightN}
      \\
      \inferrule*
        {\inferrule* { }{\judg[0,0]{\var_\id:\tightN}{\var_\id}\tightN}}
        {\judg[0,0]\emptyctx\id\mtyp}
    }
    {\judg[1,1]{\varthree : \tightN}{(\vartwo \, \varthree)\esub\var\vartwo\esub\vartwo\id}\tightN}
\]
\end{exa}
\medskip

We now show~\nameref{prop:subject_reduction},
stating that each reduction step in \UOCBV preserves typing.
Moreover, this property not only gives
qualitative information about reduction in \UOCBV
but also quantitative, as reduction correctly decrements the counters in the
corresponding typing derivations.

\begin{prop}[Exact Subject Reduction]
\label{prop:subject_reduction}
Let $\tm \tov\rulename\aset\sset\appflag \tm'$ where 
$\rulename \in \set{\ruledb,\rulelsv}$ and
$\derivable\deriv{\judg[\cm,\ce]\tctx\tm\mtyp}$, with 
$\isAppr\aset\tctx$.
Suppose moreover that if $\appflag = \app$, then either 
$\mtyp = \tightN$ or $\mtyp$ is a singleton, \ie of the form 
$\mset\typ$.
Then, there exists a derivation $\deriv'$ such that 
$\derivable{\deriv'}{\judg[\cm',\ce']\tctx{\tm'}\mtyp}$, where 
if $\rulename = \ruledb$, then $\cm > 0$ and $(\cm',\ce') = (\cm-1,\ce)$, 
and if $\rulename = \rulelsv$, then $\ce > 0$ and 
$(\cm',\ce') = (\cm,\ce-1)$.
\end{prop}

\begin{proof}
By induction on $\tm \tov\rulename\aset\sset\appflag \tm'$.
We only show the case $\tm \tov\rulelsv\aset\sset\appflag \tm'$. 
The reduction step is then of the form:
\begin{mathpar}
  \inferrule
    {\tmtwo \tov{\rulesub\var\val}{\aset \cup \set\var}\sset\appflag \tmtwo' 
      \\ \var \notin \aset \cup \sset
      \\ \val\sctx \in \HAbs\aset}
    {\tm = \tmtwo\esub\var{\val\sctx} \tov\rulelsv\aset\sset\appflag \tmtwo'\esub\var\val\sctx = \tm'}
\end{mathpar}

We thus proceed by induction on $\sctx$, showing only the base case $\sctx = \ctxhole$, since the inductive case is similar to that in~\autoref{prop:vscne_wsr_U}.

Then the derivation $\deriv$ is of the form:
\begin{mathpar}
  \inferrule* [right = \ruleTypES]
    {\derivable{\deriv_\tmtwo}{\judg[\cm_1,\ce_1]{\tctx_1; \var : \optmtyptwo}\tmtwo\mtyp}
      \\ \optmtyptwo \mleq \mtyptwo
      \\ \derivable{\deriv_\val}{\judg[\cm_2,\ce_2]{\tctx_2}\val\mtyptwo}}
    {\judg[\cm_1 + \cm_2,\ce_1 + \ce_2]{\tctx_1 + \tctx_2}{\tmtwo\esub\var\val}\mtyp}
\end{mathpar}
where $\tctx = \tctx_1 + \tctx_2$, $\cm = \cm_1 + \cm_2$, 
and $\ce = \ce_1 + \ce_2$. 
We then have $\isAppr{\aset}{\tctx_1}$ and $\isAppr{\aset}{\tctx_2}$.
Since $\var \notin \fv{\tmtwo\esub\var\val}$, 
then $\var \notin \dom\tctx$ by~\autoref{lem:relevance_systemU}.
Therefore, $\isAppr{\aset \cup \set\var}\tctx$ holds.
This implies, in particular, $\isAppr{\aset \cup \set\var}{\tctx_1}$ 
since $\dom{\tctx_1} \subseteq \dom\tctx$.
Therefore, $\isAppr{\aset \cup \set\var}{\tctx_1; \var : \optmtyptwo}$.
Since $\val \in \HAbs\aset$ and
$\isAppr{\aset}{\tctx_2}$, ~\autoref{lem:systemU_types_of_hereditary_abstractions} applied to $\deriv_{\val}$ yields $\mtyptwo \neq \tightN$.
From $\optmtyptwo \mleq \mtyptwo$, we then have $\optmtyptwo \neq \tightN$.

Let $\optmtyptwo = \mset\typtwo + \optmtyptwo_{\tmtwo'}$
be the splitting obtained by~\autoref{lem:linear_substitution}
on $\tmtwo \tov{\rulesub\var\val}{\aset \cup \set\var}\sset\appflag \tmtwo'$.
Since $\mset\typtwo + \optmtyptwo_{\tmtwo'} = \optmtyptwo$ and 
$\optmtyptwo \mleq \mtyptwo$, then there exist types 
$\mset\typtwo$ and $\mtyptwo_{\tmtwo'}$ such that 
$\mset\typtwo + \mtyptwo_{\tmtwo'} = \mtyptwo$ and
$\mset\typtwo \mleq \mset\typtwo$ and 
$\optmtyptwo_{\tmtwo'} \mleq \mtyptwo_{\tmtwo'}$ by~\autoref{lem:vscne_type_splitting}.
We can then apply~\autoref{lem:systemU_split-merge-value_counters}, 
yielding typing derivations $\deriv^1_\val$ and $\deriv^2_\val$, 
typing environments $\tctx^1_\val$ and $\tctx^2_\val$, and counters 
$\cm^1_\val$, $\cm^2_\val$, $\ce^1_\val$, and $\ce^2_\val$ such that 
$\derivable
  {\deriv^1_\val}
  {\judg[\cm^1_\val,\ce^1_\val]{\tctx^1_\val}\val{\mset\typtwo}}$
and
$\derivable
  {\deriv^2_\val}
  {\judg[\cm^2_\val,\ce^2_\val]{\tctx^2_\val}\val{\mtyptwo_{\tmtwo'}}}$,
with $\tctx_2 = \tctx^1_\val + \tctx^2_\val$ and
$\cm_2 = \cm^1_\val + \cm^2_\val$ and 
$\ce_2 = \ce^1_\val + \ce^2_\val$.

Moreover, by~\autoref{lem:linear_substitution} considering
$\deriv^1_\val$, we have $\ce_1 > 0$ and 
\[
\derivable
  {\deriv_{\tmtwo'}}
  {\judg
    [\cm_1 + \cm^1_\val, \ce_1 + \ce^1_\val - 1]
    {(\tctx_1 + \tctx^1_\val); \var : \optmtyptwo_{\tmtwo'}}
    {\tmtwo'}
    \mtyp}
\]

We then apply rule \ruleTypES with $\deriv_{\tmtwo'}$ and 
$\deriv^2_\val$ as premises, since 
$\optmtyptwo_{\tmtwo'} \mleq \mtyptwo_{\tmtwo'}$, yielding 
\[
\derivable\derivtwo
  {\judg
    [\cm_1 + \cm^1_\val + \cm^2_\val, \ce_1 + \ce^1_\val - 1 + \ce^2_\val]
    {\tctx_1 + \tctx^1_\val + \tctx^2_\val}
    {\tmtwo'\esub\var\val}
    \mtyp}
\]

In this case, $\rulename = \lsvsym$, and $e > 0$ since $e_1 > 0$ 
and $e = e_1 + e_2$.
\end{proof}

We can finally state the soundness of system $\typesystem$ with respect to \UOCBV evaluation. 
Moreover, we obtain an exact counting  for tight typing derivations, as the counters precisely capture the numbers of $\dbsym$- and $\lsvsym$-steps.

\begin{thm}[Exact Soundness of System $\typesystem$]
\label{thm:quantitative_soundness_typing}
Let $\inv{\aset}{\sset}{\tm}$ and $\isAppr{\aset}{\tctx}$. 
Let $\derivable\deriv{\judg[\cm,\ce]\tctx\tm{\mtyp}}$ be a derivation.
Then, there exists a $\tovTopb{\aset}{\sset}$-irreducible term $\tmtwo$ such 
that $\tm \tovTopb{\aset}{\sset}^{k} \tmtwo$, for some $k \in \nat$.
Moreover, if $\deriv$ is tight, then $k = \cm + \ce$ where $\cm$ and $\ce$ are 
respectively the number of $\dbsym$- and $\lsvsym$-steps.
\end{thm}

\begin{proof}
Let $\derivable\deriv{\judg[\cm,\ce]\tctx\tm\mtyp}$ be a derivation.
Since we are in \emph{primary} \UOCBV evaluation, the step kind satisfies $\rulename \in \set{\dbsym,\lsvsym}$.
Moreover, $\isAppr{\aset}\tctx$ holds by hypothesis and $\tctx$ is a tight typing environment.

We proceed by induction on $\cm + \ce$, analysing whether $\tm$ is 
irreducible in \UOCBV.
\begin{itemize}
\item 
  If $\tm$ is irreducible, then 
  $\tm \in \NF{\aset}{\sset}\nonapp$ by~\autoref{thm:uocbv_characterisation_of_normal_forms}. 
  Then, we let $\tmtwo \eqdef \tm$ since $\tm$ primary reduces to itself in $0$ steps. 

  Suppose now that $\deriv$ is tight.
  By~\autoref{prop:systemU_ucbv_nf_counters_zero} in the appendix, we obtain 
  $(\cm,\ce) = (0,0)$, and then we are done. 
\item
  If $\tm$ is not irreducible in \UOCBV evaluation, this means that 
  there exists a term $\tm'$ and a step kind 
  $\rulename \in \set{\dbsym,\lsvsym}$ such that 
  $\tm \tov\rulename\aset\sset\nonapp \tm'$.
  By~\autoref{prop:subject_reduction}, we have 
  $\derivable{\deriv'}{\judg[\cmtwo,\cetwo]\tctx{\tm'}\mtyp}$, 
  with $\cm > 0$ and $(\cmtwo,\cetwo) = (\cm - 1, \ce)$ if 
  $\rulename = \dbsym$, and with $\ce > 0$ and 
  $(\cmtwo,\cetwo) = (\cm, \ce - 1)$ if $\rulename = \lsvsym$.

  Since $\cm + \ce > \cmtwo + \cetwo$,
  and $\inv{\aset}{\sset}{\tm'}$ as $\fv{\tm} = \fv{\tm'}$, then we apply the \ih, 
  and we obtain a primary \UOCBV reduction sequence under parameters $\aset$ and $\sset$ from $\tm'$
  to a normal form in $k'$ steps for some $k' \in \nat$. 
  By joining this reduction sequence to normal form to the reduction
  step $\tm \tov\rulename\aset\sset\nonapp \tm'$, we obtain that
  $\tm$ reduces to a \UOCBV-normal form in $k$ steps, where $k = k'+1$.

  Suppose now that $\deriv$ is tight. 
  Since~\autoref{prop:subject_reduction} preserves both the 
  typing environment $\tctx$ and the type $\mtyp$, 
  the derivation $\deriv'$ is tight as well. 
  Hence, the \ih yields $k'=\cmtwo + \cetwo$, where $\cmtwo$ and $\cetwo$ are respectively
  the number of $\dbsym$- and $\lsvsym$-steps in the reduction
  sequence from $\tm'$ to normal form.
  We separate into cases,
  depending on the form the step kind $\rulename \in
  \set{\dbsym,\lsvsym}$ has:
  \begin{itemize}
  \item $\rulename = \dbsym$.
    Then, $\cm = \cm + (1 - 1) = \cmtwo + 1$ and
    $\ce = \cetwo$, so $\tm$ terminates in \UOCBV evaluation in $\cm$
    $\dbsym$-steps and $\ce$ $\lsvsym$-steps.
  \item $\rulename = \lsvsym$.
    Then, $\cm = \cmtwo$ and $\ce = \ce + (1 - 1) = \cetwo + 1$, 
    so $\tm$ terminates in \UOCBV evaluation in $\cm$ $\dbsym$-steps 
    and $\ce$ $\lsvsym$-steps.
    \qedhere
  \end{itemize}
\end{itemize}
\end{proof}

To illustrate this property, consider the tight derivation $\derivtwo'$ in~\autoref{ex:derivation-for-tight}
for our working~\autoref{exa:working-example}, where the counter in the conclusion of
$\derivtwo'$ is $(1,2)$. A reduction sequence from $\tm$ terminates in
$\var_1\esub{\var_1}\varthree\esub\var\vartwo\esub\vartwo\id \in
\NF\emptyset{\set\varthree}\nonapp$ in exactly $1$ $\dbsym$-step
and $2$ $\lsvsym$-steps (see~\autoref{ex:reduction-for-tight}). 
Note how these numbers match the $(1,2)$ counters.

Tightness in the previous example is crucial. 
Indeed, non-tight judgements may not provide exact information 
about the length of reduction sequences to normal form.
For example, the derivation of the judgement
$\judg[1,1]{\varthree : \tightN, \var : \mtyp}{\var\,\varthree}\tightN$ 
in~\autoref{ex:derivation-for-tight} is not tight, as the type $\mtyp$ is not 
tight, and the counters $(1,1)$ do not correspond to the number of 
\UOCBV-steps needed to reach the normal form of $\var\,\varthree$, 
which is already in \UOCBV-normal form.

\defn{Completeness} of system $\typesystem$ with respect to \UOCBV evaluation 
can be proved using standard techniques, requiring a subject expansion
property (\autoref{prop:subject_expansion}) based in turn on an
anti-linear substitution lemma
(\autoref{lem:systemU_anti-linear_substitution_ucbv}).  These
statements are quite technical, given the indexed formulation
of \UOCBV and the fact that
the substitution mechanism is linear.

The first property guarantees that \UOCBV-normal forms are tightly typable 
in system $\typesystem$ (\autoref{prop:nfs_are_tight_typable}).
For that, we need to construct tight environments for each normal form
$\tm \in \NF\aset\sset\appflag$; this is done by typing
the \emph{reachable} variables in $\aset$ with $\emset$ and those in
$\sset$ with $\tightN$. 
More precisely, given $\tm$, $\aset$, and $\sset$ such that $\inv\aset\sset\tm$, the \defn{tight environment}
for $\tm$ under $\aset$ and $\sset$ is written $\TEnv\aset\sset\tm$ and defined as follows:
\[
  \TEnv\aset\sset\tm(\var) \eqdef
    \begin{cases*}
      \emset  & if $\var \in \aset \cap \rv\tm$
    \\
      \tightN & if $\var \in \sset \cap \rv\tm$
    \\
      \none   & otherwise
    \end{cases*}
\]

Tight environments allow us to \emph{tightly} type \UOCBV-normal forms:
\begin{prop}[\UOCBV-Normal Forms are Tight Typable]
\label{prop:nfs_are_tight_typable}
Let $\tm$ be a term in $\NF\aset\sset\appflag$ and $\inv\aset\sset\tm$.
Then, there exists a tight environment $\TEnv\aset\sset\tm$ and a tight type $\tight$ such that
$\derivable\deriv{\judg[0,0]{\TEnv\aset\sset\tm}\tm\tight}$.
Moreover, if $\tm \in \HAbs\aset$, then $\tight = \emset$, and if 
$\tm \in \Struct\sset$, then $\tight = \tightN$.
\end{prop}
\begin{proof} 
By induction on the derivation of $\tm \in \NF\aset\sset\appflag$.
We only show the case \ruleUNFVar: 
\begin{mathpar}
  \inferrule
    {\var \in \aset \implies \appflag = \nonapp}
    {\tm = \var \in \NF\aset\sset\appflag}
\end{mathpar}

Since $\inv\aset\sset\var$ holds, then $\aset \disj \sset$ and
$\fv\var = \set\var \subseteq \aset \cup \sset$.
We then consider two cases, depending on whether $\var \in \aset$ 
or $\var \in \sset$.
\begin{itemize}
\item 
  If $\var \in \aset$, then $\TEnv\aset\sset\var(\var) = \emset$ 
  since $\var \in \aset \cap \rv\var$ and
  $\TEnv\aset\sset\var(\vartwo) = \none$ for any other variable $\vartwo$. 
  Having $\numarrows\emset = 0$, we take $\tight = \emset$ and 
  by applying the \ruleTypVar axiom we obtain 
  $\derivable\deriv{\judg[0,0]{\var : \emset}\var\emset}$.
  Moreover, $\var \in \HAbs\aset$ by the first rule in~\autoref{fig:HAbs_Struct} (above).
\item
  If $\var \in \sset$, then $\TEnv\aset\sset\var(\var) = \tightN$ 
  since $\var \in \sset \cap \rv\var$ and 
  $\TEnv\aset\sset\var(\vartwo) = \none$ for any other variable $\vartwo$. 
  Having $\numarrows{\tightN} = 0$, we take $\tight = \tightN$ and 
  by applying the \ruleTypVar axiom we obtain 
  $\derivable\deriv{\judg[0,0]{\var : \tightN}\var\tightN}$.
  Moreover, $\var \in \Struct\sset$ by the first rule in~\autoref{fig:HAbs_Struct} (below).
  \qedhere
\end{itemize}
\end{proof}

\begin{lem}[Exact Anti-Linear Substitution]
\label{lem:systemU_anti-linear_substitution_ucbv}
Let $\var$ be a variable, $\tm$ a term, $\aset$ an abstraction frame,
and $\sset$ a structure frame such that 
$\inv{\aset \cup \set\var}\sset\tm$ holds.
Consider the abstraction frames $\aset_0$ and $\asettwo$ such that 
$\aset_0 \subseteq \aset$ and $\asettwo \disj \aset$ hold.

Let $\tm \tov{\rulesub\var\val}{\aset \cup \set\var}\sset\appflag \tm'$,
where $\val \in \HAbs{\aset_0 \cup \asettwo}$, and 
let $\derivable{\deriv_{\tm'}}{\judg[\cm,\ce]{\tctx; \var : \optmtyptwo}{\tm'}\mtyp}$
such that $\isAppr{\aset \cup \asettwo \cup \set\var}{\tctx; \var : \optmtyptwo}$.
Suppose moreover that if $\appflag = \app$, then either 
$\mtyp = \tightN$ or $\mtyp = \mset\typ$.

Then, there exist typing derivations $\deriv_\tm,\deriv_\val$, 
typing environments $\tctx_\tm,\tctx_\val$, an arrow type 
$\typtwo$, and counters $(\cm_\tm,\ce_\tm), (\cm_\val,\ce_\val)$ satisfying:
\begin{enumerate}
\item 
  $\derivable
  	{\deriv_\tm}
    {\judg[\cm_\tm,\ce_\tm + 1]{\tctx_\tm;\var : (\optmtyptwo + \mset\typtwo)}\tm\mtyp}$
\item
  $\derivable
  	{\deriv_\val}
  	{\judg[\cm_\val,\ce_\val]{\tctx_\val}\val{\mset\typtwo}}$
\item
  $\tctx = \tctx_\tm + \tctx_\val$ and
  $\cm = \cm_\tm + \cm_\val$ and $\ce = \ce_\tm + \ce_\val$
\end{enumerate} 
\end{lem}
\begin{proof}
By induction on $\tm \tov{\rulesub\var\val}{\aset \cup \set\var}\sset\appflag \tm'$.
We only show cases \ruleUSub, \ruleUAppL, and \ruleUEsR.

\begin{itemize}
\item \ruleUSub.
  Then, 
  $\tm = \var \tov{\rulesub\var\val}{\aset \cup \set\var}\sset\app \val = \tm'$,
  where $\appflag = \app$.
  Recall that $\var \notin \fv\val$ by the grammar of step kinds. 
  We analyse by cases the form of $\val$:  
  \begin{itemize}
  \item 
    If $\val = \vartwo$, then $\var \neq \vartwo$.
    The derivation $\deriv_{\tm'}$ can only be obtained by the \ruleTypVar axiom, 
    so it is of the form
    $\derivable{\deriv_{\tm'}}{\judg[0,\numarrows\mtyp]{\vartwo : \mtyp}\vartwo\mtyp}$,
    where $\tctx = \vartwo : \mtyp$, $\cm = 0$, 
    $\ce = \numarrows\mtyp$, and $\optmtyptwo = \none$.
    Moreover, $\mtyp = \tightN$ or $\mtyp = \mset\typ$ by hypothesis, 
    since $\appflag = \app$. 
    
    Furthermore, since $\vartwo \in \HAbs{\aset_0 \cup \asettwo}$ holds,
    then it must be the case that $\vartwo \in \aset_0 \cup \asettwo$.
    This leads to $\vartwo \in \aset \cup \asettwo$, as $\aset_0 \subseteq \aset$.
    Then, $(\vartwo : \mtyp)(\vartwo) \neq \tightN$ since
    $\isAppr{\aset \cup \asettwo \cup \set\var}{\vartwo : \mtyp}$.
    Therefore, $\mtyp = \mset\typ$ and thus 
    $\ce = \numarrows{\mset\typ} = 1$ necessarily holds.

    We take $\deriv_\val \eqdef \deriv_{\tm'}$, 
    $\tctx_\tm \eqdef \emptyctx$, $\tctx_\val \eqdef \vartwo : \mset\typ$, 
    $\typtwo \eqdef \typ$, $\cm_\tm \eqdef 0$, $\ce_\tm \eqdef 0$, 
    $\cm_\val \eqdef 0$, and $\ce_\val \eqdef 1$, and obtain 
    $\derivable{\deriv_\tm}{\judg[0,1]{\var : \mset\typ}\var{\mset\typ}}$ 
    by the \ruleTypVar axiom.

    Moreover,
    $\tctx = \vartwo : \mset\typ = \tctx_\tm + \tctx_\val$ 
    and $\cm = 0 = \cm_\tm + \cm_\val$ and $\ce = 1 = \ce_\tm + \ce_\val$.
  \item
    If $\val = \lam\vartwo\tmtwo$ with 
    $\var \notin \fv{\lam\vartwo\tmtwo}$, then
    $\var \notin \dom{\tctx; \var : \optmtyptwo}$ by~\autoref{lem:relevance_systemU}. 
    Therefore, $\optmtyptwo = \none$.
    Moreover, $\mtyp = \tightN$ or $\mtyp = \mset\typ$ by hypothesis, 
    since $\appflag = \app$.
    Since $\deriv_{\tm'}$ can only be obtained from rule \ruleTypAbs,
    it must be the case that $\mtyp = \mset\typ$. 

    We take $\deriv_\val \eqdef \deriv_{\tm'}$, 
    $\tctx_\tm \eqdef \emptyctx$, $\tctx_\val \eqdef \tctx$, 
    $\typtwo \eqdef \typ$, $\cm_\tm \eqdef 0$, $\ce_\tm \eqdef 0$, 
    $\cm_\val \eqdef \cm$, and $\ce_\val \eqdef \ce$, and obtain
    $\derivable{\deriv_\tm}{\judg[0,1]{\var : \mset\typ}\var{\mset\typ}}$ 
    by the \ruleTypVar axiom.
    
    Moreover, 
    $\tctx = \tctx_\tm + \tctx_\val$ and 
    $\cm = \cm_\tm + \cm_\val$ and $\ce = \ce_\tm + \ce_\val$.
  \end{itemize}
\item \ruleUAppL.
  Then, the reduction step is of the form: 
  \begin{mathpar}
    \inferrule
      {\tmtwo \tov{\rulesub\var\val}{\aset \cup \set\var}\sset\app \tmtwo'}
      {\tm = \tmtwo \, \tmthree \tov{\rulesub\var\val}{\aset \cup \set\var}\sset\appflag \tmtwo' \, \tmthree = \tm'}
  \end{mathpar}

  The derivation $\deriv_{\tm'}$ can be obtained either by rule 
  \ruleTypAppP or by rule \ruleTypAppC.
  We consider the former, since the latter is analogous.
  Then, $\deriv_{\tm'}$ is of the form:
  \begin{mathpar}
    \inferrule* [right = \ruleTypAppP]
      {\inferrule*[rightskip = 1em]{}{
        \derivable
          {\deriv_{\tmtwo'}}
          {\judg[\cm_{\tmtwo'},\ce_{\tmtwo'}]{\tctx_{\tmtwo'}; \var : \optmtyptwo_{\tmtwo'}}{\tmtwo'}\tightN}}
      \\ 
        \derivable{\deriv_\tmthree}
          {\judg[\cm_\tmthree,\ce_\tmthree]{\tctx_\tmthree; \var : \optmtyptwo_\tmthree}\tmthree\tight}}
      {\judg
        [\cm_{\tmtwo'} + \cm_\tmthree, \ce_{\tmtwo'} + \ce_\tmthree]
        {(\tctx_{\tmtwo'} + \tctx_\tmthree); \var : (\optmtyptwo_{\tmtwo'} + \optmtyptwo_\tmthree)}
        {\tmtwo' \, \tmthree}
        \tightN}
  \end{mathpar}
  where $\tctx = \tctx_{\tmtwo'} + \tctx_\tmthree$, 
  $\optmtyptwo = \optmtyptwo_{\tmtwo'} + \optmtyptwo_\tmthree$,
  $\cm = \cm_{\tmtwo'} + \cm_\tmthree$, 
  $\ce = \ce_{\tmtwo'} + \ce_\tmthree$, and 
  $\mtyp = \tightN$.

  Note that 
  $\isAppr{\aset \cup \asettwo \cup \set\var}{\tctx; \var : \optmtyptwo}$
  implies
  $\isAppr{\aset \cup \asettwo \cup \set\var}{\tctx_{\tmtwo'}; \var : \optmtyptwo_{\tmtwo'}}$, 
  as $\tctx = \tctx_{\tmtwo'} + \tctx_\tmthree$ and 
  $\optmtyptwo = \optmtyptwo_{\tmtwo'} + \optmtyptwo_\tmthree$.
  
  On the other hand, it holds that $\appflag = \app$ in the reduction step
  $\tmtwo \tov{\rulesub\var\val}{\aset \cup \set\var}\sset\app \tmtwo'$
  and that the type in the conclusion of $\deriv_{\tmtwo'}$ is $\tightN$.
  Moreover, $\inv{\aset \cup \set\var}\sset{\tmtwo \, \tmthree}$
  implies in particular $\inv{\aset \cup \set\var}\sset\tmtwo$.
  We can then apply the \ih on 
  $\tmtwo \tov{\rulesub\var\val}{\aset \cup \set\var}\sset\app \tmtwo'$, 
  yielding typing derivations $\deriv_\tmtwo,\deriv_\val$,
  typing environments $\tctx_\tmtwo,\tctx_\val$, an arrow type 
  $\typtwo$, and counters $(\cm_\tmtwo,\ce_\tmtwo),(\cm_\val,\ce_\val)$ 
  satisfying:
  \begin{enumerate}[label=(\arabic*')]
  \item
    $\derivable
      {\deriv_\tmtwo}
      {\judg[\cm_\tmtwo,\ce_\tmtwo + 1]{\tctx_\tmtwo; \var : (\optmtyptwo_{\tmtwo'} + \mset\typtwo)}\tmtwo\tightN}$
  \item
    $\derivable{\deriv_\val}{\judg[\cm_\val,\ce_\val]{\tctx_\val}\val{\mset\typtwo}}$
  \item
    $\tctx_{\tmtwo'} = \tctx_\tmtwo + \tctx_\val$ and 
    $\cm_{\tmtwo'} = \cm_\tmtwo + \cm_\val$ and 
    $\ce_{\tmtwo'} = \ce_\tmtwo + \ce_\val$
  \end{enumerate}
  Taking $\deriv_\val$, 
  $\tctx_\tm \eqdef \tctx_\tmtwo + \tctx_\tmthree$, $\tctx_\val$, 
  $\typtwo$, $\cm_\tm \eqdef \cm_\tmtwo + \cm_\tmthree$, 
  $\ce_\tm \eqdef \ce_\tmtwo + \ce_\tmthree$, $\cm_\val$, and $\ce_\val$,
  we obtain the derivation
  $\derivable
    {\deriv_\tm}
    {\judg
      [\cm_\tm,\ce_\tm + 1]
      {\tctx_\tm; \var : (\optmtyptwo + \mset\typtwo)}
      {\tmtwo \, \tmthree}
      \tightN}$,
  by applying rule \ruleTypAppP with $\deriv_\tmtwo$ and 
  $\deriv_\tmthree$ as premises.
  
  Moreover,
  \[
    \tctx = \tctx_{\tmtwo'} + \tctx_\tmthree 
  =_\ih \tctx_\tmtwo + \tctx_\val + \tctx_\tmthree 
  = \tctx_\tm + \tctx_\val \text{ and}
  \] 
  \[ 
    \cm = \cm_{\tmtwo'} + \cm_\tmthree 
  =_\ih \cm_\tmtwo + \cm_\val + \cm_\tmthree 
  = \cm_\tm + \cm_\val \text{ and}
  \]
  \[
    \ce = \ce_{\tmtwo'} + \ce_\tmthree 
  =_\ih \ce_\tmtwo + \ce_\val + \ce_\tmthree 
  = \ce_\tm + \ce_\val
  \]
\item \ruleUEsR.
  Then, the reduction step is of the form:
  \begin{mathpar}
    \inferrule
      {\tmthree \tov{\rulesub\var\val}{\aset \cup \set\var}\sset\nonapp \tmthree'}
      {\tm = \tmtwo\esub\vartwo\tmthree 
        \tov{\rulesub\var\val}{\aset \cup \set\var}\sset\appflag 
        \tmtwo\esub\vartwo{\tmthree'} = \tm'}
  \end{mathpar}

  The derivation $\deriv_{\tm'}$ can be obtained only by rule \ruleTypES, 
  so it is of the form:
  \begin{mathpar}
    \inferrule
      {\inferrule*
        {}
        {\derivable
          {\deriv_\tmtwo}
          {\judg[\cm_\tmtwo,\ce_\tmtwo]
            {\tctx_\tmtwo; \var : \optmtyptwo_\tmtwo; \vartwo : \optmtypthree}\tmtwo\mtyp}}
        \\ \optmtypthree \mleq \mtypthree
        \\ 
          \derivable
            {\deriv_{\tmthree'}}
            {\judg[\cm_{\tmthree'},\ce_{\tmthree'}]
              {\tctx_{\tmthree'}; \var : \optmtyptwo_{\tmthree'}}{\tmthree'}\mtypthree}}
      {\judg
        [\cm_\tmtwo + \cm_{\tmthree'}, \ce_\tmtwo + \ce_{\tmthree'}]
        {(\tctx_\tmtwo + \tctx_{\tmthree'}); \var : (\optmtyptwo_\tmtwo + \optmtyptwo_{\tmthree'})}
        {\tmtwo\esub\vartwo{\tmthree'}}
        \mtyp}
  \end{mathpar}
  where $\tctx = \tctx_\tmtwo + \tctx_{\tmthree'}$, 
  $\optmtyptwo = \optmtyptwo_\tmtwo + \optmtyptwo_{\tmthree'}$, 
  $\cm = \cm_\tmtwo + \cm_{\tmthree'}$, and $\ce = \ce_\tmtwo + \ce_{\tmthree'}$.

  Note that $\inv{\aset \cup \set\var}\sset{\tmtwo\esub\vartwo\tmthree}$
  implies in particular $\inv{\aset \cup \set\var}\sset\tmthree$.
  On the other hand,
  $\isAppr{\aset \cup \asettwo \cup \set\var}{\tctx; \var : \optmtyptwo}$
  implies in particular 
  $\isAppr{\aset \cup \asettwo \cup \set\var}{\tctx_{\tmthree'}; \var : \optmtyptwo_{\tmthree'}}$, 
  as $\tctx = \tctx_\tmtwo + \tctx_{\tmthree'}$ and 
  $\optmtyptwo = \optmtyptwo_\tmtwo + \optmtyptwo_{\tmthree'}$.
  Moreover, it holds that $\appflag = \nonapp$ in the reduction step
  $\tmthree \tov{\rulesub\var\val}{\aset \cup \set\var}\sset\nonapp \tmthree'$.
  We can then apply the \ih on this reduction step, yielding typing 
  derivations $\deriv_\tmthree,\deriv_\val$, typing environments 
  $\tctx_\tmthree,\tctx_\val$, an arrow type $\typtwo$, and 
  counters $(\cm_\tmthree,\ce_\tmthree),(\cm_\val,\ce_\val)$
  such that:
  \begin{enumerate}[label = (\arabic*')]
  \item
    $\derivable{\deriv_\tmthree}
      {\judg
        [\cm_\tmthree,\ce_\tmthree + 1]
        {\tctx_\tmthree; \var : (\optmtyptwo_{\tmthree'} + \mset\typtwo)}
        \tmthree
        \mtypthree}$
  \item
    $\derivable{\deriv_\val}{\judg[\cm_\val,\ce_\val]{\tctx_\val}\val{\mset\typtwo}}$
  \item
    $\tctx_{\tmthree'} = \tctx_\tmthree + \tctx_\val$
    and $\cm_{\tmthree'} = \cm_\tmthree + \cm_\val$
    and $e_{\tmthree'} = \ce_\tmthree + \ce_\val$
  \end{enumerate}

  Taking $\deriv_\val$, 
  $\tctx_\tm \eqdef \tctx_\tmtwo + \tctx_\tmthree$, 
  $\tctx_\val$, $\typtwo$, $\cm_\tm \eqdef \cm_\tmtwo + \cm_\tmthree$, 
  $\ce_\tm \eqdef \ce_\tmtwo + \ce_\tmthree$, $\cm_\val$, and $\ce_\val$,
  we obtain the derivation
  $\derivable
    {\deriv_\tm}
    {\judg[\cm_\tm,\ce_\tm + 1]{\tctx_\tm; \var : (\optmtyptwo + \mset\typtwo)}
      {\tmtwo\esub\vartwo\tmthree}
      \mtyp}$ by applying rule \ruleTypES with $\deriv_\tmtwo$ 
  and $\deriv_\tmthree$ as premises. 

  Moreover, 
  \[
    \tctx = \tctx_\tmtwo + \tctx_{\tmthree'}
  =_\ih \tctx_\tmtwo + \tctx_\tmthree + \tctx_\val
  = \tctx_\tm + \tctx_\val \text{ and}
  \]
  \[
    \cm = \cm_\tmtwo + \cm_{\tmthree'}
  =_\ih \cm_\tmtwo + \cm_\tmthree + \cm_\val
  = \cm_\tm + \cm_\val \text{ and}
  \]
  \[
    \ce = \ce_\tmtwo + \ce_{\tmthree'}
  =_\ih \ce_\tmtwo + \ce_\tmthree + \ce_\val
  = \ce_\tm + \ce_\val
  \qedhere
  \] 
\end{itemize}
\end{proof}

Exact subject expansion is analogous to exact subject reduction
but applies in the \emph{reverse direction}: given a typable term 
$\tm'$ and a reduction step $\tm \tov\rulename\aset\sset\appflag \tm'$, 
the term $\tm$ is typable as well,
with the same type and typing environment.
This proposition also provides quantitative information, 
as the counters in the typing derivation of $\tm$ are 
correctly adjusted with respect to the ones in the derivation of $\tm'$.
\begin{prop}[Exact Subject Expansion]
\label{prop:subject_expansion}
Let $\inv\aset\sset\tm$ hold, and suppose
$\tm \tov\rulename\aset\sset\appflag \tm'$
where $\rulename \in \set{\ruledb,\rulelsv}$ and 
$\derivable{\deriv'}{\judg[\cm',\ce']\tctx{\tm'}\mtyp}$, 
with $\isAppr\aset\tctx$.
Furthermore, assume that if $\appflag = \app$, then either 
$\mtyp = \tightN$ or $\mtyp$ is a singleton, \ie of the form 
$\mset\typ$.
Then, there exists a derivation $\deriv$ such that 
$\derivable\deriv{\judg[\cm,\ce]\tctx\tm\mtyp}$, where 
if $\rulename = \ruledb$, then $(\cm,\ce) = (\cm'+1,\ce')$, and
if $\rulename = \rulelsv$, then $(\cm,\ce) = (\cm',\ce'+1)$.
\end{prop}
\begin{proof}
By induction on $\tm \tov\rulename\aset\sset\appflag \tm'$, where in particular, \autoref{lem:systemU_anti-linear_substitution_ucbv} is used in the case
$\tm = \tmtwo\esub\var{\val\sctx} \tov\lsvsym\aset\sset\appflag \tmtwo'\esub\var\val\sctx = \tm'$.
\end{proof}

We now conclude with the quantitative completeness of system
$\typesystem$ with respect to \UOCBV evaluation. 
That is, given a terminating \UOCBV reduction sequence from $\tm$ 
with exactly $\cm$ $\dbsym$-steps and $\ce$ $\lsvsym$-steps, there 
exists a tight derivation of $\tm$ with counters $(\cm,\ce)$.

\begin{thm}[Exact Completeness of System $\typesystem$ with Respect to \UOCBV]
\label{thm:quantitative_completeness_typing} 
Let $\inv\aset\sset\tm$ and consider a reduction sequence 
$\tm\ \tovTopb{\aset}{\sset}^n\ \tmtwo$, where $\tmtwo$ is 
$\tovTopb{\aset}{\sset}$-irreducible. 
Let $n = \cm + \ce$, where $\cm$ and $\ce$ are respectively the 
numbers of $\dbsym$- and $\lsvsym$-steps in the sequence. 
Then, there exists a tight environment $\tctx$ and a tight type 
$\tight$ such that $\isAppr\aset\tctx$ and 
$\derivable{}{\judg[\cm,\ce]\tctx\tm\tight}$.
\end{thm}

\begin{proof}
We proceed by induction on $n$:
\begin{enumerate}
\item 
  Base case: $n = 0$.
  Then, we have that $\tm \in \NF{\aset}{\sset}{\nonapp}$,
  and $\cm = \ce = 0$, as $0 = \cm + \ce$.
  By~\autoref{prop:nfs_are_tight_typable}, there exists a 
  tight type derivation $\deriv$ and a tight type $\tight$ such that 
  $\derivable\deriv{\judg[0,0]{\TEnv{\aset}{\sset}\tm}\tm\tight}$
  holds, with $\TEnv{\aset}{\sset}\tm$ a tight environment.
  Thus, we are done.
\item 
  Inductive case: $n > 0$, assuming the statement holds for $n - 1$. 
  Then, in particular, $\tm \notin \NF{\aset}{\sset}{\nonapp}$.
  The reduction sequence is of the form:
  \[
    \tm \tov{\rulename_1}{\aset}{\sset}{\nonapp} \tm' 
        \tov{\rulename_2}{\aset}{\sset}{\nonapp} \hdots 
        \tov{\rulename_n}{\aset}{\sset}{\nonapp} \tmtwo
  \] 
  where $\rulename_{i} \in \set{\dbsym,\lsvsym}$ 
  for all $i \in \set{1..n}$, and $n - 1 = \cmtwo + \cetwo$, $\cmtwo$ and 
  $\cetwo$ correspond to the number of $\dbsym$- and $\lsvsym$-steps
  in the reduction sequence from $\tm'$ to $\tmtwo$.

  By the \ih, there is a tight typing derivation $\deriv'$ such that 
  $\isAppr\aset\tctx$ and
  $\derivable{\deriv'}{\judg[\cmtwo,\cetwo]\tctx{\tm'}\tight}$.
  We can then apply \autoref{prop:subject_expansion}, yielding a typing derivation $\deriv$
  such that $\derivable\deriv{\judg[\cm,\ce]\tctx\tm\tight}$, where
  if $\rulename_1 = \dbsym$, then $(\cm,\ce) = (\cmtwo + 1, \cetwo)$ and
  if $\rulename_1 = \lsvsym$ then $(\cm,\ce) = (\cmtwo, \cetwo + 1)$. 
  Since $\deriv$ is also tight because~\autoref{prop:subject_expansion} preserves the type environment and the type, then we are done.
  \qedhere
\end{enumerate}
\end{proof}

To illustrate this result, take the reduction sequence starting from
$\tm = (\var \, \varthree)\esub\var{\vartwo\esub\vartwo\id}$ and 
ending in
$\var_1\esub{\var_1}\varthree\esub\var\vartwo\esub\vartwo\id \in 
\NF\emptyset{\set\varthree}\nonapp$ (see \autoref{ex:reduction-for-tight}).
The length of this sequence is $3 = \cm + \ce$, with 
$\cm = 1$ $\dbsym$-step and $\ce = 2$ $\lsvsym$-steps.
\autoref{ex:derivation-for-tight} shows a tight environment 
$\varthree : \tightN$ and a tight type $\tightN$ such that 
$\derivable{\derivtwo'}{\judg[1,2]{\varthree : \tightN}\tm\tightN}$.

Theorems~\ref{thm:quantitative_soundness_typing} and~\ref{thm:quantitative_completeness_typing} state that a term 
$\tm$ is $\typesystem$-typable if and only if $\tm$ is 
\UOCBV-\emph{weakly terminating}.
The characterisation of \UOCBV-termination through $\typesystem$-typability follows 
from these two theorems and the fact that all reduction sequences to normal form have the same length
in \UOCBV evaluation, following~\autoref{cor:consequences-dp}.

We can then first state the following qualitative result:
\begin{cor}[Qualitative Characterisation of \UOCBV-Termination through $\typesystem$-Typability]
\label{cor:qualitative_characterisation_of_termination}
A term $\tm$ is $\typesystem$-typable if and only if $\tm$ terminates in \UOCBV. 
\end{cor}

Moreover, this characterisation captures quantitative information on the number of \UOCBV evaluation steps \emph{exactly}:
\begin{cor}[Exact Characterisation of \UOCBV-Termination through Tight $\typesystem$-Typability]
\label{cor:quantitative_characterisation_of_termination}
A term $\tm$ is tightly $\typesystem$-typable with counters $(\cm,\ce)$ if and only if $\tm$ terminates in \UOCBV after $\cm$
function application $\dbsym$-steps and $\ce$ substitution $\lsvsym$-steps. 
\end{cor}

\section{Observational Equivalence Between the Calculi}
\label{sec:obs_eq}
In this section, we assemble all the previous results to conclude with the main contribution in this article.

\begin{cor}
\label{coro:main_result}
Let $\tm$ be a term.
The following statements are equivalent:
\begin{enumerate}
\item \label{it:v_typable}
 $\tm$ is $\systemV$-typable.
\item \label{it:u_typable}
  $\tm$ is $\typesystem$-typable.
\item \label{it:vsc_terminating}
  $\tm$ is \VSC-terminating.
\item \label{it:vscne_terminating}
  $\tm$ is \VSCne-terminating.
\item \label{it:uocbv_terminating}
  $\tm$ is \UOCBV-terminating.
\end{enumerate}
\end{cor}

\begin{proof}
The double implication (\ref{it:u_typable} $\iff$ \ref{it:vscne_terminating}) is~\autoref{cor:u-characterises-vscne}.
The double implication (\ref{it:u_typable} $\iff$ \ref{it:uocbv_terminating}) is~\autoref{cor:qualitative_characterisation_of_termination}; 
thus the type system $\typesystem$ connects \VSCne and \UOCBV.
Therefore (\ref{it:u_typable} $\iff$ \ref{it:vscne_terminating} $\iff$ \ref{it:uocbv_terminating}).

On the other hand, (\ref{it:vsc_terminating} $\iff$ \ref{it:vscne_terminating}) is~\autoref{cor:v-equiv-vscne}, 
thus the type system $\systemV$ connects both \VSC and \VSCne (\autoref{t:charact-V-VSC} and~\autoref{cor:v-characterises-vscne}, respectively).

Last, the double implication (\ref{it:v_typable} $\iff$ \ref{it:vsc_terminating}) is~\autoref{t:charact-V-VSC}.

Thus we conclude that all items are equivalent.
\end{proof}

An interesting consequence of this corollary is that the type systems $\systemV$ and $\typesystem$ type the same terms:
\begin{cor}[$\systemV$-Typability Characterises \UOCBV-Termination]
  A term $\tm$ is $\systemV$-typable if and only if $\tm$
  is \UOCBV-terminating.
\end{cor}

We can conclude this section by emphasising an important outcome.
Given a \CBV reduction system $\relation$,
we define $\tm$ to be \defn{observationally equivalent}
to $\tmtwo$, written $\tm \cong_{\relation} \tmtwo$, if and only if
for every weak context $\wctx$ we have that
$\ctxof{\wctx}{\tm}$ $\relation$-terminates if and only if
$\ctxof{\wctx}{\tmtwo} $ $\relation$-terminates.
Then two \CBV reduction systems $\relation_1$ and $\relation_2$ are 
\defn{observationally equivalent}, written $\relation_1 \cong \relation_2$, 
if and only if for every $\tm, \tmtwo$, we have
$\tm \cong_{\relation_1} \tmtwo \iff \tm \cong_{\relation_2} \tmtwo $.

\autoref{coro:main_result} above establishes 
that the (potentially) erasing \VSC, and the non-erasing calculi
\VSCne and \UOCBV are pairwise
observationally equivalent.

\begin{cor} \label{cor:obseq} \VSC $\cong$ \VSCne $\cong$ \UOCBV. \end{cor}

Thus, we can confidently say that \UOCBV optimises \VSC at the \emph{operational} level, while preserving its original semantics.

\section{Conclusions}
This paper presents the first semantic interpretation of useful \CBV in the
literature. 
It is based on a non-idempotent intersection type system called $\typesystem$. 
Despite the considerable complexity of different available specifications of the
operational semantics of useful \CBV, system $\typesystem$ remains
remarkably simple. This highlights the ability of semantic approaches
based on non-idempotent intersection types to capture intricate
operational details with streamlined intuitive methods.

System $\typesystem$ characterises \UOCBV-termination as well
as \VSCne-termination. In turn, system $\systemV$ characterises
both \VSCne-termination and $\VSC$-termination, providing a second
correspondence between non-erasing and open \CBV. Hence, \VSCne acts
as an intermediate calculus connecting $\VSC$, the open \CBV calculus
from which our development starts, with its useful optimisation
\UOCBV. Combining these two correspondences, we conclude that
\VSC, \VSCne, and \UOCBV induce the same notion of termination and
are therefore observationally equivalent. In particular, useful open \CBV 
can be understood as an optimisation of open \CBV that preserves
its observational semantics.

Moreover, system $\typesystem$ provides exact measures for the 
lengths of \UOCBV-reduction sequences to normal form; that is why
system $\typesystem$ is a \emph{quantitative} interpretation of \UOCBV. 
More precisely, a term $\tm$ is \emph{tightly} typable 
with counters $\cm$ and $\ce$ in system $\typesystem$ if and only if 
$\tm$ \UOCBV-terminates in exactly $\cm$ function 
application steps and $\ce$ substitution steps.

Our contribution is novel, as 
previous definitions of useful evaluation lack semantic interpretations, and current
quantitative interpretations of \CBV do not
consider \emph{usefulness}. 
Notably, except for~\cite{KesnerViso22}, which relies entirely
on the (verbose) \emph{persistent/consuming} paradigm, this is the only 
\emph{simple} and \emph{concise} quantitative type system for open \CBV that is able to count 
variable and \emph{abstraction substitution} steps \emph{exactly}.

Several lines of work stem from this paper. 
One first objective is to understand the flexibility of system 
$\typesystem$, particularly its potential to capture other \CBV optimisations, such as those in~\cite{AccattoliG19}.
Additionally, it would be interesting to capture recent alternative specifications of usefulness; for instance, the one in~\cite{Wu23}, which focuses exclusively on \emph{direct} useful reduction steps. 
Notably, erasing reduction steps in~\cite{Wu23}
pose a non-trivial difficulty for non-idempotent intersection type systems designed for usefulness. 
Another challenging direction is extending the methodology to define a \emph{strong} useful \CBV 
calculus (allowing reductions to take place within abstractions), 
which could be captured by a quantitative type system based on system $\typesystem$.

\bibliographystyle{alphaurl}
\bibliography{biblio}

\appendix
\section{Technical Proofs}
\label{appendix}
\subsection{Auxiliary Material for Proving~\autoref{thm:diamond-property} } 

The following definition will come in handy in the proof of~\autoref{thm:diamond-property}.
Let $\aset$ be an abstraction frame.
We inductively define the \defn{expansion} of $\aset$ under $\sctx$,
written $\expansion\aset\sctx$, as follows:
\begin{align*}
    \expansion\aset\ctxhole 
  & \eqdef \aset
\\
    \expansion\aset{\sctxtwo\esub\var\tm} 
  & \eqdef 
    \begin{cases*} 
      \expansion\aset\sctxtwo \cup \set\var & if $\tm \in \HAbs\aset$ 
    \\
      \expansion\aset\sctxtwo               & otherwise
    \end{cases*}
\end{align*}

Analogously, let $\sset$ be a structure frame.
We inductively define the \defn{expansion} of $\sset$ under $\sctx$,
written $\expansion\sset\sctx$, as follows:
\begin{align*}
    \expansion\sset\ctxhole 
  & \eqdef \sset
\\
    \expansion\sset{\sctxtwo\esub\var\tm} 
  & \eqdef 
    \begin{cases*} 
      \expansion\sset\sctxtwo \cup \set\var & if $\tm \in \Struct\sset$ 
    \\
      \expansion\sset\sctxtwo               & otherwise
    \end{cases*}
\end{align*}

Note that $\aset \subseteq \expansion\aset\sctx$ and $\sset \subseteq\expansion\sset\sctx$.

\begin{lemC}[{\cite[Lemma 216]{Milicich25}}]
\label{lem:ucbv_weakening_of_reduction_sets}
If $\tm \tov\rulename\aset\sset\appflag \tm'$,
then for all sets $\aset'$ and $\sset'$
such that $\aset \subseteq \aset'$ and $\sset \subseteq \sset'$
it holds that $\tm \tov\rulename{\aset'}{\sset'}\appflag \tm'$.
\end{lemC}

\diamondproperty*

\begin{proof}
By structural induction on $\tm$.
Case $\tm = \lam\var\tmtwo$ is impossible, as there are no rules for 
reducing abstractions.
\begin{itemize}
\item $\tm = \var$.
  This case does not apply since we would have, on the one hand,
  $\var \tov{\rulesub\var{\val_1}}{\aset' \cup \set\var}\sset\app \val_1 = \tm_1$ 
  and on the other hand
  $\var \tov{\rulesub\var{\val_2}}{\aset' \cup \set\var}\sset\app \val_2 = \tm_2$, 
  with $\val_1 = \val_2$ by the hypothesis,
  which contradicts the hypothesis stating $\tm_1 \neq \tm_2$.
\item $\tm = \tmtwo \, \tmthree$.
  Since $\inv\aset\sset{\tmtwo \, \tmthree}$, then 
  $\inv\aset\sset\tmtwo$ and $\inv\aset\sset\tmthree$ hold.
  Applications reduce by rules \ruleUDb, \ruleUAppL, and \ruleUAppR,
  so we consider the following subcases:
  \begin{itemize}
  \item \ruleUDb-\ruleUDb.
    This case does not apply since it ends up being that 
    $\tm_1 = \tm_2$, which contradicts the hypothesis.
  \item \ruleUDb-\ruleUAppL.
    Then,
    $\tm = (\lam\var{\tmtwo'})\sctx \, \tmthree 
    \tov\dbsym\aset\sset\appflag 
    \tmtwo'\esub\var\tmthree\sctx = \tm_1$, 
    where $\tmtwo = (\lam\var{\tmtwo'})\sctx$ and $\rulename_1 = \dbsym$.

    On the other hand,
    $\tm = (\lam\var{\tmtwo'})\sctx \, \tmthree 
    \tov{\rulename_2}\aset\sset\appflag \tmtwo_2 \, \tmthree = \tm_2$
    is derived from 
    $(\lam\var{\tmtwo'})\sctx \tov{\rulename_2}\aset\sset\app \tmtwo_2$,
    where it can be easily noted that $\tmtwo_2$ must be of the form $(\lam\var{\tmtwo'})\sctxtwo$.

    Then, $\tm_2 = (\lam\var{\tmtwo'})\sctxtwo \, \tmthree
    \tov\dbsym\aset\sset\appflag
    \tmtwo'\esub\var\tmthree\sctxtwo = \tm'$ by rule \ruleUDb, and
    $\tm_1 = \tmtwo'\esub\var\tmthree\sctx 
    \tov{\rulename_2}\aset\sset\appflag 
    \tmtwo'\esub\var\tmthree\sctxtwo = \tm'$, by rule any of the explicit substitution congruence rules.

    The following diagram summarises the proof:
    \begin{mathpar}
      \xymatrixcolsep{5pc}
      \xymatrixrowsep{3pc}
      \xymatrix{
          \tm = (\lam\var{\tmtwo'})\sctx \, \tmthree 
            \arUr\dbsym\aset\sset\appflag 
            \arUd{\rulename_2}\aset\sset\appflag
        & \tmtwo'\esub\var\tmthree\sctx = \tm_1
            \arsdUd{\rulename_2}\aset\sset\appflag
      \\
          \tm_2 = (\lam\var{\tmtwo'})\sctxtwo \, \tmthree
            \arsdUr\dbsym\aset\sset\appflag
        & \tmtwo'\esub\var\tmthree\sctxtwo = \tm'
      }
    \end{mathpar}
  \item \ruleUDb-\ruleUAppR.
    Then,
    $\tm = (\lam\var{\tmtwo'})\sctx \, \tmthree 
    \tov\dbsym\aset\sset\appflag \tmtwo'\esub\var\tmthree\sctx = \tm_1$,
    where
    $\tmtwo = (\lam\var{\tmtwo'})\sctx$ and $\rulename_1 = \dbsym$.

    On the other hand,
    $\tm = (\lam\var{\tmtwo'})\sctx \, \tmthree 
    \tov{\rulename_2}\aset\sset\appflag 
    (\lam\var{\tmtwo'})\sctx \, \tmthree_2 = \tm_2$ is derived from
    $(\lam\var{\tmtwo'})\sctx \in \Struct\sset$ and
    $\tmthree \tov{\rulename_2}\aset\sset\nonapp \tmthree_2$.
    
    This case is impossible since there is no rule from~\autoref{fig:HAbs_Struct} to conclude
    $(\lam\var{\tmtwo'})\sctx \in \Struct\sset$.
  \item \ruleUAppL-\ruleUAppR.
    Then,
    $\tm = \tmtwo \, \tmthree \tov{\rulename_1}\aset\sset\appflag \tmtwo_1 \, \tmthree = \tm_1$,
    derived from 
    (1) $\tmtwo \tov{\rulename_1}\aset\sset\app \tmtwo_1$.
    
    On the other hand,
    $\tm = \tmtwo \, \tmthree \tov{\rulename_2}\aset\sset\appflag \tmtwo \, \tmthree_2 = \tm_2$,
    derived from (2) $\tmtwo \in \Struct\sset$ and
    (3) $\tmthree \tov{\rulename_2}\aset\sset\nonapp \tmthree_2$.
    
    We can apply rule \ruleUAppL with (1) as premise, yielding
    $\tm_2 = \tmtwo \, \tmthree_2 \tov{\rulename_1}\aset\sset\appflag
    \tmtwo_1 \, \tmthree_2 = \tm'$.
    Given that (1) and (2) hold, then $\tmtwo_1 \in \Struct\sset$ since 
    the membership in hereditary abstractions and structures sets is closed under reduction (\cf~\cite[Lemma 223]{Milicich25}).
    With this result and (3), we can apply rule \ruleUAppR, yielding
    $\tm_1 = \tmtwo_1 \, \tmthree \tov{\rulename_2}\aset\sset\appflag
    \tmtwo_1 \, \tmthree_2 = \tm'$.

    The following diagram summarises the proof:
    \begin{mathpar}
      \xymatrixcolsep{5pc}
      \xymatrixrowsep{3pc}
      \xymatrix{
          \tm = \tmtwo \, \tmthree 
            \arUr{\rulename_1}\aset\sset\appflag
            \arUd{\rulename_2}\aset\sset\appflag
        & \tmtwo_1 \, \tmthree = \tm_1
            \arsdUd{\rulename_2}\aset\sset\appflag
      \\
          \tm_2 = \tmtwo \, \tmthree_2
            \arsdUr{\rulename_1}\aset\sset\appflag
        & \tmtwo_1 \, \tmthree_2 = \tm'
      }
    \end{mathpar}
  \item \ruleUAppL-\ruleUAppL.
    Then,
    $\tm = \tmtwo \, \tmthree \tov{\rulename_1}\aset\sset\appflag \tmtwo_1 \, \tmthree = \tm_1$,
    derived from $\tmtwo \tov{\rulename_1}\aset\sset\app \tmtwo_1$.

    On the other hand,
    $\tm = \tmtwo \, \tmthree \tov{\rulename_2}\aset\sset\appflag \tmtwo_2 \, \tmthree = \tm_2$,
    derived from 
    $\tmtwo \tov{\rulename_2}\aset\sset\app \tmtwo_2$, 
    where $\tmtwo_1 \neq \tmtwo_2$
    since $\tmtwo_1 \, \tmthree \neq \tmtwo_2 \, \tmthree$ by hypothesis.

    The \ih on $\tmtwo$ yields $\tmtwo'$ such that
    $\tmtwo_1 \tov{\rulename_2}\aset\sset\app \tmtwo'$
    and
    $\tmtwo_2 \tov{\rulename_1}\aset\sset\app \tmtwo'$.
    Applying rule \ruleUAppL to reduce both $\tmtwo_1 \, \tmthree$ and
    $\tmtwo_2 \, \tmthree$, we obtain
    $\tm_1 = \tmtwo_1 \, \tmthree \tov{\rulename_2}\aset\sset\appflag \tmtwo' \, \tmthree = \tm'$
    and
    $\tm_2 = \tmtwo_2 \, \tmthree \tov{\rulename_1}\aset\sset\appflag \tmtwo' \, \tmthree = \tm'$,
    respectively.

    The following diagram summarises the proof:
    \begin{mathpar}
      \xymatrixcolsep{5pc}
      \xymatrixrowsep{3pc}
      \xymatrix{
          \tm = \tmtwo \, \tmthree 
            \arUr{\rulename_1}\aset\sset\appflag 
            \arUd{\rulename_2}\aset\sset\appflag
        & \tmtwo_1 \, \tmthree = \tm_1
            \arsdUd{\rulename_2}\aset\sset\appflag
      \\
          \tm_2 = \tmtwo_2 \, \tmthree
            \arsdUr{\rulename_1}\aset\sset\appflag
        & \tmtwo' \, \tmthree = \tm'
      }
    \end{mathpar}
  \item \ruleUAppR-\ruleUAppR.
    Analogous to the previous case.
  \end{itemize}
\item $\tm = \tmtwo\esub\var\tmthree$.
  The invariant $\inv\aset\sset{\tmtwo\esub\var\tmthree}$ implies 
  $\inv{\aset \cup \set\var}\sset\tmtwo$, 
  $\inv\aset{\sset \cup \set\var}\tmtwo$, and $\inv\aset\sset\tmthree$.
  Closures reduce by rules \ruleULsv, \ruleUEsR, \ruleUEsLAbs, and 
  \ruleUEsLStruct, so we consider the following cases:
  \begin{itemize}
  \item \ruleULsv-\ruleULsv.
    Then
    \begin{mathpar}
      \inferrule* [right = \ruleULsv]
        {\tmtwo \tov{\rulesub\var\val}{\aset \cup \set\var}\sset\appflag \tmtwo_1
          \\ \var \notin \aset \cup \sset
          \\ \val\sctx \in \HAbs\aset}
        {\tm = \tmtwo\esub\var{\val\sctx} \tov\lsvsym\aset\sset\appflag
        \tmtwo_1\esub\var\val\sctx = \tm_1}
    \end{mathpar}
    where $\tmthree = \val\sctx$ and $\rulename_1 = \lsvsym$.
    
    On the other hand,
    \begin{mathpar}
      \inferrule* [right = \ruleULsv]
        {\tmtwo \tov{\rulesub\var\val}{\aset \cup \set\var}\sset\appflag \tmtwo_2
          \\ \var \notin \aset \cup \sset
          \\ \val\sctx \in \HAbs\aset}
        {\tm = \tmtwo\esub\var{\val\sctx} \tov\lsvsym\aset\sset\appflag
        \tmtwo_2\esub\var\val\sctx = \tm_2}
    \end{mathpar}
    where $\rulename_2 = \ruleULsv$ and $\tmtwo_1 \neq \tmtwo_2$,
    since $\tmtwo_1\esub\var{\val\sctx} \neq \tmtwo_2\esub\var{\val\sctx}$.

    The \ih on $\tmtwo$ yields a term $\tmtwo'$ such that
    $\tmtwo_1 \tov{\rulesub\var\val}{\aset \cup \set\var}\sset\appflag \tmtwo'$ and
    $\tmtwo_2 \tov{\rulesub\var\val}{\aset \cup \set\var}\sset\appflag \tmtwo'$.
    Moreover, note that $\expansion{(\aset \cup \set\var)}\sctx = \expansion\aset\sctx \cup \set\var$,
    and thus
    $\tmtwo_1 
    \tov{\rulesub\var\val}{\expansion\aset\sctx \cup \set\var}{\expansion\sset\sctx}\appflag 
    \tmtwo'$ 
    and
    $\tmtwo_2 
    \tov{\rulesub\var\val}{\expansion\aset\sctx \cup \set\var}{\expansion\sset\sctx}\appflag 
    \tmtwo'$ hold
    by~\autoref{lem:ucbv_weakening_of_reduction_sets}.
    Furthermore, the hypothesis $\val\sctx \in \HAbs\aset$ implies
    $\val \in \HAbs{\expansion\aset\sctx}$ (\cite[Lemma 221]{Milicich25}).
    Since 
    $\var \notin \aset \cup \sset$ by premise in rule \ruleULsv, 
    and $\var \notin \dom\sctx$
    by $\alpha$-conversion, 
    then $\var \notin \expansion\aset\sctx \cup \expansion\sset\sctx$.

    We can then apply rule \ruleULsv on both 
    $\tmtwo_1\esub\var\val$ and $\tmtwo_2\esub\var\val$, yielding
    $\tmtwo_1\esub\var\val 
    \tov\lsvsym{\expansion\aset\sctx}{\expansion\sset\sctx}\appflag 
    \tmtwo'\esub\var\val$ and
    $\tmtwo_2\esub\var\val 
    \tov\lsvsym{\expansion\aset\sctx}{\expansion\sset\sctx}\appflag 
      \tmtwo'\esub\var\val$, respectively.
    We thus conclude with
    $\tm_1 = \tmtwo_1\esub\var\val\sctx 
    \tov\lsvsym\aset\sset\appflag 
    \tmtwo'\esub\var\val\sctx = \tm'$ and
    $\tm_2 = \tmtwo_2\esub\var\val\sctx 
    \tov\lsvsym\aset\sset\appflag 
    \tmtwo'\esub\var\val\sctx = \tm'$
    by applying the rules \ruleUEsLAbs or \ruleUEsLStruct successively, as appropriate.

    The following diagram summarises the proof:
    \begin{mathpar}
      \xymatrixcolsep{5pc}
      \xymatrixrowsep{3pc}
      \xymatrix{
        \tm = \tmtwo\esub\var{\val\sctx} 
          \arUr\lsvsym\aset\sset\appflag 
          \arUd\lsvsym\aset\sset\appflag
      & \tmtwo_1\esub\var\val\sctx = \tm_1
          \arsdUd\lsvsym\aset\sset\appflag
      \\
        \tm_2 = \tmtwo_2\esub\var\val\sctx
          \arsdUr\lsvsym\aset\sset\appflag
      & \tmtwo'\esub\var\val\sctx = \tm'
      }
    \end{mathpar}
  \item \ruleULsv-\ruleUEsR. 
    Then
    \begin{mathpar}
      \inferrule* [right = \ruleULsv]
        {(1)\ \tmtwo \tov{\rulesub\var\val}{\aset \cup \set\var}\sset\appflag \tmtwo_1
          \\ (2)\ \var \notin \aset \cup \sset
          \\ (3)\ \val\sctx \in \HAbs\aset}
        {\tm = \tmtwo\esub\var{\val\sctx} \tov\lsvsym\aset\sset\appflag
        \tmtwo_1\esub\var\val\sctx = \tm_1}
    \end{mathpar}
    where $\tmthree = \val\sctx$ and $\rulename_1 = \lsvsym$.

    On the other hand,
    $\tm = \tmtwo\esub\var{\val\sctx} 
    \tov{\rulename_2}\aset\sset\appflag 
    \tmtwo\esub\var{\tmthree_2} = \tm_2$, derived from
    $(4)\ \val\sctx \tov{\rulename_2}\aset\sset\nonapp \tmthree_2$.
    Moreover, it can be easily noted that $\tmthree_2 = \val\sctxtwo$.

    We can then apply accordingly rule \ruleUEsR, \ruleUEsLAbs, or 
    \ruleUEsLStruct to $(\rulename_2, \aset, \sset, \appflag)$-reduce 
    $\tm_1 = \tmtwo_1\esub\var\val\sctx$ to the term 
    $\tmtwo_1\esub\var\val\sctxtwo = \tm'$.

    Since (3) and (4), then (5) $\val\sctxtwo \in \HAbs\aset$ since 
    the membership in hereditary abstractions and structures sets is closed under reduction.
    We can apply rule \ruleULsv with this result and (1) and (2) as
    premises, yielding
    $\tm_2 = \tmtwo\esub\var{\val\sctxtwo} 
    \tov\lsvsym\aset\sset\appflag \tmtwo_1\esub\var\val\sctxtwo = \tm'$.

    The following diagram summarises the proof:
    \begin{mathpar}
      \xymatrixcolsep{5pc}
      \xymatrixrowsep{3pc}
      \xymatrix{
        \tm = \tmtwo\esub\var{\val\sctx}
          \arUr\lsvsym\aset\sset\appflag 
          \arUd{\rulename_2}\aset\sset\appflag
      & \tmtwo_1\esub\var\val\sctx = \tm_1
          \arsdUd{\rulename_2}\aset\sset\appflag
      \\
        \tm_2 = \tmtwo\esub\var{\val\sctxtwo}
          \arsdUr\lsvsym\aset\sset\appflag
      & \tmtwo_1\esub\var\val\sctxtwo = \tm'
      }
    \end{mathpar}
  \item \ruleULsv-\ruleUEsLAbs.
    Then
    \begin{mathpar}
      \inferrule* [right = \ruleULsv]
        {\tmtwo \tov{\rulesub\var\val}{\aset \cup \set\var}\sset\appflag \tmtwo_1
          \\ (1)\ \var \notin \aset \cup \sset
          \\ (2)\ \val\sctx \in \HAbs\aset}
        {\tm = \tmtwo\esub\var{\val\sctx} \tov\lsvsym\aset\sset\appflag
        \tmtwo_1\esub\var\val\sctx = \tm_1}
    \end{mathpar}
    where $\tmthree = \val\sctx$ and $\rulename_1 = \lsvsym$.

    On the other hand,
    \begin{mathpar}
      \inferrule* [right = \ruleUEsLAbs]
        {\tmtwo \tov{\rulename_2}{\aset \cup \set\var}\sset\appflag \tmtwo_2
          \HS \val\sctx \in \HAbs\aset
          \HS \var \notin \aset \cup \sset
          \HS (3)\ \var \notin \fv{\rulename_2}}
        {\tm = \tmtwo\esub\var{\val\sctx} \tov{\rulename_2}\aset\sset\appflag
         \tmtwo_2\esub\var{\val\sctx} = \tm_2}
    \end{mathpar}

    The \ih on $\tmtwo$ yields $\tmtwo'$ such that
    $\tmtwo_1 
    \tov{\rulename_2}{\aset \cup \set\var}\sset\appflag 
    \tmtwo'$ and 
    (4) $\tmtwo_2 \tov{\rulesub\var\val}{\aset \cup \set\var}\sset\appflag 
    \tmtwo'$.
    Moreover,
    (5) $\tmtwo_1 \tov{\rulename_2}{\expansion\aset\sctx \cup \set\var}{\expansion\sset\sctx}\appflag 
    \tmtwo'$ by~\autoref{lem:ucbv_weakening_of_reduction_sets}.
    The hypothesis (2) implies 
    (6) $\val \in \HAbs{\expansion\aset\sctx}$ (\cite[Lemma 221]{Milicich25}).

    Given (1) and the fact that $\var \notin \dom\sctx$ by $\alpha$-conversion, then 
    (7) $\var \notin \expansion\aset\sctx \cup \expansion\sset\sctx$.
    Therefore, (8)
    $\tmtwo_1\esub\var\val 
    \tov{\rulename_2}{\expansion\aset\sctx}{\expansion\sset\sctx}\appflag 
    \tmtwo'\esub\var\val$ by rule \ruleUEsLAbs,
    with (5), (6), (7), and (3) as premises.
    And we conclude with
    $\tm_1 = \tmtwo_1\esub\var\val\sctx 
    \tov{\rulename_2}\aset\sset\appflag 
    \tmtwo'\esub\var\val\sctx = \tm'$ by successively applying rule 
    \ruleUEsLAbs or rule \ruleUEsLStruct accordingly.
    
    On the other hand, by rule \ruleULsv with (4), (1), and (2) as 
    premises, we obtain
    $\tm_2 = \tmtwo_2\esub\var{\val\sctx} 
    \tov\lsvsym\aset\sset\appflag 
    \tmtwo'\esub\var\val\sctx = \tm'$.
    
    The following diagram summarises the proof:
    \begin{mathpar}
      \xymatrixcolsep{5pc}
      \xymatrixrowsep{3pc}
      \xymatrix{
        \tm = \tmtwo\esub\var{\val\sctx} 
          \arUr\lsvsym\aset\sset\appflag 
          \arUd{\rulename_2}\aset\sset\appflag
      & \tmtwo_1\esub\var\val\sctx = \tm_1
          \arsdUd{\rulename_2}\aset\sset\appflag
      \\
        \tm_2 = \tmtwo_2\esub\var{\val\sctx}
          \arsdUr\lsvsym\aset\sset\appflag
      & \tmtwo'\esub\var\val\sctx = \tm'
      }
    \end{mathpar}
  \item \ruleULsv-\ruleUEsLStruct.
    Then
    \begin{mathpar}
      \inferrule* [right = \ruleULsv]
        {\tmtwo \tov{\rulesub\var\val}{\aset \cup \set\var}\sset\appflag \tmtwo_1
          \\ \var \notin \aset \cup \sset
          \\ \val\sctx \in \HAbs\aset}
        {\tm = \tmtwo\esub\var{\val\sctx} \tov\lsvsym\aset\sset\appflag
        \tmtwo_1\esub\var{\val}\sctx = \tm_1}
    \end{mathpar}
    where $\tmthree = \val\sctx$ and $\rulename_1 = \lsvsym$.

    On the other hand,
    $\tm = \tmtwo\esub\var{\val\sctx} 
    \tov{\rulename_2}\aset\sset\appflag \tmtwo_2\esub\var{\val\sctx} = \tm_2$,
    derived from 
    $\tmtwo \tov{\rulename_2}\aset{\sset \cup \set\var}\appflag \tmtwo_2$,
    $\val\sctx \in \Struct\sset$, $\var \notin \aset \cup \sset$, and 
    $\var \notin \fv{\rulename_2}$.
    Since $\inv\aset\sset{\val\sctx}$, then in particular 
    $\aset \disj \sset$.
    Hence,
    $\val\sctx \in \HAbs\aset$ and $\val\sctx \in \Struct\sset$,
    and at the same time it should be the case that
    $\val\sctx \notin \HAbs\aset$ or $\val\sctx \notin \Struct\sset$.
    Therefore, we reach a contradiction, so this case is not possible.
  \item \ruleUEsR-\ruleUEsR.
    Then
    $\tm = \tmtwo\esub\var\tmthree \tov{\rulename_1}\aset\sset\appflag
     \tmtwo\esub\var{\tmthree_1} = \tm_1$,
    derived from $\tmthree \tov{\rulename_1}\aset\sset\nonapp \tmthree_1$.

    On the other hand,
    $\tm = \tmtwo\esub\var\tmthree \tov{\rulename_2}\aset\sset\appflag
    \tmtwo\esub\var{\tmthree_2} = \tm_2$, derived from
    $\tmthree \tov{\rulename_2}\aset\sset\nonapp \tmthree_2$,
    where $\tmthree_1 \neq \tmthree_2$
    since $\tmtwo\esub\var{\tmthree_1} \neq \tmtwo\esub\var{\tmthree_2}$
    by hypothesis.

    The \ih on $\tmthree$ yields $\tmthree'$ such that
    $\tmthree_1 \tov{\rulename_2}\aset\sset\nonapp \tmthree'$
    and 
    $\tmthree_2 \tov{\rulename_1}\aset\sset\nonapp \tmthree'$.
    Applying rule \ruleUEsR on both $\tmtwo\esub\var{\tmthree_1}$ and
    $\tmtwo\esub\var{\tmthree_2}$, 
    we obtain $\tm_1 = \tmtwo\esub\var{\tmthree_1} 
    \tov{\rulename_2}\aset\sset\appflag 
    \tmtwo\esub\var{\tmthree'} = \tm'$ 
    and $\tm_2 = \tmtwo\esub\var{\tmthree_2} 
    \tov{\rulename_1}\aset\sset\appflag 
    \tmtwo\esub\var{\tmthree'} = \tm'$, respectively.

    The following diagram summarises the proof:
    \begin{mathpar}
      \xymatrixcolsep{5pc}
      \xymatrixrowsep{3pc}
      \xymatrix{
        \tm = \tmtwo\esub\var\tmthree
          \arUr{\rulename_1}\aset\sset\appflag 
          \arUd{\rulename_2}\aset\sset\appflag
      & \tmtwo\esub\var{\tmthree_1} = \tm_1
          \arsdUd{\rulename_2}\aset\sset\appflag
      \\
        \tm_2 = \tmtwo\esub\var{\tmthree_2}
          \arsdUr{\rulename_1}\aset\sset\appflag
      & \tmtwo\esub\var{\tmthree'} = \tm'
      }
    \end{mathpar}
  \item \ruleUEsLAbs-\ruleUEsLAbs.
    Analogous to the previous case.
  \item \ruleUEsLStruct-\ruleUEsLStruct.
    Analogous to the previous case.
  \item \ruleUEsR-\ruleUEsLAbs.
    Then, 
    $\tm = \tmtwo\esub\var\tmthree 
    \tov{\rulename_1}\aset\sset\appflag 
    \tmtwo\esub\var{\tmthree_1} = \tm_1$, which is derived from
    (1) $\tmthree \tov{\rulename_1}\aset\sset\nonapp \tmthree_1$.

    On the other hand, we have
    $\tm = \tmtwo\esub\var\tmthree \tov{\rulename_2}\aset\sset\appflag
    \tmtwo_2\esub\var\tmthree = \tm_2$, derived from
    (2) $\tmtwo \tov{\rulename_2}{\aset \cup \set\var}\sset\appflag \tmtwo_2$,
    (3) $\tmthree \in \HAbs\aset$, 
    (4) $\var \notin \aset \cup \sset$, and 
    (5) $\var \notin \fv{\rulename_2}$.
    We can apply rule \ruleUEsLAbs with (2), (3), (4), and (5) as 
    premises, yielding
    $\tm_1 = \tmtwo\esub\var{\tmthree_1} 
    \tov{\rulename_2}\aset\sset\appflag 
    \tmtwo_2\esub\var{\tmthree_1} = \tm'$.
    Then, we can apply rule \ruleUEsR with (1) as premise, yielding
    $\tm_2 = \tmtwo_2\esub\var\tmthree 
    \tov{\rulename_1}\aset\sset\appflag 
    \tmtwo_2\esub\var{\tmthree_1} = \tm'$.

    The following diagram summarises the proof:
    \begin{mathpar}
      \xymatrixcolsep{5pc}
      \xymatrixrowsep{3pc}
      \xymatrix{
        \tm = \tmtwo\esub\var\tmthree
          \arUr{\rulename_1}\aset\sset\appflag 
          \arUd{\rulename_2}\aset\sset\appflag
      & \tmtwo\esub\var{\tmthree_1} = \tm_1
          \arsdUd{\rulename_2}\aset\sset\appflag
      \\
        \tm_2 = \tmtwo_2\esub\var\tmthree
          \arsdUr{\rulename_1}\aset\sset\appflag
      & \tmtwo_2\esub\var{\tmthree_1} = \tm'
      }
    \end{mathpar}
  \item \ruleUEsR-\ruleUEsLStruct.
    Analogous to the previous case.
  \item \ruleUEsLAbs-\ruleUEsLStruct.
    Then, the two reduction steps are
    \begin{mathpar}
      \inferrule{
        \tmtwo \tov{\rulename_1}{\aset \cup \set\var}\sset\appflag \tmtwo_1
        \\ \tmthree \in \HAbs\aset
        \\ \var \notin \aset \cup \sset
        \\ \var \notin \fv{\rulename_1}
      }{
        \tm = \tmtwo\esub\var\tmthree \tov{\rulename_1}\aset\sset\appflag \tmtwo_1\esub\var\tmthree = \tm_1 
      }
    \end{mathpar}
    and
    \begin{mathpar}
      \inferrule{
        \tmtwo \tov{\rulename_2}\aset{\sset \cup \set\var}\appflag \tmtwo_2
        \\ \tmthree \in \Struct\sset
        \\ \var \notin \aset \cup \sset
        \\ \var \notin \fv{\rulename_2}
      }{
        \tm = \tmtwo\esub\var\tmthree \tov{\rulename_2}\aset\sset\appflag \tmtwo_2\esub\var\tmthree = \tm_2 
      }
    \end{mathpar}

    Since $\inv\aset\sset\tmthree$ then in particular $\aset \disj \sset$
    so it must be the case that $\tmthree \notin \HAbs\aset$ or $\tmthree \notin \Struct\sset$.
    However, we obtain as well that $\tmthree \in \HAbs\aset$ and 
    $\tmthree \in \Struct\sset$ by premises of rules \ruleUEsLAbs and
    \ruleUEsLStruct, respectively.
    Therefore, we reach a contradiction, so this case is not possible.
    \qedhere
  \end{itemize}
\end{itemize}
\end{proof}

\subsection{Auxiliary Material for Proving~\autoref{thm:quantitative_soundness_typing}} 

We now state two additional results needed to prove
the exact soundness of system $\typesystem$ with respect to \UOCBV evaluation.

\begin{lem}
\label{lem:systemU_multitype_implies_reducibility}
Let $\tm$ be a term, $\aset$ be an abstraction frame, and $\sset$ a 
structure frame such that $\inv\aset\sset\tm$ holds.
Suppose there exists a derivation $\deriv$ satisfying
$\derivable\deriv{\judg[\cm,\ce]{\tctx;\tctxtwo}\tm\tight}$,
such that $\tctx$ is tight and $\isAppr\aset\tctx$ holds.
Additionally, assume that $\dom\tctxtwo \neq \emptyset$ and 
$\dom\tctxtwo \subseteq \aset$, and that for all 
$\var \in \dom\tctxtwo$, the optional type $\tctxtwo(\var)$ is a 
non-empty arrow multi-type.
Then, $\tm \notin \NF\aset\sset\appflag$.
\end{lem}

\begin{proof}
By induction on $\deriv$. Cases \ruleTypVar and \ruleTypAbs do not apply, 
and we only show case \ruleTypES.
Then, $\deriv$ is of the form:
\begin{mathpar}
  \inferrule
    {\derivable{\deriv_\tmtwo}{\judg[\cm_1,\ce_1]{\tctx_1;\tctxtwo_1; \var : \optmtyp}\tmtwo\tight}
      \\ \optmtyp \mleq \mtyp
      \\ \derivable{\deriv_\tmthree}{\judg[\cm_2,\ce_2]{\tctx_2;\tctxtwo_2}\tmthree\mtyp}}
    {\judg
      [\cm_1 + \cm_2, \ce_1 + \ce_2]
      {(\tctx_1 + \tctx_2);(\tctxtwo_1 + \tctxtwo_2)}
      {\tmtwo\esub\var\tmthree}
      \tight}
\end{mathpar}
where $\tctx = \tctx_1 + \tctx_2$, 
$\tctxtwo = \tctxtwo_1 + \tctxtwo_2$, 
$\cm = \cm_1 + \cm_2$, $\ce = \ce_1 + \ce_2$, and
$\tm = \tmtwo\esub\var\tmthree$. 

Moreover, $\inv\aset\sset{\tmtwo\esub\var\tmthree}$ implies in 
particular $\inv\aset\sset\tmthree$.
We reason by contradiction, assuming 
$\tmtwo\esub\var\tmthree \in \NF\aset\sset\appflag$.
This assumption can be derived either by rule \ruleUNFEsAbs or by 
rule \ruleUNFEsStruct.
We may assume $\var \notin (\aset \cup \sset)$ by $\alpha$-conversion.
\begin{itemize}
\item \ruleUNFEsAbs.
  Then
  \begin{mathpar}
    \inferrule
      {\tmtwo \in \NF{\aset \cup \set\var}\sset\appflag
        \\ \tmthree \in \NF\aset\sset\nonapp
        \\ \tmthree \in \HAbs\aset}
      {\tmtwo\esub\var\tmthree \in \NF\aset\sset\appflag}
  \end{mathpar}
  By hypothesis, $\dom{\tctxtwo_1 + \tctxtwo_2} \neq \emptyset$,
  so $\dom{\tctxtwo_1} \neq \emptyset$ and $\dom{\tctxtwo_2} \neq \emptyset$.
  Also, $\mtyp$ can be tight or not.
  We have to analyse three possible cases:
  \begin{itemize}
  \item 
    If $\mtyp$ not tight, then $\optmtyp = \mtyp$ since 
    $\optmtyp \mleq \mtyp$.
    We have that $\tctx = \tctx_1 + \tctx_2$ implies that 
    $\tctx_1$ is tight and $\isAppr\aset{\tctx_1}$, as $\tctx$ is 
    tight and $\isAppr\aset\tctx$.
    Moreover, $\isAppr{\aset \cup \set\var}{\tctx_1}$ holds, as $\var \notin \dom{\tctx_1}$.
    Additionally, $\tctxtwo = \tctxtwo_1 + \tctxtwo_2$
    implies that $\dom{\tctxtwo_1; \var : \mtyp} \subseteq \aset \cup \set\var$,
    since $\dom\tctxtwo \subseteq \aset$ holds by the hypothesis.
    And by the hypothesis, it is straightforward to note that any element 
    in $\dom{\tctxtwo_1; \var : \mtyp}$ gets assigned a 
    non-empty arrow multi-type.
    We can then apply the \ih on $\deriv_\tmtwo$, yielding
    $\tmtwo \notin \NF{\aset \cup \set\var}\sset\appflag$,
    which leads to a contradiction with our assumption
    $\tmtwo \in \NF{\aset \cup \set\var}\sset\appflag$. 
    Therefore, $\tmtwo\esub\var\tmthree \notin \NF\aset\sset\appflag$.
  \item 
    If $\mtyp = \tight$ and $\tctxtwo_2 \neq \emptyctx$.
    We have that $\tctx = \tctx_1 + \tctx_2$ implies that 
    $\tctx_2$ is tight and $\isAppr\aset{\tctx_2}$, as $\tctx$ is 
    tight and $\isAppr\aset\tctx$.
    We can apply the \ih on $\tmthree$, yielding
    $\tmthree \notin \NF\aset\sset\appflag$, in particular 
    $\appflag$ can be $\nonapp$, which leads to a contradiction with 
    our assumption $\tmthree \in \NF\aset\sset\nonapp$.
    Therefore, $\tmtwo\esub\var\tmthree \notin \NF\aset\sset\appflag$.
  \item 
    If $\mtyp = \tight$ and $\tctxtwo_1 \neq \emptyctx$ and 
    $\tctxtwo_2 = \emptyctx$, then
    We have that $\inv\aset\sset{\tmtwo\esub\var\tmthree}$ implies 
    $\inv{\aset \cup \set\var}\sset\tmtwo$.
    Furthermore, $\isAppr{\aset \cup \set\var}{\tctx_1}$ holds by 
    following the same reasoning as in the first case analysis, and 
    moreover $\isAppr\aset{\tctx_2;\tctxtwo_2}$ holds, and it implies 
    that $\mtyp \neq \tightN$ by applying~\autoref{lem:systemU_types_of_hereditary_abstractions} 
    on $\deriv_\tmthree$.
    Therefore, we are in the case where $\optmtyp \neq \tightN$,
    thus $\isAppr{\aset \cup \set\var}{\tctx_1; \var : \optmtyp}$ holds.
    We can then apply the \ih on $\tmtwo$, yielding
    $\tmtwo \notin \NF{\aset \cup \set\var}\sset\appflag$, which 
    leads to a contradiction with our assumption
    $\tmtwo \in \NF{\aset \cup \set\var}\sset\appflag$. 
    Therefore, $\tmtwo\esub\var\tmthree \notin \NF\aset\sset\appflag$.
    \qedhere
  \end{itemize}
\end{itemize}
\end{proof}

The lemma above intuitively states that system $\typesystem$ stores
non-empty multi-type information on typing environments whenever the 
typed term reduces in \UOCBV.

The following property shows that the 
counters in a typing derivation of a \UOCBV-normal form are $(0,0)$:
\begin{prop}
\label{prop:systemU_ucbv_nf_counters_zero}
Let $\tm$ be a term, $\aset$ an abstraction frame, and $\sset$ a 
structure frame such that $\inv\aset\sset\tm$ and 
$\tm \in \NF\aset\sset\appflag$ hold.
Suppose moreover there exists a tight derivation $\deriv$ such that 
$\derivable\deriv{\judg[\cm,\ce]\tctx\tm\tight}$, where
$\isAppr\aset\tctx$ and additionally, if $\appflag = \app$, 
then $\tight = \tightN$.
Then, $(\cm,\ce) = (0,0)$.
\end{prop}

\begin{proof}
By induction on $\tm \in \NF\aset\sset\appflag$. 
We only show the cases \ruleUNFVar, \ruleUNFLam, and \ruleUNFEsAbs.
\begin{itemize}
\item \ruleUNFVar.
  Then 
  \begin{mathpar}
    \inferrule
      {\var \in \aset \Rightarrow \appflag = \nonapp}
      {\tm = \var \in \NF\aset\sset\appflag}
  \end{mathpar}

  And $\deriv$ can only be obtained from the \ruleTypVar axiom, so
  it is of the form
  $\judg[0,\numarrows\tight]{\var : \tight}\var\tight$,
  with $\numarrows\tight = 0$ by definition, so we are done, as the 
  counters are $(0,0)$.
\item \ruleUNFLam.
  Then, $\tm = \lam\var\tmtwo \in \NF\aset\sset\nonapp$.
  The derivation $\deriv$ can only be obtained by rule \ruleTypAbs
  and the only possible case for which this derivation is 
  tight is when the resulting type of the judgement in the conclusion 
  is $\emset$, which holds with no premises: 
  $\judg[0,0]\emptyctx{\lam\var\tmtwo}\emset$.
  So we are done as the counters are $(0, 0)$.
\item \ruleUNFEsAbs.
  Then
  \begin{mathpar}
    \inferrule
      {\tmtwo \in \NF{\aset \cup \set\var}\sset\appflag 
        \\ \tmthree \in \NF\aset\sset\nonapp
        \\ \tmthree \in \HAbs\aset}
      {\tm = \tmtwo\esub\var\tmthree \in \NF\aset\sset\appflag}
  \end{mathpar}

  The derivation $\deriv$ can only be obtained by rule \ruleTypES, 
  so it is of the form:
  \begin{mathpar}
    \inferrule
      {\derivable{\deriv_\tmtwo}{\judg[\cm_1,\ce_1]{\tctx_1; \var : \optmtyp}\tmtwo\tight}
        \\ \optmtyp \mleq \mtyp
        \\ \derivable{\deriv_\tmthree}{\judg[\cm_2,\ce_2]{\tctx_2}\tmthree\mtyp}}
      {\judg[\cm_1 + \cm_2,\ce_1 + \ce_2]{\tctx_1 + \tctx_2}{\tmtwo\esub\var\tmthree}\tight}
  \end{mathpar}
  where $\tctx = \tctx_1 + \tctx_2$, $\cm = \cm_1 + \cm_2$, 
  and $\ce = \ce_1 + \ce_2$.
  
  By $\alpha$-conversion, we may assume $\var \notin (\aset \cup \sset)$.
  We analyse two cases, depending on whether $\mtyp$ is tight or not:
  \begin{itemize}
  \item
    If $\mtyp$ is tight, we first show that $\deriv_\tmtwo$ is a tight 
    derivation, starting by showing that its conclusion is a tight 
    judgement.
    The type is tight because of the hypothesis.
    The typing environment $\tctx_1$ is tight because $\tctx$ is tight 
    and $\tctx = \tctx_1 + \tctx_2$, and moreover $\optmtyp$
    is tight because $\mtyp$ is tight and $\optmtyp \mleq \mtyp$.
    Therefore, $\deriv_\tmtwo$ is tight.
    And $\isAppr\aset\tctx$ implies $\isAppr\aset{\tctx_1}$ because 
    $\dom{\tctx_1} \subseteq \dom\tctx$, and in particular 
    $\var \notin \dom{\tctx_1}$ by $\alpha$-conversion, so it is easy to conclude that
    $\isAppr{\aset \cup \set\var}{\tctx_1}$. 
    On the other hand, $\tmthree \in \HAbs\aset$ by hypothesis, so
    the type in the conclusion of $\deriv_\tmthree$ cannot be $\tightN$
    by~\autoref{lem:systemU_types_of_hereditary_abstractions}, so 
    $\optmtyp \neq \tightN$ is tight because $\optmtyp \mleq \mtyp$.
    Therefore, $\isAppr{\aset \cup \set\var}{\var : \optmtyp}$, 
    so that $\isAppr{\aset \cup \set\var}{\tctx_1; \var : \optmtyp}$
    holds.
    Moreover, if $\appflag = \app$ in 
    $\tmtwo \in \NF{\aset \cup \set\var}\sset\appflag$, then 
    $\tight = \tightN$ by hypothesis.
    We can then apply the \ih on 
    $\tmtwo \in \NF{\aset \cup \set\var}\sset\appflag$, yielding
    $(\cm_1,\ce_1) = (0,0)$.

    On the other hand, we show that $\deriv_\tmthree$ is tight.
    We show first that its conclusion is a tight judgement.
    The type $\mtyp$ is tight.
    The typing environment $\tctx_2$ is tight because $\tctx$ is tight 
    and $\tctx = \tctx_1 + \tctx_2$.
    And $\isAppr\aset\tctx$ implies $\isAppr\aset{\tctx_2}$ because 
    $\dom{\tctx_2} \subseteq \dom\tctx$. 
    Additionally, $\appflag = \nonapp$ in 
    $\tmthree \in \NF\aset\sset\nonapp$, so we do not look at the
    specific shape of $\mtyp$. 
    We can then apply the \ih on 
    $\tmthree \in \NF\aset\sset\nonapp$, yielding
    $(\cm_2,\ce_2) = (0,0)$.
    Then we conclude that 
    $(\cm,\ce) = (\cm_1 + \cm_2, \ce_1 + \ce_2) = (0, 0)$.
  \item 
    If $\mtyp$ is not tight, then $\optmtyp$ cannot be $\none$. 
    Moreover, $\optmtyp = \mtyp$ and it should be the case that 
    $\mtyp$ is a non-empty arrow multi-type, since otherwise $\mtyp$ 
    would be tight.
    Therefore, $\deriv_\tmtwo$ is not a tight derivation.

    Moreover, $\isAppr{\aset \cup \set\var}{\tctx_1}$ holds following 
    the same reasoning we did in the previous case.
    In particular, $\tctx_1$ is a tight typing environment because 
    $\tctx$ is and $\tctx = \tctx_1 + \tctx_2$.

    Then, $\tmtwo \notin \NF{\aset \cup \set\var}\sset\appflag$ by~\autoref{lem:systemU_multitype_implies_reducibility}, which 
    contradicts the hypothesis.
    Hence, this case is impossible.
    \qedhere
  \end{itemize}
\end{itemize}
\end{proof}

\end{document}